\documentclass[11pt]{article}
\usepackage[letterpaper,margin=1in]{geometry}
\usepackage{graphicx} \usepackage{xcolor}
\usepackage{hyperref}
\usepackage{amsmath,amsthm,amssymb,amsfonts,thmtools,thm-restate}
\usepackage[ruled,linesnumbered,noend]{algorithm2e}
\usepackage{todonotes}
\usepackage{enumerate}
\usepackage{tikz}
\usetikzlibrary{arrows.meta,calc,positioning}

\hypersetup{colorlinks,allcolors=blue}
\usepackage{cleveref}

\allowdisplaybreaks

\title{Tight Bounds for $(1+\epsilon)$-Approximation for Minimum Spanning Trees in Doubling Metrics}
\author{Jan H\"ockendorff\thanks{Email: \texttt{hoeckendorff@cs.uni-koeln.de}.
Funded by the Deutsche Forschungsgemeinschaft (DFG, German Research Foundation) - Project Number 459420781.} \\
University of Cologne
\and 
Felix Hommelsheim\thanks{Email: \texttt{hommelsheim@cs.uni-koeln.de}.
Funded by the Deutsche Forschungsgemeinschaft (DFG, German Research Foundation) - Project Number 573939419.} \\
University of Cologne
\and  
Christian Sohler\thanks{Email: \texttt{csohler@uni-koeln.de}.} \\
University of Cologne 
\and 
Di Yue\thanks{Email: \texttt{diyue@cs.toronto.edu}.} \\
University of Toronto}

\date{}

\theoremstyle{plain}
\newtheorem{theorem}{Theorem}[section]
\newtheorem{lemma}[theorem]{Lemma}

\newtheorem{observation}[theorem]{Observation}

\theoremstyle{definition}
\newtheorem{definition}[theorem]{Definition}

\theoremstyle{remark}

\newtheorem*{remark*}{Remark}

\makeatletter
\AtBeginDocument{\@ifundefined{theHtheorem}{}{}\@ifundefined{theHlemma}{}{}\@ifundefined{theHclaim}{}{}\@ifundefined{theHcorollary}{}{}\@ifundefined{theHquestion}{}{}\@ifundefined{theHproposition}{}{}\@ifundefined{theHfact}{}{}\@ifundefined{theHobservation}{}{}\@ifundefined{theHproperty}{}{}\@ifundefined{theHdefinition}{}{}\@ifundefined{theHexample}{}{}\@ifundefined{theHproblem}{}{}\@ifundefined{theHremark}{}{}}
\makeatother

\DeclareMathOperator*{\E}{\mathbb{E}}
\DeclareMathOperator{\ddim}{\operatorname{ddim}}
\DeclareMathOperator{\dist}{\mathbf{d}}
\DeclareMathOperator{\diam}{\operatorname{diam}}

\DeclareMathOperator*{\argmin}{argmin}
\DeclareMathOperator*{\argmax}{argmax}

\DeclareMathOperator{\wdeg}{wdeg}
\DeclareMathOperator{\lev}{lev}

\DeclareMathOperator{\MST}{MST}

\newcommand{\err}{\mathrm{err}}

\newcommand{\R}{\mathbb{R}}
\newcommand{\Z}{\mathbb{Z}}

\newcommand{\WSPC}{\mathrm{WSPC}}

\newcommand{\nets}{\mathcal{N}}
\newcommand{\child}{C}
\newcommand{\desc}{D}
\newcommand{\wspcball}[1]{B(#1, 2^{\lev(#1)})}
\newcommand{\gad}{{\mathrm{gad}}}

\newcommand{\ineqref}[2]{\mathrel{\overset{\textnormal{\tiny(#2)}}{\leq}}}

\begin{document}

\maketitle

\begin{abstract}
The minimum spanning tree (MST) problem is one of the most basic optimization problems on metric spaces and graphs.
We study the problem of computing a $(1+\epsilon)$-approximation to the MST of an $n$-point metric space of doubling dimension $\ddim$. In doubling metrics, previous deterministic algorithms incur a running time with dependence $\epsilon^{-O(\ddim)}$.

We give a deterministic algorithm that computes a $(1+\epsilon)$-approximate MST in time
\[
O_{\ddim}\left(n\log n+
\min\left\{n^2,\frac{n}{\epsilon}\log(1/\epsilon)\right\}\right)
\]
in the standard real RAM model with a unit-cost exact-distance oracle. For bounded doubling dimension, this improves the previous dependence on $\epsilon$ to essentially linear in $\epsilon^{-1}$. We complement this algorithm with a lower bound of
\[
\Omega \!\left(n\log n+
\min\left\{n^2,\frac{n}{\epsilon}\log(1/\epsilon)\right\}\right)
\]
on running time in the same model, already for metrics of doubling dimension at most seven, showing that our running time is essentially tight.

For points in $\mathbb R^d$ with $d = O(1)$, exploiting Euclidean geometry allows us to further improve the dependence on $\epsilon$. We obtain a $(1+\epsilon)$-approximate MST in expected time
\[
O_d\left( n \log n + \min\left\{n^2, \frac{n}{\epsilon^{1 - \beta_d}} \log^{1-\beta_d}(1/\epsilon) 
\right\} 
\right)
\]
for some $\beta_d=\Theta(1/d)$, improving the accuracy-dependent term $n\epsilon^{-2}\log^2(1/\epsilon)$ of Arya and Mount [SODA'16].

\end{abstract}

\thispagestyle{empty}
\newpage
\tableofcontents

\thispagestyle{empty}
\newpage
\setcounter{page}{1}
\section{Introduction}

Computing a minimum spanning tree (MST) is one of the most fundamental problems in combinatorial optimization, with numerous applications.
In this paper we consider the following approximate version of the problem.
Given an $n$-point metric space $(X,\dist)$ and an error parameter $\epsilon \in (0,1)$, the objective is to compute a spanning tree on $X$ whose total weight is at most $(1+\epsilon)$ times the weight of the exact MST.
We are interested in the setting where the input metric has bounded doubling dimension, a standard notion of low-dimensionality for general metric spaces that includes bounded-dimensional Euclidean spaces as a special case.
Throughout, we let $\ddim$ denote the doubling dimension of $X$; a formal definition is given in \Cref{sec:preliminaries}.

A straightforward solution is to construct the complete weighted graph on $X$ and compute its MST, which takes quadratic time in the number of points.
In Euclidean spaces, much better exact algorithms are known.
In the plane, Shamos and Hoey showed that the Euclidean MST is a subgraph of the Delaunay triangulation, which yields an $O(n \log n)$-time algorithm~\cite{ShamosHoey75}.
For higher but constant dimension, the problem can be solved in subquadratic time~\cite{Yao82,AgarwalES91}, but the running times remain far from linear except in very low dimensions. Recently, Chan and Zheng \cite{CZ24} obtained an algorithm
with $O(n^{2\lceil d/2 \rceil / (\lceil d/2 \rceil +1)})$, i.e., $O(n^{4/3})$ expected running time for $d=3$ and $d=4$.

This has led to extensive work on approximation algorithms.
In Euclidean spaces, Vaidya's results on sparse spanners imply an $O(\epsilon^{-d}n\log n)$-time algorithm for computing a $(1+\epsilon)$-approximate MST~\cite{Vaidya91}.
Callahan and Kosaraju introduced the well-separated pair decomposition (WSPD) \cite{CallahanKosaraju95} and showed how to combine it with approximate bichromatic closest-pair computations to obtain a deterministic algorithm with running time $O(n\log n+n\epsilon^{-d/2}\log(1/\epsilon))$ for fixed dimension~\cite{CallahanK93}.
Arya and Mount obtained the best previously known deterministic running-time bound whose exponent of $1/\epsilon$ is independent of the dimension: for fixed $d$, their algorithm runs in time $O_d(n\log n+n\epsilon^{-2}\log^2(1/\epsilon))$~\cite{DBLP:conf/soda/AryaM16}.

Randomized algorithms offer further dimension-dependent bounds.
Arya and Chan improved the dependence on $\epsilon$ through faster approximate bichromatic closest-pair computations, obtaining an approximate Euclidean MST in expected time $O(n\log n\,\epsilon^{-(d/3+O(1))})$~\cite{AryaChan14}.
Using additional bit-packing techniques in the word-RAM model, they further obtain an expected running time of $O^*(n+n\epsilon^{-d/3}/\log^{2/3}n)$, where $O^*$ suppresses factors of the form $\epsilon^{-O(1)}$.
Arya, da Fonseca, and Mount later gave a randomized algorithm for $(1+\epsilon)$-approximate Euclidean MST with expected running time $O(n\log n\,\epsilon^{-(d/4+\alpha)})$, where $\alpha>0$ is an arbitrarily small constant~\cite{AryaFM17}.
Their approximate bichromatic closest-pair algorithm runs in expected time $O(n\epsilon^{-(d/4+\alpha)})$.
Concurrently, Chan obtained related improvements for approximate bichromatic closest pair via the polynomial method~\cite{Chan18}.
In high-dimensional Euclidean spaces, Har-Peled, Indyk, and Sidiropoulos construct spanners with $O(n^{1+1/c^2}\log n\log c)$ edges, which also yield $O(c)$-approximate spanning trees~\cite{HIS13}.

The more general setting of doubling metrics has also received considerable attention.
The \emph{doubling dimension} of a metric space $(X,\dist)$, denoted by $\ddim(X,\dist)$, or simply $\ddim(X)$ when the distance function is clear from the context, is the smallest $t\geq0$ such that every ball can be covered by at most $2^t$ balls of half the radius.
A metric space is called a \emph{doubling metric} if its doubling dimension is bounded.
Bounded doubling dimension generalizes bounded-dimensional Euclidean spaces and includes other natural metric classes, such as fixed-dimensional normed spaces equipped with the $\ell_1$ or $\ell_\infty$ distance.

Metrics with bounded doubling dimension have been studied systematically by Talwar~\cite{Talwar04}, and Har-Peled and Mendel developed a randomized near-linear-time construction of hierarchical nets, with applications to WSPDs and spanner construction~\cite{DBLP:journals/siamcomp/Har-PeledM06}.
These tools yield $(1+\epsilon)$-approximation algorithms for MST with near-linear dependence on $n$, but with a dependence on $\epsilon$ of the form $\epsilon^{-O(\ddim)}$.
Gottlieb and Krauthgamer explicitly give an approximate-MST bound of $2^{O(\ddim)}n\log n+\epsilon^{-O(\ddim)}n$ as part of their algorithms for nearly doubling metrics~\cite{GottliebK13}.

Lower bounds for MST construction depend on the metric class and computational model.
Das, Kapoor, and Smid established algebraic computation-tree lower bounds for approximate Euclidean MST construction, implying $\Omega(n\log n)$ time for every fixed approximation factor $c$ that is sufficiently large~\cite{DKS97}.
For general metrics given by a distance oracle, Indyk showed that computing a $c$-approximate MST requires $\Omega(n^2/c)$ distance queries for constant $c>1$~\cite{Indyk99}.
Andoni et al. showed that exact MST construction requires $\Omega(n^2)$ distance queries even in metrics of constant doubling dimension~\cite{ANOY14}.
These results leave open the sharp dependence on $\epsilon$ for approximate MST construction in doubling metrics.

Another line of work considers approximating only the \emph{weight} of the MST.
This direction originated with the graph algorithm of Chazelle, Rubinfeld, and Trevisan~\cite{ChazelleRubinfeldTrevisan05} (see also~\cite{PB25}).
Czumaj et al.\ adapted these ideas to Euclidean space~\cite{CzumajEtAl05}, and Czumaj and Sohler gave a sublinear-time algorithm for approximating the MST weight in arbitrary metric spaces, assuming access to a distance oracle~\cite{CzumajSohler09}.
These algorithms estimate the optimum weight, whereas our goal is to compute an explicit spanning tree.

\subsection{Our Results}
\label{sec:results}
For general doubling metrics, we use the standard real RAM with a unit-cost exact-distance oracle, specified in \Cref{def:computational-model}.
The input consists of point labels, and the oracle supplies their exact pairwise distances.
For Euclidean inputs, point coordinates are supplied explicitly to the real RAM, as in prior work on approximate Euclidean MSTs, including~\cite{DBLP:conf/soda/AryaM16}.

\paragraph{MST in doubling metrics.}
We give a deterministic algorithm that computes a $(1+\epsilon)$-approxi\-mation to MST in an $n$-point metric space of doubling dimension $\ddim$ in time $2^{O(\ddim)} n \bigl(\log n + \epsilon^{-1}\log(1/\epsilon)\bigr)$.
For bounded doubling dimension, this improves the previous $\epsilon^{-O(\ddim)}$ dependence to essentially linear in $\epsilon^{-1}$.
As a special case, it also improves the previous best deterministic bound in bounded-dimensional Euclidean spaces due to Arya and Mount~\cite{DBLP:conf/soda/AryaM16} by almost a factor of~$1/\epsilon$.

\begin{restatable}{theorem}{thmmain}
	\label{thm:MST_alg_improved}
Let $\epsilon \in (0, 1/2)$.
 There is a deterministic algorithm in the real RAM model with a unit-cost exact-distance oracle of \Cref{def:computational-model} that, given as input an $n$-point metric space $(X,\dist)$ with doubling dimension $\ddim$, computes a $(1+\epsilon)$-approximation to MST in time
\[
2^{O(\ddim)} \left(
  n \log n + \min\left\{n^2, \frac{n}{\epsilon} \log(1/\epsilon) \right\}
  \right).
\]
\end{restatable}

\paragraph{MST in low-dimensional Euclidean spaces.}
In the Euclidean MST problem, the input is an $n$-point subset of $\R^d$ given in vector representation, and the MST is defined w.r.t. Euclidean distance, i.e., $\dist(x, y) := \|x - y\|_2, \ \forall x, y \in X$.
Throughout, we consider low-dimensional Euclidean spaces, i.e., $d = O(1)$.
It is known that $(\R^d, \|\cdot\|_2)$ has doubling dimension $\Theta(d)$, so \Cref{thm:MST_alg_improved} directly implies a $(1 + \epsilon)$-approximation to Euclidean MST running in time $O_d(n (\log n + \epsilon^{-1} \log(1/\epsilon)))$.
Using a faster subroutine for Euclidean bichromatic closest pair, we are able to further save an $\epsilon^{-\Theta(1/d)}$ factor in the running time.

\begin{restatable}{theorem}{thmeuclideanmst}
    \label{thm:euclidean_mst}
    Let $d \in \Z_+$ be a constant and $\epsilon \in (0, 1/2)$.
    There is a real-RAM algorithm that, given the coordinates of an $n$-point subset $X \subset \R^d$, computes a $(1 + \epsilon)$-approximation to Euclidean MST in expected time 
    \[
    2^{O(d)} n \log n + 
O_{d, \kappa}\left(
    \min\left\{n^2, \frac{n}{\epsilon^{1 - \beta_d}} \log^{1-\beta_d}(1/\epsilon) 
    \right\}
    \right), 
    \qquad \text{for } \beta_d = \frac{1}{\lceil d/2 \rceil + 1} - \kappa,
    \]
    where $\kappa > 0$ is a constant which can be made arbitrarily small,
    and $O_{d, \kappa}(\cdot)$ hides factors that depend only on $d$ and $\kappa$.
\end{restatable}
For $d\leq2$, an exact MST can already be computed in $O(n\log n)$ time~\cite{ShamosHoey75}.
For $d \geq 3$,
the previous best algorithm for computing a $(1+\epsilon)$-approximate Euclidean MST, due to Arya and Mount, runs in time $O_d(n\log n+n\epsilon^{-2}\log^2(1/\epsilon))$~\cite{DBLP:conf/soda/AryaM16}.
\Cref{thm:euclidean_mst} directly improves their algorithm.
For every fixed $d\geq3$, it also improves the $\epsilon^{-d/4}$  dependence over~\cite{AryaFM17} and places it in an additive term, rather than coupling it multiplicatively with $n\log n$.
Perhaps surprisingly, \Cref{thm:euclidean_mst} also implies that the $\epsilon^{-1}$ dependence in \Cref{thm:MST_alg_improved}, which is optimal for general doubling metrics given by an exact-distance oracle (see \Cref{thm:lowerbounds}), is not tight in the low-dimensional Euclidean setting.
We leave it as an interesting open question to determine the optimal dependence on $\epsilon$ for low-dimensional Euclidean MST.

\paragraph{Lower bounds.}
We complement our algorithm with lower bounds showing that its dependence on $n$ and $\epsilon$ is essentially optimal for general doubling metrics.

\begin{restatable}{theorem}{thmlowerbounds}
\label{thm:lowerbounds}
For every sufficiently large integer $n$ and $0<\epsilon\le\frac12$, any randomized algorithm in the real RAM model with a unit-cost exact-distance oracle of \Cref{def:computational-model} that returns an explicit $(1+\epsilon)$-approximate spanning tree with probability at least $2/3$ on every $n$-point metric of doubling dimension at most $7$ has worst-case expected running time
\[
 \Omega\!\left(n\log n+\min\left\{n^2,\frac{n}{\epsilon}\log(1/\epsilon)\right\}\right).
\]
\end{restatable}

The two terms arise from lower bounds on different operations.
The $\Omega(n\log n)$ term bounds comparisons, including comparisons of real expressions computed from exact oracle answers.
After fixing the random tape, the comparison outcomes determine the labeled output in the real RAM of \Cref{def:computational-model}.
The accuracy-dependent term bounds exact-distance queries even when all computation between queries is unrestricted and free.
Both bounds therefore imply running-time lower bounds in the stated real-RAM model.
In both constructions, the points are given by labels without geometric information.
For bounded doubling dimension, these lower bounds show that the algorithm of \Cref{thm:MST_alg_improved} is asymptotically optimal in its dependence on $n$ and $\epsilon$.

\subsection{Technical Overview}

This overview explains the main ideas in a simplified setting.
To keep the presentation accessible, we suppress constant factors and technical details, use WSPDs in place of the WSPCs employed by the algorithm, and omit the neighborhood enlargements needed in the charging arguments.
The precise definitions, statements, and proofs appear in the subsequent sections.

On the algorithmic side, our main technical contribution is a level-sensitive error charging scheme, 
together with a novel time complexity analysis based on the probabilistic method.
On the lower-bound side, we use an output-counting argument to show that even an approximate MST of a relabeled path requires $\Omega(n\log n)$ comparisons.
For the accuracy-dependent bound, we construct several clusters with hidden cheap connections to a shared core; reusing the core across geometrically separated scales yields the additional $\log(1/\epsilon)$ factor while keeping the doubling dimension bounded.
Both techniques are simple and natural, and
we believe they may find applications in other MST-related problems.

\subsubsection{A  Conceptual Separation between MST in Doubling Metrics and Euclidean Spaces}
Before reviewing our techniques, we use a simple example to illustrate that the structure of MST in doubling metrics is fundamentally different from MST in Euclidean spaces.
In constant-dimensional Euclidean spaces, it is not hard to show that the maximum degree of MSTs is bounded by a constant.
Perhaps surprisingly, this fails dramatically in bounded doubling metrics: there exist such metrics whose MST has maximum degree $\Omega(n)$; this is true even for high-dimensional Euclidean spaces with bounded doubling dimension.

\begin{restatable}{proposition}{lemmalowerbounddegree}
	\label{lemma:MST_degree_lower_bound}
For every integer $n$ there exists a finite point set $X \subset \mathbb{R}^{n}$ equipped with the Euclidean distance such that the metric space $(X, \|\cdot\|_2)$ has bounded doubling dimension, i.e., $\ddim(X)=O(1)$, while the minimum spanning tree of $X$ has maximum degree $\Omega(n)$.
\end{restatable}

\begin{proof}
Let $e_1, e_2, \dots, e_n$ be the standard basis vectors of $\R^n$.
    Define $X = \{\mathbf{0}, 2 e_1, 2^2 e_2, \dots, 2^n e_n\}$.
    It is easy to see that $\ddim(X) = O(1)$.
    However, the MST of $X$ connects every point to $\mathbf{0}$, and thus has maximum degree $n$. See Figure \ref{fig:unbounded-mst-degree} for a visualization.
\end{proof}

\begin{figure}[t]
\centering
\begingroup
\definecolor{coverGreen}{HTML}{009E73}
\definecolor{outerOrange}{HTML}{D55E00}

\resizebox{\linewidth}{!}{\begin{tikzpicture}[
  x=1cm,
  y=1cm,
  font=\small,
  line cap=round,
  line join=round,
  every node/.style={inner sep=2pt},
  heading/.style={font=\bfseries,anchor=west},
  mst/.style={draw=blue!65!black,line width=1.2pt},
  vertex/.style={
    circle,
    fill=blue!65!black,
    inner sep=1.8pt
  }
]

\node[heading] at (-.5,5.1)
  {(a) Exponentially spaced points};

\coordinate (o) at (2.2,1.6);

\coordinate (p1) at (1.0,1.6);
\coordinate (p2) at (2.2,4.0);
\coordinate (pn) at (5.8,1.6);

\draw[mst] (o)
  -- node[midway,above=5pt] {$2$} (p1);

\draw[mst] (o)
  -- node[midway,left=5pt] {$2^2$} (p2);

\draw[mst] (o)
  -- node[midway,above=5pt] {$2^n$} (pn);

\node[vertex] at (p1) {};
\node[vertex] at (p2) {};
\node[vertex] at (pn) {};
\node[circle,fill=black,inner sep=2pt] at (o) {};

\node[below=5pt] at (o)
  {$\mathbf{0}$};

\node[left=5pt] at (p1)
  {$2e_1$};

\node[above=5pt] at (p2)
  {$2^2e_2$};

\node[right=5pt] at (pn)
  {$2^ne_n$};

\node[rotate=-35] at (3.8,3.2) {$\cdots$};

\node[align=center,text width=6.2cm] at (2.8,.65) {
  For $i\neq j$,\\[3pt]
  $\|2^ie_i-2^je_j\|_2
    =\sqrt{2^{2i}+2^{2j}}>2^i$.
};

\node[text=black!55] at (2.8,-.1)
  {Schematic; the actual axes are mutually orthogonal.};

\begin{scope}[xshift=7.6cm]

\node[heading] at (0,5.1)
  {(b) Covering at one scale};

\node[
  align=center,
  text width=6cm,
  text=black!55
] at (3,4.35) {
  Distances from the origin\\
  (schematic spacing)
};
\begin{scope}[yshift=-.3cm]
\fill[coverGreen!12,rounded corners=3pt]
  (-.18,2.86) rectangle (3.48,3.34);

\draw[coverGreen!70!black,rounded corners=3pt]
  (-.18,2.86) rectangle (3.48,3.34);

\draw[->,black!60]
  (-.3,3.1) -- (6.35,3.1);

\foreach \x in {0,.55,1.7,3.2}{
  \fill[blue!65!black]
    (\x,3.1) circle (1.8pt);
}

\foreach \x/\lab in {
  0/0,
  .55/2,
  1.7/{2^2},
  2.45/{\cdots},
  3.2/{2^j},
  4.6/{2^{j+1}},
  5.9/{2^{j+2}}
}{
  \node[anchor=base] at (\x,2.43)
    {$\lab$};
}

\draw[outerOrange,thick]
  (4.6,3.1) circle (.2);

\fill[outerOrange]
  (4.6,3.1) circle (1.8pt);

\fill[black!35]
  (5.9,3.1) circle (1.8pt);

\node[anchor=base] at (3.2,3.7)
  {$R/2$};

\node[anchor=base] at (4.6,3.7)
  {$R$};

\node[
  align=center,
  text=coverGreen!65!black,
  text width=3.7cm
] at (1.6,1.65) {
  Covered by $B(\mathbf{0},R/2)$.
};

\node[
  align=center,
  text=outerOrange,
  text width=2.8cm
] at (4.9,1.65) {
  $\leq1$ point in $(R/2,R]$.
};

\node[align=center,text width=6cm] at (3,.65)
  {Two balls of radius $R/2$ suffice.};
\end{scope}
\end{scope}
\end{tikzpicture}}
\endgroup

\caption{The construction in \Cref{lemma:MST_degree_lower_bound}.
The points are $\mathbf{0},2e_1,\dots,2^ne_n$.
(a) The origin is the unique nearest neighbor of every nonzero point, so the cut property forces all edges to the origin into the MST.
(b) Every nonsingleton ball of radius $R$ contains only points whose distance from the origin is at most $R$.
The points at distance at most $R/2$ are covered by $B(\mathbf{0},R/2)$, and at most one point remains.
A second ball therefore completes the cover.
Thus $\ddim(X) = O(1)$.
The panel illustrates $R=2^{j+1}$; the same argument holds for arbitrary $R$.}
\label{fig:unbounded-mst-degree}
\end{figure}

However, we are able to show that every doubling metric admits a $(1+\epsilon)$-approximate MST of maximum degree $2^{O(\ddim)} \log(1/\epsilon)$,
which is also tight for $\ddim = O(1)$.
See \Cref{lemma:apx_MST_degree,lemma:apx_MST_degree_lower_bound} for the formal statement.
These structural results may be of independent interest.

Conceptually, these structural results suggest that MST becomes structurally different once the Euclidean geometry is lost,
and designing MST algorithms in doubling metrics can be more challenging than in Euclidean spaces.

\subsubsection{Upper Bounds}

In this section, we discuss our algorithmic techniques.
We focus on general doubling metrics, i.e., \Cref{thm:MST_alg_improved}.
Throughout, we assume the doubling dimension of $(X, \dist)$ is constant.

\paragraph{Algorithmic framework.}
At a high level, our algorithm follows the classical well-separated pair approach to MST~\cite{CallahanK93,DBLP:conf/soda/AryaM16}. Let $t$ be a sufficiently large constant.
Roughly speaking, a $t$-well-separated pair decomposition (WSPD) of $X$ is a collection of pairs \sloppy $\Phi = \{(U_1, V_1), (U_2, V_2), \dots, (U_k, V_k)\}$ of nonempty subsets of $X$, 
such that for each pair $(U, V) \in \Phi$, there exists some level $\ell \geq 0$ with $\diam(U), \diam(V) = \Theta(2^\ell)$ and $\dist(U, V) \geq t \cdot 2^\ell$.
We call $\ell$ the \emph{level} of pair $(U, V)$.\footnote{In fact, the formal definition of WSPD is more general and does not depend on the level.
We use the level-based convention for ease of presentation.}
Moreover, for any two distinct points $x, y \in X$, there exists exactly one pair $(U, V) \in \Phi$ such that $x$ lies in either $U$ or $V$ and $y$ in the other set.
To compute an approximate MST of $X$, one standard approach is to select one representative edge $(u, v) \in U \times V$ from each WSPD pair $(U, V)$, and return the MST $T$ of the graph containing all representative edges~\cite{CallahanK93,DBLP:conf/soda/AryaM16}.
If all representative edges satisfy 
\begin{align}\label{eqn:selected_edges}
    \dist(u, v) \leq (1 + \epsilon) \dist(U, V) + \err(U, V), \qquad \forall \text{ representative edge } (u, v) \in U \times V,
\end{align}
then $T$ is a spanning tree of $X$, and its weight $w(T)$ satisfies 
\begin{align}\label{eqn:intro_MST_weight}
    w(T) \leq (1 + \epsilon) \MST(X) + \sum_{(U, V) \in \Phi} \err(U, V).
\end{align}

It is known that the WSPD can be constructed in $O(n \log n)$ time and has size $O(n)$.
The main challenge is therefore to approximate the closest pairs in \eqref{eqn:selected_edges} quickly enough.
Na\"ively, one can first construct $\epsilon \diam(U)$ and $\epsilon \diam(V)$-nets\footnote{Formally, for $\rho > 0$ and a set $S$, a subset $N \subseteq S$ is called a $\rho$-net of $S$ if (i) $\forall x, y \in N, x \neq y, \ \dist(x, y) \geq \rho$, and (ii) $\forall z \in S$, there exists $x \in N$, such that $\dist(z, x) \leq \rho$.}
on $U$ and $V$ respectively, and select the closest pair of net points as the representative edge $(u, v)$.
This incurs an $\epsilon$ multiplicative error and no additive error (i.e., $\err(U, V) = 0$) in \eqref{eqn:selected_edges}.
However, as the nets may have size $\epsilon^{-O(\ddim)}$, one ends up with an $\epsilon^{-O(\ddim)}$ factor in the running time.

\paragraph{Net size threshold in~\cite{DBLP:conf/soda/AryaM16}.}
To reduce the $\epsilon^{-O(\ddim)}$ dependence, \cite{DBLP:conf/soda/AryaM16} introduce a more refined approach that is sensitive to the net size:
Let $\gamma = \gamma(\epsilon) > 1$ be a threshold.
Instead of directly constructing an $\epsilon \diam(U)$-net (resp. $\epsilon \diam(V)$-net), the algorithm uses a sparser $\epsilon' \diam(U)$-net (resp. $\epsilon' \diam(V)$-net) of size at most $\gamma$, for some $\epsilon' \geq \epsilon$ that may depend on $U$ and $V$.
This incurs an additive error of $\err(U, V) = \epsilon'(\diam(U) + \diam(V)) = \Theta(\epsilon' 2^\ell)$ in \eqref{eqn:selected_edges}.
Moreover, since both nets have size at most $\gamma$, computing the closest net points takes $O(\gamma^2)$ time, and the total time complexity scales as $O(n \gamma^2)$.
Hence, $\gamma$ controls the tradeoff between approximation ratio and time complexity --- 
on the one hand, we require $\gamma$ to be sufficiently large so that we can choose smaller $\epsilon'$ for each pair $(U, V)$ and bound the total additive error (i.e., $\sum_{(U, V) \in \Phi} \err(U, V)$) in \eqref{eqn:intro_MST_weight} by $O(\epsilon) \MST(X)$; on the other hand, we require $\gamma$ to be sufficiently small so we have better time complexity.

\paragraph{Warm-up: an $O(n \epsilon^{-2} \log^2 n )$ algorithm.}
To illustrate the tradeoff between $\gamma$ and $\epsilon$,
let us first show that one basic version of the analysis from~\cite[Section 4]{DBLP:conf/soda/AryaM16} carries over to doubling metrics.
The geometric core of that analysis is
the following standard fact~\cite{CzumajEtAl05}: 
the MST of every $\rho$-net $N$ is lower bounded by $\rho |N|$.
More generally, let $T^*$ be the MST of $X$;
then for every subset $S \subseteq X$, 
\begin{align}\label{eqn:intro_MST_net_lb}
    N \text{ is an } \epsilon' \diam(S)\text{-net of } S
    \qquad \Longrightarrow \qquad 
    \epsilon' \diam(S) \cdot |N| \leq O(1) \cdot w(T^* \cap S),
\end{align}
where $T^* \cap S$ denotes the set of edges that have at least one endpoint in $S$.
See \Cref{lemma:net_size_MST} for a formal statement.
Therefore, if a WSPD pair $(U, V)$ admits an $\epsilon' \diam(U)$-net (resp. $\epsilon' \diam(V)$-net) of size at least $\gamma$, we can charge the additive error term to the local MST weight of either $U$ or $V$, i.e., 
\begin{align*}
    \err(U, V) \leq \frac{1}{\gamma} (w(T^* \cap U) + w(T^* \cap V)).
\end{align*}
For each MST edge $e \in T^*$, it is charged with $1/\gamma$ fraction of its weight $w(e)$ by each set $U$ in the WSPD that contains it.
Standard WSPD constructions ensure that
at most $O(\log n)$ such sets $U$ contain an  endpoint of $e$, hence 
$e$ is charged with a total of $O(\gamma^{-1} \log n) w(e)$, and the total additive error therefore sums up to $O(\gamma^{-1} \log n) \MST(X)$.
We thus need to set $\gamma = \Theta(\epsilon^{-1} \log n)$ to get a $(1 + \epsilon)$-approximation, yielding the $O(n \epsilon^{-2} \log^2 n)$ running time.

We note that the $\log n$ factor in the analysis above is unavoidable.
To see this, consider the following simple $1$-dimensional example:
Let
\begin{align}\label{eqn:one_dim_eg}
    X = \{0, 1, 2, 3, \dots, n\}.
\end{align}
For an arbitrary $t$-WSPD of $X$ with $t \geq 2$ and all $0 \leq i < j \leq \log n$, 
the following two pairs $(0, 2^i)$ and $(0, 2^j)$ cannot be contained in the same WSPD pair.
Otherwise, the separation property is violated.
Therefore, the point $0$ is contained in at least $\log n$ WSPD pairs.
Consequently, each MST edge incident to $0$ will be charged $\Omega(\log n)$ times.
Hence, new ideas are required to break the $\log n$ bottleneck.

\paragraph{Failure of edge decomposition in doubling metrics.}
\cite{DBLP:conf/soda/AryaM16} bypass the $\log n$ barrier via a more careful analysis of the same algorithm.
Specifically, they use a delicate charging analysis based on the decomposition of MST edges, which is specifically tailored to Euclidean geometry:
For ease of presentation,
consider the Euclidean plane $\R^2$ and suppose $X \subseteq [0, n]^2$.
Let $s_{\min}$ denote the minimum interpoint distance in $X$, and $T^*$ denote the MST of $X$.
Geometrically, $T^*$ is a collection of line segments in the plane.
In the analysis, they partition the square $[0, n]^2$ into small grids of side length $s_{\min}$.
Let $\mathcal{C}$ be the set of ``nonempty'' grids, i.e., the grids that intersect with the line segments in $T^*$.
In particular, a grid will be added to $\mathcal{C}$ if either it fully contains one MST edge, 
or there is one MST edge crossing it, even if the grid itself contains no point in $X$.
Obviously, $\mathcal{C}$ induces a partition of all line segments in $T^*$, where each edge of $T^*$ is broken into smaller line segments.
It is also not hard to show that $s_{\min} \cdot |\mathcal{C}|$ is a lower bound for $\MST(X)$.
The analysis distributes the additive error $\err(U, V)$ of all WSPD pairs to grids in $\mathcal{C}$ in a clever way, such that each grid in $\mathcal{C}$ receives a charge of $O(\gamma^{-1} \log(1/\epsilon)) s_{\min}$.
The total error can thus be bounded by $O(\gamma^{-1} \log(1/\epsilon)) s_{\min} \cdot |\mathcal{C}| \leq O(\gamma^{-1} \log(1/\epsilon)) \cdot \MST(X)$.
Choosing $\gamma = \Theta(\epsilon^{-1} \log(1/\epsilon))$ gives a $(1 + \epsilon)$-approximation algorithm in time $O(\gamma^2 n) = O(n \epsilon^{-2} \log^2(1/\epsilon))$.

However, the edge decomposition argument fails completely in doubling metrics.
Since we do not have a continuous ambient space as in the Euclidean setting, ``line segments'' are not well-defined in doubling metrics.
Consequently, it is not immediately clear how to decompose the MST edges, which requires breaking MST edges into smaller segments.
Perhaps the most natural way to do this is by introducing extra metric points as ``breakpoints'' of the edges.
However, such breakpoints can increase the doubling dimension significantly, and it is unclear how this can be implemented efficiently.
Therefore, new algorithmic ideas are required for doubling metrics.

\paragraph{Critical level and level-sensitive strategy.}
To bypass the $\log n$ barrier in doubling metrics, we need to improve the algorithm.
Intuitively, using a uniform threshold $\gamma$ for all WSPD pairs is wasteful.
In the $1$-dimensional example in \eqref{eqn:one_dim_eg}, for each level $0 \leq \ell \leq \log n$, there can be at most $n/2^\ell$ disjoint sets of diameter $2^\ell$.
For these level-$\ell$ WSPD pairs, the total time complexity of closest pair computation is $O(\gamma^2 n / 2^\ell)$, which is much smaller than $\gamma^2 n$ for large $\ell$.
This observation inspires the following level-sensitive strategy:
For higher level WSPD pairs, we can set the net size threshold more aggressively to, e.g., $\gamma_\ell := 1.1^\ell \gamma$.
Then the time complexity still adds up to $O(\gamma^2 n)$.
However, the error charged to a single edge $e$ (i.e., $O(\gamma_\ell^{-1}) w(e)$) now decreases geometrically with the level, with a total sum of $O(\gamma^{-1}) w(e)$, thus saving a $\log n$ factor.

This simple idea can be extended to more general settings.
For each level $\ell$, let $n_\ell$ denote the number of WSPD pairs at level-$\ell$.
Standard WSPD constructions ensure that 
\begin{align}\label{eqn:intro_lev_lb}
    \forall \ell, \qquad 2^\ell \cdot n_\ell \leq O(1) \cdot \MST(X),
\end{align}
i.e., $2^\ell \cdot n_\ell$ is a lower bound for $\MST(X)$ up to a constant factor.
We define the \emph{critical level} $\ell^*$ to be the level that achieves the strongest lower bound relative to its scale, i.e., 
\begin{align}\label{eqn:intro_critical_lev}
    \ell^* := \argmax_{\ell} 2^{\ell/2} \cdot n_\ell.
\end{align}
The maximality of $\ell^*$ immediately implies the following level-sensitive properties:
When $\ell < \ell^*$, the MST lower bound $2^\ell \cdot n_\ell$ in \eqref{eqn:intro_lev_lb} decreases geometrically as the level $\ell$ decreases;
when $\ell > \ell^*$, the number of WSPD pairs $n_\ell$ at level $\ell$ decreases geometrically as $\ell$ increases.
Formally, 
\begin{align*}
    \ell < \ell^* \qquad &\Longrightarrow \qquad 2^\ell \cdot n_\ell \leq 2^{-(\ell^* - \ell)/2} \MST(X); \\
    \ell > \ell^* \qquad &\Longrightarrow \qquad n_\ell \leq 2^{-(\ell - \ell^*) / 2} n_{\ell^*}.
\end{align*}

\begin{figure}[t]
\centering
\begingroup
\definecolor{scaleLow}{HTML}{0072B2}
\definecolor{scaleMiddle}{HTML}{009E73}
\definecolor{scaleHigh}{HTML}{D55E00}

\resizebox{\linewidth}{!}{\begin{tikzpicture}[
  x=1cm,
  y=1cm,
  font=\small,
  every node/.style={inner sep=2pt},
  heading/.style={font=\bfseries},
  cluster/.style={draw=black!30,line width=.5pt},
  chosen/.style={line width=1.1pt},
  explanation/.style={align=center,text width=4.5cm}
]

\node at (8,7.5) {
  $\displaystyle
  \ell^* = \argmax_{\ell}
  2^{\ell/2} \cdot n_\ell$
};

\node[anchor=west,text=black!55] at (.3,6.85)
  {Finer scales};
\node[anchor=east,text=black!55] at (15.7,6.85)
  {Coarser scales};

\draw[scaleLow,line width=1pt]
  (.3,6.5) -- (5.2,6.5);
\draw[scaleMiddle,line width=1pt]
  (5.2,6.5) -- (10.8,6.5);
\draw[scaleHigh,line width=1pt,-{Stealth}]
  (10.8,6.5) -- (15.7,6.5);

\foreach \x in {5.2,8,10.8}{
  \draw[black!55] (\x,6.35) -- (\x,6.65);
}

\node at (5.2,6.05)
  {$\ell^* - \log(1/\epsilon)$};
\node at (8,6.05)
  {$\ell^*$};
\node at (10.8,6.05)
  {$\ell^* + \log(1/\epsilon)$};
\node[anchor=west] at (15.8,6.5) {$\ell$};

\node[heading,text=scaleLow] at (2.65,5.35)
  {Low levels};
\node[heading,text=scaleMiddle] at (8,5.35)
  {Middle levels};
\node[heading,text=scaleHigh] at (13.35,5.35)
  {High levels};

\node[explanation] at (2.65,4.65)
   {No refinement required.};

\node[explanation] at (8,4.65)
  {Construct nets with size threshold};

\node[explanation] at (13.35,4.65)
  {Construct nets with size threshold};

\node at (8,3.95)
  {$\gamma=\epsilon^{-1}\log(1/\epsilon)$};
\node at (13.35,3.95)
  {$\tau_\ell=1.1^{\ell-\ell^*}$};

\draw[cluster] (1.55,2.85) ellipse (.9 and .55);
\draw[cluster] (3.75,2.85) ellipse (.9 and .55);

\foreach \x/\y in {
  1.05/3.05,1.30/2.55,1.95/3.15,2.05/2.65,
  3.25/3.10,3.50/2.55,4.05/3.15,4.30/2.65
}{
  \fill[black!50] (\x,\y) circle (1.1pt);
}

\draw[chosen,scaleLow] (1.55,2.85) -- (3.75,2.85);
\fill[scaleLow] (1.55,2.85) circle (1.8pt);
\fill[scaleLow] (3.75,2.85) circle (1.8pt);

\node[above=3pt] at (1.55,2.85) {$u$};
\node[above=3pt] at (3.75,2.85) {$v$};
\node at (1.55,2.05) {$U$};
\node at (3.75,2.05) {$V$};

\draw[cluster] (6.9,2.85) ellipse (.9 and .55);
\draw[cluster] (9.1,2.85) ellipse (.9 and .55);

\foreach \x/\y in {
  6.4/3.05,6.65/2.55,7.05/3.15,
  8.95/3.15,9.40/2.55,9.65/3.00
}{
  \fill[scaleMiddle] (\x,\y) circle (1.3pt);
}

\draw[chosen,scaleMiddle] (7.35,2.80) -- (8.65,2.90);
\fill[scaleMiddle] (7.35,2.80) circle (1.8pt);
\fill[scaleMiddle] (8.65,2.90) circle (1.8pt);

\node[below=3pt] at (7.35,2.80) {$u$};
\node[below=3pt] at (8.65,2.90) {$v$};
\node at (6.9,2.05) {$U$};
\node at (9.1,2.05) {$V$};

\draw[cluster] (12.25,2.85) ellipse (.9 and .55);
\draw[cluster] (14.45,2.85) ellipse (.9 and .55);

\foreach \x/\y in {
  11.75/3.05,12.00/2.55,12.40/3.15,
  14.30/3.15,14.75/2.55,15.00/3.00
}{
  \fill[scaleHigh] (\x,\y) circle (1.3pt);
}

\draw[chosen,scaleHigh] (12.70,2.80) -- (14.00,2.90);
\fill[scaleHigh] (12.70,2.80) circle (1.8pt);
\fill[scaleHigh] (14.00,2.90) circle (1.8pt);

\node[below=3pt] at (12.70,2.80) {$u$};
\node[below=3pt] at (14.00,2.90) {$v$};
\node at (12.25,2.05) {$U$};
\node at (14.45,2.05) {$V$};

\node[explanation] at (2.65,1.05)
  {Choose an arbitrary pair.};

\node[explanation, text=black!60] at (2.65, .25)
  {Lower bound decreases geometrically};

\node[explanation] at (8,1.05)
  {Choose the closest pair of net points.};

\node[explanation] at (13.35,1.05)
  {Choose the closest pair of net points.};

\node[explanation,text=black!60] at (8,.25)
  {Only $O(\log(1/\epsilon))$ levels.};

\node[explanation,text=black!60] at (13.35,.25)
  {WSPD pairs decrease geometrically.};

\end{tikzpicture}}
\endgroup

\caption{Level-sensitive strategy for edge selection:
Low-level WSPD pairs select an arbitrary pair of points.
At middle and high levels, 
construct nets until reaching the corresponding size thresholds $\gamma$ and $\tau_\ell$,
and select the closest net points.
The sets and their spatial arrangement are schematic.}
\label{fig:scale-regimes}
\end{figure}

These properties inspire our level-sensitive algorithmic strategy:
For \emph{high levels} that are significantly above $\ell^*$ (i.e., $\ell > \ell^* + \log(1/\epsilon)$), the number of WSPD pairs decreases geometrically, and therefore we can afford to spend geometrically increasing computational effort on higher levels by setting the net size threshold more aggressively to $\tau_\ell := 1.1^{\ell - \ell^*}$,
while keeping the total running time $O(n)$.
Each MST edge $e$ is charged an $O(\tau_\ell^{-1})$ fraction of its weight by each level $\ell$, which is a total of $O(\epsilon)$ fraction, thus keeping the total additive error below $O(\epsilon) \MST(X)$.
For \emph{low levels} that are significantly below~$\ell^*$ (i.e., $\ell < \ell^* - \log(1/\epsilon)$), since the lower bound in \eqref{eqn:intro_lev_lb} decreases geometrically, even an arbitrary edge $(u, v) \in U \times V$ with $\err(U, V) = \Theta(2^\ell)$
contributes a total error of only $\sum_\ell 2^\ell \cdot n_\ell = O(\epsilon) \MST(X)$.
The time complexity is $O(1)$ per pair, thus $O(n)$ in total.

It remains to deal with the \emph{middle levels} (i.e., $\ell^* - \log(1/\epsilon) \leq \ell \leq \ell^* + \log(1/\epsilon)$).
For these levels, we use a uniform threshold $\gamma$.
The time complexity is therefore $O(\gamma^2 n)$.
Following the previous error charging analysis, each MST edge $e$ is charged an $O(\gamma^{-1})$ fraction per level, but for only $\log(1/\epsilon)$ (as opposed to $\log n$) levels!
The total additive error thus scales as $O(\gamma^{-1} \log(1/\epsilon)) \MST(X)$.
By setting $\gamma = \Theta(\epsilon^{-1} \log(1/\epsilon))$, we get a $(1 + \epsilon)$-approximation with time complexity $O(\gamma^2 n) = O(n \epsilon^{-2} \log^2(1/\epsilon))$.
This already matches the bound for Euclidean spaces in~\cite{DBLP:conf/soda/AryaM16}.
In fact,
a more refined analysis shows that each point queries its distance to at most $O(\gamma)$ many other points per level, so the time complexity is $O(n \gamma \log(1/\epsilon)) = O(n \epsilon^{-1} \log^2(1/\epsilon))$.

\Cref{fig:scale-regimes} summarizes the level-sensitive edge selection strategy.

\paragraph{Shaving $\log(1/\epsilon)$ via the probabilistic method.}
To achieve the final bound $O(n \epsilon^{-1} \log(1/\epsilon))$, we need to save a $\log(1/\epsilon)$ factor in the time complexity analysis.
Note that the current $\log^2(1/\epsilon)$ comes from two sources:
one is from $\gamma$, as each point participates in $O(\gamma)$ distance comparisons per level;
the other comes from the number of middle levels, which is $O(\log(1/\epsilon))$.

The second source of $\log(1/\epsilon)$ seems unavoidable under the current framework, as each point is contained in at least one WSPD pair at each (middle) level, and we must create a range of $\log(1/\epsilon)$ middle levels.
The first source, however, is far from tight.
There exist points that take part in far less than $\gamma$ distance comparisons during the closest pair computation.
Indeed, consider a WSPD pair $(U, V)$ with $|U|, |V| \ll \gamma$; then the closest pair $(u, v) \in U \times V$ can be selected in time $|U| \cdot |V| \ll \gamma^2$.
In other words, each point from $U$ or $V$ queries its distance to at most $\max\{|U|, |V|\}$ other points.

With the intuition above, we bound the time complexity of computing the representative edge from $(U, V)$ by 
$\min\{|U|, \gamma\} \cdot \min\{|V|, \gamma\}$.
Thus, the total complexity can be upper bounded by
\begin{align}\label{eqn:intro_total_time}
    \sum_{(U, V) \in \Phi} \min\{|U|, \gamma\} \cdot \min\{|V|, \gamma\}.
\end{align}

To upper bound \eqref{eqn:intro_total_time},
our key technical ingredient is a new structural property of the well-separated pair decomposition:
\begin{align}\label{eqn:structural_property_WSPD}
    \forall T \subseteq X, \qquad 
    \Big|\Big\{(U, V) \in \Phi \colon U \cap T \neq \emptyset \text{ and } V \cap T \neq \emptyset\Big\}\Big| \leq O(|T|).
\end{align}
Namely, for every subset $T$ of $X$, there are at most $O(|T|)$ pairs $(U, V)$ such that $T$ intersects both $U$ and $V$.
This bound holds across the entire WSPD, rather than separately at each level.
It provides control over the distribution of pairs beyond the usual bound on their total number.

We combine the structural property \eqref{eqn:structural_property_WSPD} with a novel sampling argument:
Construct a random subset $T \subseteq X$ by including every point of $X$ independently with probability $1/\gamma$.
Consider the number $\xi$ of WSPD pairs $(U, V)$ such that $U \cap T \neq \emptyset$ and $V \cap T \neq \emptyset$ hold simultaneously.
We count $\E[\xi]$ using two different methods.
On the one hand, by the structural property \eqref{eqn:structural_property_WSPD}, we have $\E[\xi] \leq O(\E[|T|]) = O(n/\gamma)$.
On the other hand, for each pair $(U, V)$, as $U$ and $V$ are disjoint, the events $\{U \cap T \neq \emptyset\}$ and $\{V \cap T \neq \emptyset\}$ are independent.
A simple probability calculation gives
\begin{align*}
    \Pr[U \cap T \neq \emptyset \text{ and } V \cap T \neq \emptyset] \geq \frac{1}{\gamma^2}
    \min\{|U|, \gamma\} \cdot \min\{|V|, \gamma\}.
\end{align*}
This already upper bounds \eqref{eqn:intro_total_time} by $\gamma^2 \E[\xi] \leq O(n \gamma) = O(n \epsilon^{-1} \log(1/\epsilon))$.
We note that the sampling is used only in the analysis; the algorithm itself remains deterministic.

\paragraph{Improvement in Euclidean spaces.}
For points given by coordinates in fixed-dimensional Euclidean space, we can accelerate the closest-pair computations within the same framework.
For two retained sets of sizes $a$ and $b$, the general-metric algorithm examines all $ab$ candidate pairs.
In Euclidean space, we instead use a bichromatic closest-pair algorithm whose main running-time term is $(ab)^{1-\beta_d}$, where $\beta_d=\Theta(1/d)$, together with an additional term that is nearly linear in $a+b$.
Since this subroutine returns an exact closest pair, the approximation analysis remains unchanged.

The improvement also carries over to the aggregate running time.
Our analysis bounds the sum of the products $ab$ over all relevant pairs, while the WSPD contains only $O_d(n)$ pairs.
H\"older's inequality then bounds the sum of $(ab)^{1-\beta_d}$ by $O_d(n^{\beta_d})$ times the $(1-\beta_d)$-th power of the aggregate product bound.
This reduces the accuracy-dependent term from $n\epsilon^{-1}\log(1/\epsilon)$ to $n\epsilon^{-(1-\beta_d)}\log^{1-\beta_d}(1/\epsilon)$.

\subsubsection{Lower Bounds}

Our lower bound combines two constructions.
The first forces $\Omega(n\log n)$ comparisons, while the second forces many exact-distance queries as a function of $\epsilon$, even when computation between queries is unrestricted and free.
Together, they give the running-time lower bound in \Cref{thm:lowerbounds}.

\paragraph{The comparison lower bound.}
We use equally spaced points on the real line whose labels are permuted.
There are $n!/2$ distinct labeled metrics of this form, since reversing the order preserves all distances.
An exact MST reveals the order up to reversal, giving the usual connection to sorting.
An approximate MST, however, need not connect consecutive points, so the same output tree may be valid for several different orders.

Our main observation is that, for any fixed approximation factor $\alpha$, a single labeled tree can approximate only $O(\alpha)^n$ many of these orders.
To see the intuition, root the tree at one label and describe an order by the position of the root and the signed displacements along the tree edges.
The approximation guarantee bounds the sum of the absolute displacements by $\alpha n$.
Counting the possible displacements and signs therefore gives only $O(\alpha)^n$ many orders.
Since there are $n!/2$ input metrics, the algorithm must distinguish sufficiently many outputs to require $\Omega(n\log n)$ comparisons.
An output-counting argument makes this reasoning apply also to randomized algorithms with constant success probability and expected running time.
It uses only a constant number of outcomes per comparison, and therefore permits comparisons of arithmetic expressions formed from exact distances.

\paragraph{The accuracy-dependent query lower bound.}
For the accuracy-dependent bound, we construct a metric consisting of a large core and several smaller satellites.
Our starting point is the following doubling subset $X_0 = \{\mathbf{0}, 2 e_1, 2^2 e_2, \dots, \epsilon^{-1} e_{\log(1/\epsilon)}\}$ under $\ell_1$ distance.
The MST $T^*$ of $X_0$ connects every point to the origin $\mathbf{0}$, and $\MST(X_0) = \Theta(\epsilon^{-1})$.
Moreover, any different connection increases the spanning tree weight by at least $2$.
Thus, for a sufficiently small constant $c > 0$, $T^*$ is the only $(1 + c \epsilon)$-approximate MST for $X_0$.
Since $T^*$ already connects every point to its nearest neighbor,
any $(1 + c \epsilon)$-approximation algorithm must find the nearest neighbor for each point.

However, the analysis above does not give any meaningful lower bound.
First, as there are only $\log(1/\epsilon)$ points in $X_0$, even an $O(\log^2(1/\epsilon))$ brute force computation is cheap.
Second, since the construction is deterministic, an algorithm can cheat by simply guessing $\mathbf{0}$ as the common nearest neighbor.

Nevertheless, we can amplify the instance by replacing each point in $X_0$ by a \emph{cluster} of points, while keeping the distance between clusters unchanged.
In particular, we replace the origin $\mathbf{0}$ by a unit grid consisting of $\epsilon^{-1} \log(1/\epsilon)$ points, called the \emph{core} $S$.
Replace each point $2^i e_i$ by a unit grid consisting of $\epsilon^{-1}$ points, called a \emph{satellite} $P_i$.
The resulting point set $X$ is the (disjoint) union of the core and all $\log(1/\epsilon)$ satellites, which contains $n = 2 \epsilon^{-1} \log(1/\epsilon)$ points.

Among all pairs of points between $S$ and $P_i$, there is one \emph{marked pair} at distance $\rho_i = 2^{i - 1} \rho$, which is sampled uniformly at random from $S \times P_i$.
All other pairs in $S \times P_i$ are at distance $\rho_i + 1$.
For two distinct satellites $P_i, P_j$, all pairs in $P_i \times P_j$ are at distance $\rho_i + \rho_j$.
The distances $\rho_i$ are chosen in an appropriate way such that 
(i) the doubling dimension of $X$ is $O(1)$; 
(ii) the MST of $X$ first connects all points inside the core and satellites, then connects all $\log(1/\epsilon)$ marked pairs between the core and all satellites; and 
(iii) $\MST(X) = \Theta(\epsilon^{-1} \log(1/\epsilon))$.
\Cref{fig:exact-lb-gadget} illustrates the construction and the corresponding search blocks.

\begin{figure}[t]
\centering
\begingroup
\definecolor{coreBlue}{HTML}{0072B2}
\definecolor{satelliteGreen}{HTML}{009E73}
\definecolor{markedOrange}{HTML}{D55E00}

\resizebox{\linewidth}{!}{\begin{tikzpicture}[
  x=1cm,
  y=1cm,
  font=\small,
  line cap=round,
  line join=round,
  every node/.style={inner sep=2pt},
  grid edge/.style={draw=black!23,line width=.35pt},
  marked edge/.style={draw=markedOrange,line width=1.2pt},
  panel heading/.style={font=\bfseries,anchor=west}
]

\node[panel heading] at (.1,7.5)
  {(a) Shared core, multiple scales};
\node[panel heading] at (8.2,7.5)
  {(b) Independent search blocks};

\fill[coreBlue!6] (.4,2.8) rectangle (2.8,5.6);
\draw[coreBlue!85!black,line width=.6pt]
  (.4,2.8) rectangle (2.8,5.6);

\node at (1.6,6.0) {$S$};
\node at (1.6,2.4) {$m = \epsilon^{-1} \log(1/\epsilon)$ points};

\foreach \y in {3.15,3.85,4.55,5.25}{
  \draw[grid edge] (.8,\y) -- (2.4,\y);
}
\foreach \x in {.8,1.333,1.867,2.4}{
  \draw[grid edge] (\x,3.15) -- (\x,5.25);
  \foreach \y in {3.15,3.85,4.55,5.25}{
    \fill[coreBlue!85!black] (\x,\y) circle (1.1pt);
  }
}

\foreach \cy/\idx in {6.2/1,4.25/2,1.9/{k}}{
  \fill[satelliteGreen!6]
    (5.6,{\cy-.4}) rectangle (7.2,{\cy+.4});
  \draw[satelliteGreen!85!black,line width=.6pt]
    (5.6,{\cy-.4}) rectangle (7.2,{\cy+.4});

  \node at (6.4,{\cy+.72}) {$P_{\idx} $};

  \foreach \dy in {-.22,.22}{
    \draw[grid edge]
      (5.85,{\cy+\dy}) -- (6.95,{\cy+\dy});
  }
  \foreach \x in {5.85,6.4,6.95}{
    \draw[grid edge]
      (\x,{\cy-.22}) -- (\x,{\cy+.22});
    \fill[satelliteGreen!85!black]
      (\x,{\cy-.22}) circle (1.1pt);
    \fill[satelliteGreen!85!black]
      (\x,{\cy+.22}) circle (1.1pt);
  }
}

\node at (6.4,3.15) {$\vdots$};
\node at (6.4,1.15) {$b = \epsilon^{-1}$ points each};

\draw[marked edge]
  (2.4,5.25) -- node[pos=.48,above=5pt,text=blue]
  {$h_1:\ \rho$} (5.85,5.98);

\draw[marked edge]
  (1.867,4.55) -- node[pos=.55,above=5pt,text=blue]
  {$h_2:\ 2 \rho$} (6.4,4.47);

\draw[marked edge]
  (2.4,3.15) -- node[pos=.48,below=5pt,text=blue]
  {$h_{k}:\ 2^{k-1}\rho$} (5.85,1.68);

\foreach \x/\y in {
  2.4/5.25,5.85/5.98,
  1.867/4.55,6.4/4.47,
  2.4/3.15,5.85/1.68
}{
  \fill[markedOrange] (\x,\y) circle (1.8pt);
}

\node at (12.1,6.95)
  {The same core labels index every block};

\foreach \bx/\idx in {8.9/1,11.0/2,13.8/{k}}{
  \node at ({\bx+.825},6.4) {$S\times P_{\idx}$};
  \fill[black!3]
    (\bx,2.9) rectangle ({\bx+1.65},5.9);
}

\fill[markedOrange!23]
  (10.0,4.1) rectangle (10.55,4.7);
\fill[markedOrange!23]
  (11.0,5.3) rectangle (11.55,5.9);
\fill[markedOrange!23]
  (14.35,2.9) rectangle (14.9,3.5);

\foreach \bx in {8.9,11.0,13.8}{
  \foreach \j in {1,2}{
    \draw[grid edge]
      ({\bx+.55*\j},2.9) -- ({\bx+.55*\j},5.9);
  }
  \foreach \j in {1,2,3,4}{
    \draw[grid edge]
      (\bx,{2.9+.6*\j}) -- ({\bx+1.65},{2.9+.6*\j});
  }
  \draw[black!60,line width=.55pt]
    (\bx,2.9) rectangle ({\bx+1.65},5.9);
}

\foreach \y in {3.5,4.1,4.7,5.3}{
  \draw[black!20,densely dotted]
    (10.55,\y) -- (11.0,\y);
  \draw[black!20,densely dotted]
    (12.65,\y) -- (13.8,\y);
}

\node[fill=white] at (13.225,4.4) {$\cdots$};

\foreach \x/\y in {
  10.275/4.4,11.275/5.6,14.625/3.2
}{
  \fill[markedOrange] (\x,\y) circle (2pt);
}

\draw[black!55]
  (8.73,2.9) -- (8.6,2.9) -- (8.6,5.9) -- (8.73,5.9);
\node[rotate=90] at (8.3,4.4) {$m$ shared rows};

\foreach \bx in {8.9,11.0,13.8}{
  \draw[black!55]
    (\bx,2.75) -- (\bx,2.6)
    -- ({\bx+1.65},2.6) -- ({\bx+1.65},2.75);
  \node at ({\bx+.825},2.25) {$b$ columns};
}

\fill[markedOrange!23] (9.0,1.45) rectangle (9.3,1.75);
\draw[black!55] (9.0,1.45) rectangle (9.3,1.75);
\fill[markedOrange] (9.15,1.6) circle (1.7pt);
\node[anchor=west] at (9.45,1.6) {$h_i:\ \rho_i$};

\fill[black!3] (12.0,1.45) rectangle (12.3,1.75);
\draw[black!55] (12.0,1.45) rectangle (12.3,1.75);
\node[anchor=west] at (12.45,1.6)
  {$\rho_i+1$ otherwise};

\node at (12.1,.8)
  {$mb$ candidates in each of $k$ independent blocks};

\node[text=blue,align=center] at (8,.05)
  {Grid sizes, matrix dimensions, and spatial placement are schematic.};

\end{tikzpicture}}
\endgroup

\caption{
The shared-core gadget and its hidden-distance blocks.
The core $S$ has $m = \epsilon^{-1} \log(1/\epsilon)$ points, and each of the $k = \log(1/\epsilon)$ satellites $P_i$ has $b = \epsilon^{-1}$ points, giving $n = 2 \epsilon^{-1} \log(1/\epsilon)$ points in total.
Each cluster carries the unit-grid metric.
In each block $S\times P_i$, one independently and uniformly chosen marked pair $h_i$ has distance $\rho_i=2^{i - 1}\rho$, while every other pair has distance $\rho_i+1$.
Distances between distinct satellites $P_i$ and $P_j$ are $\rho_i+\rho_j$.
Each matrix uses the same $m$ core points as its rows.
}
\label{fig:exact-lb-gadget}
\end{figure}

Our key observation is that:
If a spanning tree misses $z$ marked pairs, then its weight is at least $\MST(X) + z$.
Hence,
for a sufficiently small constant $c > 0$, any $(1 + c \epsilon)$-approximation algorithm must detect at least $\Omega(\log(1/\epsilon))$ marked pairs.
In other words, the algorithm must be able to compute the (exact) closest pairs for at least $\Omega(\log(1/\epsilon))$ core-satellite pairs $S \times P_i$.
Recall that for each satellite $P_i$, the marked pair in $S \times P_i$ is uniformly hidden among $\epsilon^{-2} \log(1/\epsilon)$ pairs.
Therefore, finding $\Omega(\log(1/\epsilon))$ hidden pairs takes 
\begin{align*}
    \Omega(\epsilon^{-2} \log^2(1/\epsilon)) = \Omega(n \epsilon^{-1} \log(1/\epsilon))
\end{align*}
distance queries in expectation.

\subsection{Organization of the Paper}

The remainder of the paper is organized as follows.
In \Cref{sec:preliminaries}, we define our computational model, recall basic metric notions, and introduce net hierarchies.
In \Cref{sec:WSPC_main}, we construct well-separated pair covers and establish their structural properties.
In \Cref{sec:main}, we present our main algorithm and analyze it.
Specifically, \Cref{sec:main:algo} describes the algorithm, \Cref{sec:main:analysis} proves its approximation guarantee, and \Cref{sec:main:time} establishes the running-time bound.
In \Cref{sec:euclidean}, we obtain the improved running time for low-dimensional Euclidean spaces.
In \Cref{sec:lower-bounds}, we prove our comparison and exact-distance query lower bounds.
Finally, \Cref{sec:conclusion} discusses open questions and directions for future work.

\section{Preliminaries}
\label{sec:preliminaries}

We first specify the computational model used for our results on general doubling metrics.
We then recall the metric notions and geometric tools used throughout the paper.

\begin{definition}[Real RAM with a unit-cost exact-distance oracle]
\label{def:computational-model}
We use the standard real-RAM convention with separate real and integer registers (see, e.g.,~\cite{EHM20}).
Real registers store exact real numbers and support the arithmetic operations $+,-,\times,/$ and comparisons.
Integer registers store point labels, pointers, and counters and support ordinary RAM operations, including indirect addressing.
Integer values may also be used as exact real arithmetic operands.

The input consists of $n$ distinct point labels and an oracle that, given two point labels $x,y$, returns the exact distance $\dist(x,y)$ in a real register.
Accuracy parameters are supplied as exact real inputs.
The labels carry no geometric information, and no embedding or other geometric representation is supplied.
Each oracle call, real arithmetic operation, comparison, and ordinary RAM operation costs one unit.

Comparisons may involve arbitrary previously computed real values, accuracy parameters, and fixed real program constants, and have at most the three outcomes $<,=,>$.
The program contains only a fixed finite collection of real constants, independent of the input.
Real values affect discrete control flow through these comparisons.
There is no primitive for converting a real-register value to an integer-register value, and real-register values cannot be used directly as memory addresses, oracle-query labels, or output vertex labels.
In particular, unit-cost floor operations are not available.
Randomized algorithms use an input-independent random tape.

An explicit spanning tree is output as a list of $n-1$ pairs of point labels stored in integer registers.
\end{definition}

For Euclidean coordinate inputs, square-root operations can be avoided by comparing squared distances and squared radii and handling rescaling implicitly.
The approximation guarantee still concerns the sum of Euclidean edge lengths.

Consider a metric space $(X, \dist)$ with $|X| \geq 2$.
For a point $x \in X $ and $r > 0$, define the \emph{ball} centered at $x$ with radius $r$ to be $B(x, r) := \{y \in X  \colon \dist(x, y) \leq r \}$.
The \emph{$r$-neighborhood} of a subset $S \subseteq X$ is defined as $B(S, r) := \bigcup_{x \in S} B(x, r)$. 
Denote the \emph{diameter} of $S$ to be $\diam(S) := \max_{x, y \in S} \dist(x, y)$.
The \emph{aspect ratio} of $S$ is defined as the ratio between the largest and smallest interpoint distances of $S$, denoted as $\Delta(S) := \frac{\max_{x, y \in S}\dist(x, y)}{\min_{x,y\in S,\;x\ne y}\dist(x, y)}$. 
Denote $\dist(x, S) := \min_{y \in S} \dist(x, y)$ as the \emph{distance} from $x$ to point set $S$, 
and $\dist(S, T) := \min_{x \in S, y \in T} \dist(x, y)$ as the distance between two sets $S$ and $T$.

\begin{definition}[Doubling dimension~\cite{GuptaKL03}]
\label{def:ddim}
The \emph{doubling dimension} of a metric space $(X, \dist)$ 
is the smallest $t\geq0$ such that every metric ball can be covered by at most $2^t$ balls of half the radius.
The doubling dimension of a point set $S\subseteq X$ is 
the doubling dimension of the metric space $(S, \dist)$,
and is denoted by $\ddim(S)$.
\end{definition}

\begin{definition}[Packing, covering and net]
    Consider a metric space $(X, \dist)$ and a subset $S \subseteq X$.
    For $\rho > 0$, $S$ is \emph{$\rho$-packing} if for all distinct $u, v \in S$, $\dist(u, v) \geq \rho$.
    $S$ is $\rho$-covering for $X$ if for every $x \in X $, there exists $u \in S$ such that $\dist(x, u) \leq \rho$.
    $S$ is called a $\rho$-net of $X$ if it is both $\rho$-packing and $\rho$-covering for $X $.
\end{definition}

\begin{lemma}[Packing Property~\cite{GuptaKL03}]
    \label{lemma:packing}
    If $S$ is $\rho$-packing, then $|S| \leq (2 + \diam(S) / \rho)^{O(\ddim(S))}$.
\end{lemma}

\subsection{Net Hierarchy}
\label{sec:net_hierarchy}
In this section we define net hierarchies that are the main data structure utilized in our algorithm. This data structure is very similar to existing work like the navigating nets construction by  Krauthgamer and Lee \cite{KrauthgamerL04} or the net tree construction by Har-Peled and Mendel \cite{DBLP:journals/siamcomp/Har-PeledM06}. We still give a precise definition of the data structure since it differs in details from related work and we are not aware of results in prior work that have the exact properties and model of computation needed for our purposes and could be used as a black box.

Let $(X, \dist)$ be a metric space with doubling dimension $\ddim$.
Without loss of generality, assume the minimum interpoint distance of $X$ is greater than $1$ and that the diameter of $X$ is $\Delta $.
Let $L = \lceil \log \Delta \rceil + 4 = O(\log \Delta)$.
Construct a sequence of \emph{nested nets} on $
X = N_0 \supseteq N_1 \dots \supseteq N_L, 
$
such that 
\begin{align}
    \forall 0 \leq \ell \leq L, \qquad N_\ell \text{ is } 2^{\ell - 3}\text{-packing and } 2^{\ell - 2}\text{-covering for } X.
    \label{eqn:N_l_properties}
\end{align}
Specifically, $N_0 = X$ and $N_L$ contains only one point in $X$.
The net hierarchy is defined to be the collection of these nets, i.e. $
    \nets := \{N_0, N_1, \dots, N_L\}.
$
It is clear to see the existence of such nested nets satisfying \eqref{eqn:N_l_properties}, 
by letting $N_0 = X$ and each $N_\ell$ be a $2^{\ell - 3}$-net of $N_{\ell - 1}$ for $\ell \geq 1$.

Elements in the data structure $\nets$ are represented as pairs $(u, \ell) \in X \times \{0, 1, \dots, L\}$, where $u \in N_\ell$ is a net point, and $\ell$ is called the \emph{level} of point $u$.
Note that if $(u, \ell)$ is an element in the net hierarchy, then $(u, \ell')$ is also an element for all $\ell' < \ell$, as the nets are nested.
When the context is clear we simply write $u$ as an element, denote its level by $\lev(u) = \ell$, and call $u$ a level-$\ell$ \emph{net point}.

\begin{definition}[Children]\label{def:children}
    Let $u, v$ be two net points of the net hierarchy.
    Say $v$ is a \emph{child} of $u$, if 
    (i) $\lev(v) = \lev(u) - 1$ and 
    (ii) $\dist(u, v) \leq 2^{\lev(u) + 1}$.
    Denote by $\child(u)$ the set of children of $u$.
\end{definition}

\begin{definition}[Descendants]\label{def:descendants}
    Let $u, v$ be two net points of the net hierarchy.
    Say $v$ is a \emph{descendant} of $u$, if either 
    (i) $v = u$, or
    (ii) there exists $w$, such that $w$ is a descendant of $u$ and $v$ is a child of $w$.
    For $\ell \leq \lev(u)$, denote by $\desc_\ell(u) \subseteq N_\ell$ the set of descendants of $u$ at level $\ell$, and by $\desc(u) := \bigcup_{\ell \leq \lev(u)} \desc_\ell(u)$ the set of descendants of $u$.
    In particular, $\desc_{\lev(u)}(u) = \{u\}$ and $\desc_{\lev(u) - 1}(u) = \child(u)$.
\end{definition}

\begin{remark*}
    Given the definition of children and descendants, it might be tempting to think of the net hierarchy as a tree, where each node corresponds to a net point.
    However, we note that the net hierarchy does not have a tree structure, as different nodes may share the same child/descendant.
\end{remark*}

The following lemma shows that the descendants of a net point $u$ are both packing and covering for a local ball around $u$ of certain radius.

\begin{restatable}{lemma}{lemmadescpackingcovering}\label{lemma:descendants_packing_covering}
    Let $u \in N_\ell$ be a level-$\ell$ net point.
    Then for every integer $k < \ell$,
    \begin{enumerate}[(1)]
        \item \label{it:desc_packing}
        Packing: $\desc_k(u) \subseteq B(u, 2^{\ell + 2})$, and $\desc_k(u)$ is $2^{k - 3}$-packing for $B(u, 2^{\ell + 2})$.
        \item \label{it:desc_covering}
        Covering: for every $x \in B(u, 2^\ell)$, there exists $v \in \desc_k(u)$, such that $\dist(x, v) \leq 2^{k - 2}$.
    \end{enumerate}
    For $k=\ell$, the packing and containment statement still holds, and covering holds with radius $2^\ell$. 
\end{restatable}

\begin{restatable}{lemma}{lemmanumchildrenparent}\label{lemma:num_children_parent}
    Let $u \in N_\ell$ be a level-$\ell$ net point.
    Then 
    \begin{enumerate}[(1)]
        \item \label{it:child_num} 
        $|\child(u)| \leq 2^{O(\ddim)}$.
        \item \label{it:ancestor_num}
        For each $k > \ell$, $u$ is the descendant of at most $2^{O(\ddim)}$ level-$k$ net points.
    \end{enumerate}
\end{restatable}

In the net hierarchy data structure $\nets$, each net point $u \in N_\ell$ maintains a set of pointers to all its children $\child(u) \subseteq N_{\ell - 1}$, for all $\ell \in \{1, \dots, L\}$.
The data structure can be constructed in near-linear time, as summarized in the following \Cref{lemma:net_hierarchy_construction}.
We provide the proof of \Cref{lemma:net_hierarchy_construction} in \Cref{sec:appendix:net-hierarchy} for completeness.

\begin{restatable}{lemma}{lemmanethierarchyconstruction}\label{lemma:net_hierarchy_construction}
    There is a data structure that, given 
    an $n$-point metric space $(X, \dist)$ with doubling dimension $\ddim$, minimum interpoint distance greater than $1$ and diameter $\Delta$,
    maintains the net hierarchy $\nets = \{N_0, N_1, \dots, N_L\}$ over $X$, such that 
    $X = N_0 \supseteq N_1 \dots \supseteq N_L$
    and each $N_\ell$ is $2^{\ell - 3}$-packing and $2^{\ell - 2}$-covering for $X$.
    Moreover, each net point $u \in N_\ell$ stores a set of pointers to all its children $\child(u) \subseteq N_{\ell - 1}$, for all $\ell \in \{1, \dots, L\}$.
    The data structure can be constructed in time $2^{O(\ddim)} n \log \Delta$.
\end{restatable}

\subsection{Bounding the Aspect Ratio and Approximating the Doubling Dimension}
\label{sec:spread_and_ddim}

By merging nearby points, we can reduce the aspect ratio $\Delta$ of $X$ to $n^{10} / \epsilon$.
After rescaling, we can assume henceforth that the minimum interpoint distance of $X$ is greater than $1$ and the diameter of $X$ is $\Delta \leq n^{10} / \epsilon$.
The formal statement together with its proof is given in \Cref{lemma:bounded_spread}.

Using the net hierarchy, we can approximate $\ddim(X)$ within a constant factor in time $2^{O(\ddim(X))} n \log \Delta = 2^{O(\ddim(X))} n \log(n/\epsilon)$.
\sloppy
See \Cref{lemma:ddim_apx} for a formal statement.

\section{Well-Separated Pair Cover}
\label{sec:WSPC_main}

Similar to~\cite{DBLP:conf/soda/AryaM16}, we use a well-separated pair decomposition (WSPD)-based algorithm for MST.
Unlike its simplicity in Euclidean spaces~\cite{CallahanKosaraju95}, 
WSPD in doubling spaces (e.g., \cite{Talwar04,DBLP:journals/siamcomp/Har-PeledM06}) has complicated geometric structures.
In this section, we construct a slightly relaxed version of WSPD, referred to as the \emph{well-separated pair cover (WSPC)}, which turns out to be sufficient for our MST application.
This relaxed notion has been studied by e.g.~\cite{DBLP:journals/siamcomp/Har-PeledM06,har2011geometric,HRR25,deberg2026}.
Our construction of WSPC is simple, and it also has nice geometric structures that integrate well with the net hierarchy in \Cref{sec:net_hierarchy}.

Let $t > 0$.
For two subsets $A, B \subseteq X$, we say a pair $(A, B)$ is $t$-separated if $\dist(A, B) > t \cdot \max\{\diam(A), \diam(B)\}$.
The well-separated pair cover (WSPC) is defined as follows.

\begin{definition}[Well-separated pair cover (WSPC)]
    \label{def:WSPC}
    Consider a finite metric space $(X, \dist)$.
    For $t \geq 2$, a \emph{$t$-well-separated pair cover} of $X$ is a collection of pairs $\Phi = \{(A_1, B_1), \dots, (A_k, B_k)\}$ of nonempty subsets of $X$ satisfying the following:
    \begin{enumerate}[(1)]
    \item \label{it:WSPC_separation}
    Separation: Each pair $(A, B) \in \Phi$ is $t$-separated, i.e., $\dist(A, B) > t \cdot \max\{\diam(A), \diam(B)\}$.
    \item \label{it:WSPC_covering}
    Covering: For any two distinct points $p,q \in X$, there exists \emph{at least} one pair $(A,B) \in \Phi$ such that $p$ lies in either $A$ or $B$ and $q$ in the other set.
\end{enumerate} 
\end{definition}

\begin{remark*}
    Compared to the standard definition of well-separated pair decomposition~\cite{CallahanKosaraju95,Talwar04}, the only difference in \Cref{def:WSPC} is in \eqref{it:WSPC_covering}, where we do not require $A_i \times B_i$ and $A_j \times B_j$ to be disjoint for $i \neq j$.\footnote{When referring to WSPC pairs, we slightly abuse the notation and treat $A \times B$ as a collection of unordered pairs, i.e., $A \times B := \{\{a, b\} \colon a \in A, b \in B, a \neq b\}$.}
\end{remark*}

The following lemma shows the relationship between well-separated pair cover and MST.
Similar results also appear in~\cite[Lemma 2.2]{DBLP:conf/soda/AryaM16} and~\cite[Lemmas 3.1 and 3.2]{CallahanK93} but are stated for WSPD and Euclidean spaces.
We extend them to also work for WSPC and general metrics.
We provide a proof in \Cref{sec:WSPC} for completeness.

\begin{restatable}{lemma}{lemmaWSPCMST}\label{lemma:WSPC_MST}
    Given a metric space $(X, \dist)$ with $|X| = n$ and a $t$-WSPC $\Phi$ for $t \geq 2$:
    \begin{enumerate}[(1)]
        \item For every pair $(U, V) \in \Phi$, there is at most one edge of $U \times V$ in the MST of $X$.
        \item \label{it:WSPC_spanning}
        For each MST edge $e_i^*$, let $(U_i, V_i) \in \Phi$ be an arbitrary pair that contains $e_i^*$, and let $(p_i, q_i)$ be an arbitrary pair of points from $(U_i, V_i)$.
        Then $\{(p_i, q_i)\}_{i = 1}^{n - 1}$ forms a spanning tree of $X$.
\end{enumerate}
\end{restatable}

Next, we show how to construct WSPC in doubling metrics efficiently based on the net hierarchy in \Cref{sec:net_hierarchy}.
We describe the construction algorithm in \Cref{alg:WSPC}, which is inspired by~\cite{Talwar04,DBLP:journals/siamcomp/Har-PeledM06}.

\begin{algorithm}[!ht]
\caption{$\WSPC(u, v)$}
\DontPrintSemicolon
\label{alg:WSPC}
assume $\lev(u) \geq \lev(v)$ (exchange $u$ and $v$ otherwise) \;

\If { $\{u, v\}$ is marked}{
    \Return $\emptyset$\;
} 
mark $\{u,v\}$ \;
\If{$\dist(u, v) > (2t + 2) 2^{\lev(u)}$}{
    \Return $\{(u, v)\}$ \;
}
\ElseIf{$\lev(u) = \lev(v) = 0$ and $u = v$}{
    \Return $\emptyset$ \;
}
\ElseIf{$\lev(u) = \lev(v) = 0$ and $u \neq v$}{
    \Return $\{(u, v)\}$ \;
}
\Else{
    \Return $\bigcup_{u' \in \child(u)} \WSPC(u', v)$
}
\end{algorithm}

To construct a $t$-WSPC, we call \Cref{alg:WSPC} with $\WSPC(r, r)$ and parameter $t$, where $r$ is the ``root'' of the net hierarchy (i.e., the only net point in $N_L$).
It returns a collection $\Psi$ of pairs of net points, which can be subsequently used to construct a $t$-WSPC.
In our application, we will set $t = \Theta(1)$.
The following lemma summarizes the properties of $\Psi$ that we will utilize.

\begin{restatable}{lemma}{lemmaWSPCdoubling}\label{lemma:WSPC_doubling}
    There is an algorithm that, given as input 
    a metric space $(X, \dist)$ with $\ddim(X) \leq \ddim$ and minimum interpoint distance greater than $1$ and diameter $\Delta$, 
    a net hierarchy $\nets$ over $X$ as in \Cref{sec:net_hierarchy},
    and a parameter $t \geq 32$,
    computes in $t^{O(\ddim)} n \log \Delta$ time a collection $\Psi$ of pairs of net points,
such that each pair $(u, v) \in \Psi$ satisfies the following:
    \begin{enumerate}[(1)]
\item \label{it:bounded_lev_diff}
        Bounded level difference: $|\lev(u) - \lev(v)| \leq 1$.
        \item \label{it:separation}
        Separation: Either $\dist(u, v) > (2 t + 2) \cdot \max\{2^{\lev(u)}, 2^{\lev(v)}\}$ or 
        $\lev(u) = \lev(v) = 0$.
        \item \label{it:Psi_covering}
        Covering: For any two distinct points $x, y \in X$, there exists at least one pair $(u, v) \in \Psi$ such that $x$ is in either $\wspcball{u}$ or $\wspcball{v}$ and $y$ in the other set.
\item \label{it:bounded_deg}
        Bounded degree: Each net point $u$ appears in at most $t^{O(\ddim)}$ pairs in $\Psi$.
        \item \label{it:few_NN}
        Few near neighbors: For each subset $T$ of $X$, there are at most $t^{O(\ddim)} |T|$ pairs $(u, v) \in \Psi$ such that $\dist(u, T) \leq 2^{\lev(u) + 2}$ and $\dist(v, T) \leq 2^{\lev(v) + 2}$ hold simultaneously.
        Consequently, $|\Psi| \leq t^{O(\ddim)} n$.
    \end{enumerate}
    In particular, 
    $\Phi = \{(\wspcball{u}, \wspcball{v}) \colon (u, v) \in \Psi\}$ is a $t$-WSPC of $X$, and it can be constructed 
    in time $t^{O(\ddim)} n \log \Delta$.
\end{restatable}

Let us briefly explain these properties:
\eqref{it:bounded_lev_diff} shows that each pair $(u, v) \in \Psi$ lies in almost the same level of the net hierarchy. 
\eqref{it:separation} and \eqref{it:Psi_covering} correspond to the separation and covering properties in \Cref{def:WSPC}.
They together imply that $\Phi = \{(\wspcball{u}, \wspcball{v}) \colon (u, v) \in \Psi\}$ is a $t$-WSPC of $X$.
In our MST algorithm, we will set $t = \Theta(1)$.
\eqref{it:bounded_deg} shows that each net point $u$ appears in only a constant number of pairs in $\Psi$; we will use it extensively in our charging of approximation error.
These are all standard properties that also appear in prior work e.g.~\cite{Talwar04,DBLP:journals/siamcomp/Har-PeledM06}.
Property \eqref{it:few_NN}, however, is more involved.
It suggests that for each subset $T \subseteq X$, there are only $O(|T|)$ pairs $(u, v) \in \Psi$ such that both $u$ and $v$ are close to $T$.
This property is crucial to our time complexity analysis.
To the best of our knowledge, it is not explicitly stated in any previous work before.

We note that although both $\Psi$ and $\Phi$ can be constructed efficiently, our MST algorithm does not need to construct the WSPC $\Phi$ explicitly.
Instead, we merely construct the collection $\Psi$ of centers and, starting from these centers, apply a net-based search on the net hierarchy.

Below we focus on the proof of \eqref{it:few_NN}.
The proof of properties \eqref{it:bounded_lev_diff}-\eqref{it:bounded_deg}, as well as the construction of $\Phi$, will be deferred to \Cref{sec:WSPC}.
The proofs of these properties are standard, and do not rely on \eqref{it:few_NN}.
Our proof of \eqref{it:few_NN} is based on a careful charging argument on top of $T$.
A similar charging technique can also be found in~\cite{KrauthgamerL04}, while their purpose is to bound the space complexity of the ANN data structure.

\begin{proof}[Proof of \Cref{lemma:WSPC_doubling} \eqref{it:few_NN}]
    The claim is immediate if $T=\emptyset$, so assume $T\neq\emptyset$. 
    We first count the number of qualified pairs $(u, v)$ with $\lev(u) = \lev(v) = 0$.
    For each point $x \in T$, there are at most $2^{O(\ddim)}$ net points $u \in N_0$ within distance $4$ from $x$.
    By \eqref{it:bounded_deg}, $u$ appears in at most $t^{O(\ddim)}$ pairs in $\Psi$, 
    so there are at most $t^{O(\ddim)} |T|$ pairs $(u, v)$ with $\dist(u, T) \leq 4$, $\dist(v, T) \leq 4$ and $\lev(u) = \lev(v) = 0$.

    Below we assume $\max\{\lev(u), \lev(v)\} > 0$.
    We first claim that for each pair $(u, v) \in \Psi$ with $\max\{\lev(u), \lev(v)\} > 0$, 
    \begin{align}\label{eqn:not_quite_separated}
        (2 t + 2) \max\{2^{\lev(u)}, 2^{\lev(v)}\} 
        \leq \dist(u, v) \leq (4 t + 8) \max\{2^{\lev(u)}, 2^{\lev(v)}\}.
    \end{align}
    The lower bound follows directly from property \eqref{it:separation}.
    For the upper bound,
    assume wlog that $\lev(u) \geq \lev(v)$.
    Obviously, $u, v$ are not the root node, we can thus consider the last call of \Cref{alg:WSPC} that invokes the call to $\WSPC(u, v)$.
    We have two cases:

    If the last call is $(u', v)$, where $u$ is the child of $u'$,
    then since the call to $(u', v)$ does not return $\{(u', v)\}$, we have
    $\dist(u', v) \leq (2t + 2) \max\{2^{\lev(u')}, 2^{\lev(v)}\} = (2 t + 2) 2^{\lev(u')}$.
    By triangle inequality, 
    \begin{align*}
        \dist(u, v) \leq \dist(u, u') + \dist(u', v)
        \leq 2^{\lev(u') + 1} + (2 t + 2) 2^{\lev(u')}
        = (2t + 4) 2^{\lev(u')}
        = (4t + 8) 2^{\lev(u)}.
    \end{align*}

    If the last call is $(u, v')$, where $v$ is a child of $v'$, then 
    since the call to $(u, v')$ does not return $\{(u, v')\}$, we have
    $\dist(u, v') \leq (2t + 2) \max\{2^{\lev(u)}, 2^{\lev(v')}\} \leq (4t + 4) 2^{\lev(u)}$.
    By triangle inequality,
    \begin{align*}
        \dist(u, v) \leq \dist(u, v') + \dist(v, v')
        \leq (4 t + 4) 2^{\lev(u)} + 2^{\lev(v') + 1}
        \leq (4t + 8) 2^{\lev(u)}.
    \end{align*}
    This proves \eqref{eqn:not_quite_separated}.

    Let $\mathcal{M} = \{M_0, M_1, \dots, M_L\}$ be a net hierarchy for $T$,
    i.e., $T = M_0 \supseteq M_1 \supseteq \dots \supseteq M_L$, and each $M_\ell$ is $2^{\ell - 3}$-packing and $2^{\ell - 2}$-covering for $T$.
    We note that $\mathcal{M}$ is independent of the net hierarchy $\nets$ of $X$.
    For each $x \in T$, define the \emph{death level} of $x$ as the smallest level at which it is absent, with the value $L+1$ if $x$ is in $M_L$, denoted by 
    \begin{align*}
        \ell_{\mathrm{death}}(x) := 1 + \max\{0 \leq \ell \leq L \colon x \in M_\ell\}.
    \end{align*}
    Since $T = M_0$, we note that $\ell_{\mathrm{death}}(\cdot)$ is well-defined.

    Now let us consider a pair $(u, v) \in \Psi$ with $\max\{\lev(u), \lev(v)\} > 0$.
    Assume wlog that $\lev(u) \geq \lev(v)$, and denote $\ell := \lev(u)$.
    By \eqref{eqn:not_quite_separated}, we have $(2 t + 2) 2^\ell \leq \dist(u, v) \leq (4 t + 8) \cdot 2^\ell$.
    If $\dist(u, T) \leq 2^{\lev(u) + 2}$ and $\dist(v, T) \leq 2^{\lev(v) + 2}$, then there exist $p, q \in T$, such that 
    $\dist(u, p) \leq 2^{\ell + 2}$ and $\dist(v, q) \leq 2^{\ell + 2}$.
    By triangle inequality, 
    \begin{align*}
        \dist(p, q) \leq \dist(u, v) + \dist(u, p) + \dist(v, q) 
        \leq (4t + 16) \cdot 2^\ell; \\
        \dist(p, q) \geq \dist(u, v) - \dist(u, p) - \dist(v, q) 
        \geq (2t - 6) \cdot 2^\ell.
    \end{align*}
    Let $x, y \in M_\ell$ be the closest points to $p$ and $q$, respectively.
    Since $M_\ell$ is $2^{\ell - 2}$-covering for $T$, we have $\dist(x, p) \leq 2^{\ell - 2}$ and $\dist(y, q) \leq 2^{\ell - 2}$.
    By triangle inequality, 
    \begin{align}
        &\dist(x, u) \leq \dist(x, p) + \dist(p, u) \leq 8 \cdot 2^\ell; \label{eqn:dxu}\\
        &\dist(y, v) \leq \dist(y, q) + \dist(q, v) \leq 8 \cdot 2^\ell; \label{eqn:dyv}\\
        &\dist(x, y) \leq \dist(p, q) + \dist(x, p) + \dist(y, q) \leq (4 t + 17) \cdot 2^\ell; \label{eqn:dxy} \\
        &\dist(x, y) \geq \dist(p, q) - \dist(x, p) - \dist(y, q) \geq (2 t - 7) \cdot 2^\ell 
        > 0. \label{eqn:x_y_distinct}
    \end{align}
\eqref{eqn:x_y_distinct} implies that $x$ and $y$ are distinct.
    By \eqref{eqn:dxy}, there exists a constant $c$, such that at most one of $x$ and $y$ can be in $M_{\ell + \lceil \log t \rceil + c}$.
    Equivalently, either $\ell_{\mathrm{death}}(x) \leq \ell + \lceil \log t \rceil + c$ or $\ell_{\mathrm{death}}(y) \leq \ell + \lceil \log t \rceil + c$.
    Charge the pair $(u, v)$ to $x$ if $\ell_{\mathrm{death}}(x) \leq \ell + \log t + c$; charge $(u, v)$ to $y$ otherwise.

    We now fix $z \in T$ and count the number of pairs $(u, v)$ charged to $z$.
    In fact, if $\ell = \max\{\lev(u), \lev(v)\}$, then $(u, v)$ is charged to $z$ only if $\ell_{\mathrm{death}}(z) - \lceil\log_2 t\rceil - c \leq \ell < \ell_{\mathrm{death}}(z)$.
    Moreover, by \eqref{eqn:dxu} and \eqref{eqn:dyv}, there exists $w \in \{u, v\}$, such that $\dist(w, z) \leq 8 \cdot 2^\ell$.
    For a fixed level $\ell$, the number of such $w$ is $2^{O(\ddim)}$, and the number of pairs $(u, v)$ charged to $z$ with $\max\{\lev(u), \lev(v)\} = \ell$ is at most $2^{O(\ddim)} \cdot t^{O(\ddim)} = t^{O(\ddim)}$ (Property~\eqref{it:bounded_deg}).
    Finally, note that there are at most $\lceil\log_2 t\rceil + c$ such relevant levels $\ell$, thus the number of pairs $(u, v)$ charged to $z$ is at most 
    \begin{align*}
        t^{O(\ddim)} (\lceil\log_2 t\rceil + c) = t^{O(\ddim)}.
    \end{align*}

    In conclusion, the total number of pairs $(u, v) \in \Psi$ with $\dist(u, T) \leq 2^{\lev(u) + 2}$ and $\dist(v, T) \leq 2^{\lev(v) + 2}$ is $t^{O(\ddim)} |T|$.
    Taking $T = X$ gives $|\Psi| \leq t^{O(\ddim)} n$.
\end{proof}

\section{A Fast Algorithm for MST in Doubling Metrics}
\label{sec:main}
In this section, we give a $(1 + \epsilon)$-approximation algorithm for MST with improved running time.

\thmmain*

At a high level, our algorithm follows the classical well-separated-pair approach to MST, going back to Callahan and Kosaraju and used in particular by Arya and Mount in the Euclidean setting~\cite{CallahanK93,DBLP:conf/soda/AryaM16}.
We first construct a net hierarchy together with a well-separated pair cover (WSPC).
For each pair $(U,V)$ in the cover, we add a representative edge intended to approximate the closest pair between $U$ and $V$.
Thus, as in previous work, the main task is to determine how accurately these representative edges need to be computed at the different scales. 

Arya and Mount search for these representative edges with an accuracy that is identical for all scales. By a basic analysis of this approach, the precision parameter has to compensate  $O(\log n)$ scales, which invokes an additional $\log n$ factor in the running time. They work around this issue with an improved analysis based on decomposition of edge segments, which is tailored to the Euclidean space and does not apply to doubling spaces.

Our starting point is the nested net hierarchy
$
X = N_0 \supseteq N_1 \supseteq \dots \supseteq N_L .$
As shown by \Cref{lemma:net_size_MST}, each level $\ell$ yields a lower bound on $\MST(X)$ through the quantity $2^{\ell/2}|N_\ell|$.
We therefore single out the level $\ell^*$ that gives the strongest such bound, namely
\begin{equation}\label{eqn:def_l*}
    \ell^* := \argmax_{0 \leq \ell \leq L} 2^{\ell/2} \cdot |N_\ell|.
\end{equation}
This level directly determines the algorithmic treatment of the different scales.

Recall that each WSPC pair $(U,V)$ has the form $U = B(u,2^{\lev(u)})$ and $V = B(v,2^{\lev(v)})$.
We associate with $(U,V)$ the level $\ell := \max\{\lev(u),\lev(v)\}$.
Depending on this level, our algorithm handles WSPC pairs differently.
More precisely, we call a level $\ell$ \emph{low} if $\ell < \ell^* - 10\log(1/\epsilon)$, 
\emph{high} if $\ell > \ell^* + 10\log(1/\epsilon)$,
and \emph{middle} otherwise.

If $\ell$ is a low level, then the diameter of $U$ and $V$ is already exponentially smaller than the lower bound certified at level $\ell^*$.
We therefore do not need an accurate closest-pair computation: it suffices to connect their centers $(u,v)$.
Although this is only a crude approximation for an individual pair, the corresponding distances are so small that the total error contributed by all low-level pairs is negligible.

If $\ell$ is a middle level, then we need a more accurate representative edge.
Here we follow the net-based closest-pair approach of Arya and Mount~\cite{DBLP:conf/soda/AryaM16}: we refine $U$ and $V$ by descending in the hierarchy until either we reach scale $\epsilon 2^{\lev(u)}$ or the corresponding descendant sets become too large, and then connect the closest pair among the resulting net points.
The key point is that there are only $O(\log(1/\epsilon))$ middle levels in total.
Thus, even though this step is more expensive than the crude treatment of low levels, it is applied only on a short interval of scales, which avoids the additional $\log n$ factor that would arise from carrying out the same refinement over all levels.

Finally, if $\ell$ is a high level, we again use a net-based search, but now with a threshold that increases with $\ell$.
In other words, for larger scales we allow larger nets and continue refining further before stopping.
This yields a more accurate representative edge for each high-level pair.
While this increases the work spent on an individual pair, the definition of $\ell^*$ implies that the number of net points---and hence also the number of relevant pairs---decreases geometrically with the level.
Therefore, the total running time over all high levels remains near-linear.
At the same time, the resulting approximation error can still be summed over all high-level pairs and charged to the optimum MST.

\subsection{The Algorithm}
\label{sec:main:algo}

We now describe the algorithm formally.
Throughout the section, we may assume that $\epsilon$ is a power of $2$.\footnote{Otherwise, we can round $\epsilon$ to its nearest power of $2$ by the following procedure:
initialize $K=1$, and repeatedly increment $K$ until $2^{-K} \leq \epsilon$. 
This takes $O(\log(1/\epsilon))$ comparisons and arithmetic operations, without evaluating a logarithm or converting a real number to an integer.}
Let
\[
\gamma := \epsilon^{-1} \log(1/\epsilon).
\]
For each high level $\ell \geq \ell^* + 10 \log(1/\epsilon)$, define
\[
\tau_\ell := 1.1^{\ell - \ell^*}.
\]
Thus, $\gamma$ is the threshold used at middle levels, while $\tau_\ell$ is a level-dependent threshold used at high levels that grows geometrically with $\ell$.
Our algorithm is described in \Cref{alg:MST_improved}.

\begin{algorithm}[!ht]
\caption{Improved $(1 + \epsilon)$-Approximation for MST}
\DontPrintSemicolon
\label{alg:MST_improved}
    \KwIn{finite metric space $(X, \dist)$, 
precision parameter $\epsilon \in (0, 1/2)$}
compute the net hierarchy $\nets = \{N_0, N_1, \dots, N_L\}$ \label{line:net_hierarchy}\;
let $\ell^* \gets \argmax_{0 \leq \ell \leq L} 2^{\ell/2} \cdot |N_\ell|$ \label{line:critical_level}\;
apply \Cref{lemma:WSPC_doubling} with $t = 32$ to compute the collection $\Psi$ of pairs of net points \label{line:WSPC}\;
let $G \gets (X, \emptyset)$ \label{line:G_initial}\;
\For{$(u, v) \in \Psi$}{\label{line:U_V_starts}
        \If{$\max\{\lev(u), \lev(v)\} < \ell^* - 10 \log(1/\epsilon)$}{
            let $(p, q) = (u, v)$ \label{line:low_level}\;
        }
        \ElseIf{$\ell^* - 10 \log(1/\epsilon) \leq \max\{\lev(u), \lev(v)\} \leq \ell^* + 10 \log(1/\epsilon)$}{
            let $0 \leq k \leq \min\{\log(1/\epsilon), \lev(u), \lev(v)\}$ be the largest integer, such that $|\desc_{\lev(u) - k}(u)|, |\desc_{\lev(v) - k}(v)| \leq \gamma$ \label{line:middle_level}\;
compute $(p, q) = \argmin_{p \in \desc_{\lev(u) - k}(u), q \in \desc_{\lev(v) - k}(v)} \dist(p, q)$ \label{line:edge_middle_lev}\;
        }
        \Else{
            let $\ell \gets \max\{\lev(u), \lev(v)\}$ \;
            let $0 \leq k \leq \min\{\log(1/\epsilon), \lev(u), \lev(v)\}$ be the largest integer, such that $|\desc_{\lev(u) - k}(u)|, |\desc_{\lev(v) - k}(v)| \leq \tau_\ell$  \label{line:high_level}\;
compute $(p, q) = \argmin_{p \in \desc_{\lev(u) - k}(u), q \in \desc_{\lev(v) - k}(v)} \dist(p, q)$ \label{line:edge_high_lev}\;
        }
add $(p, q)$ to $G$\; \label{line:U_V_ends}
    }
    \Return $\MST(G)$ \label{line:return}
    
\end{algorithm}

Line \ref{line:WSPC} can be viewed as computing a WSPC for $X$.
Technically, we do not need to compute the WSPC explicitly, but only construct a collection $\Psi$ of their ``centers''.
From Line \ref{line:U_V_starts} to Line \ref{line:U_V_ends}, we handle each pair $(u, v) \in \Psi$ differently depending on which level it lies in.
Call $(u, v) \in \Psi$ a \emph{low-level pair} if their levels satisfy
$\max\{\lev(u), \lev(v)\} < \ell^* - 10 \log(1/\epsilon)$,
a \emph{middle-level pair} if $\max\{\lev(u), \lev(v)\} \in [\ell^* - 10 \log(1/\epsilon), \ell^* + 10 \log(1/\epsilon)]$,
and a \emph{high-level pair} otherwise.
If $(u, v)$ is a low-level pair, we add to $G$ the edge $(u, v)$, which connects the ``centers'' of the corresponding WSPC pair (Line \ref{line:low_level}).
If $(u, v)$ is a middle-level pair, we choose the largest integer $0 \le k \le \log(1/\epsilon)$ such that
both $u$ and $v$ have at most $\gamma$ descendants at $k$ levels below, i.e., 
$|D_{\lev(u) - k}(u)|, |D_{\lev(v) - k}(v)| \le \gamma$ (Line \ref{line:middle_level}), and add the closest pair in $D_{\lev(u) - k}(u) \times D_{\lev(v) - k}(v)$ to $G$ (Line \ref{line:edge_middle_lev}).
The extra condition $k \leq \min\{\lev(u), \lev(v)\}$ in Line \ref{line:middle_level} makes sure the sets $\desc_{\lev(u) - k}(u)$ and $\desc_{\lev(v) - k}(v)$ are well-defined.
If $(u, v)$ is a high-level pair, we repeat the same procedure, except with a carefully chosen threshold $\tau_\ell = 1.1^{\ell - \ell^*}$ that increases with the level (Line \ref{line:high_level}).
Hence, at higher levels we may refine further before stopping, which increases the work per pair but also yields a more accurate representative edge.
The analysis will show that this is affordable because the number of high-level pairs decreases rapidly with $\ell$.

\subsection{Analysis}
\label{sec:main:analysis}
In this section, we show that \Cref{alg:MST_improved} returns a $(1 + O_{\ddim}(\epsilon))$-approximate MST of $X$.
Our main result is the following:

\begin{lemma}\label{lemma:improved_alg_correctness}
    \Cref{alg:MST_improved} returns a $(1 + 2^{O(\ddim)} \epsilon)$-approximate MST.
\end{lemma}

Let $T^*$ denote the MST of $X$.
For a subset $P \subseteq X$, we slightly abuse notation and write
\[
P \cap T^* := \{e \in T^* : e \text{ has at least one endpoint in } P\}.
\]
For every subset $E \subseteq T^*$ of edges, denote the edge weight of $E$ as $w(E) := \sum_{e \in E} w(e)$, where $w(e) = \dist(x, y)$ for $e = (x, y)$.
We will repeatedly use the following lower bound for the local MST weight.
The lemma was proven by~\cite{CzumajEtAl05} for Euclidean spaces, and we extend it to doubling metrics.

\begin{lemma}\label{lemma:net_size_MST}
    Let $T^*$ be the MST of $X$.
    Let $A \subseteq X$ be an arbitrary subset and $N \subseteq A$ be $\rho$-packing for $A$.
    If $|N| \geq 2$,
    then
    $
    w\Big(B(A, \diam(A)) \cap T^*\Big)
    \geq \frac{\rho |N|}{2}
    $.
\end{lemma}

\begin{proof}

    Denote $N = \{x_1, x_2, \dots, x_s\}$ for $s \geq 2$, with $x_{s+1} := x_1$.
    Since $T^*$ is a tree, there is a traversal of $T^*$ that traverses every edge of $T^*$ twice and every point in $X$ at least once.
    Without loss of generality, assume this traversal goes through $N$ in the order $x_1, x_2, \dots, x_s$, and denote the (multiple) path between $x_i$ and $x_{i + 1}$ on that traversal by $P(x_i, x_{i + 1}) \subseteq T^*$.
    We note that $P(x_i, x_{i + 1})$ can traverse an edge for at most twice.
    Then
    \[
    \sum_{i = 1}^s w(P(x_i, x_{i + 1})) = 2w(T^*) = 2\MST(X).
    \]
    Fix $i \in [s]$, and consider the following two cases for $P(x_i, x_{i + 1})$.

    If every edge in $P(x_i, x_{i + 1})$ has at least one endpoint in $B(A, \diam(A))$, then by definition we have 
    $
    P(x_i, x_{i + 1}) = P(x_i, x_{i + 1}) \cap B(A, \diam(A))$.
    By the triangle inequality, 
    \[
    w\Big(P(x_i, x_{i + 1}) \cap B(A, \diam(A))\Big)
    = w(P(x_i, x_{i + 1}))
    \geq \dist(x_i, x_{i + 1}).
    \]

    If there exists an edge in $P(x_i, x_{i + 1})$ for which neither endpoint lies in $B(A, \diam(A))$, let $(p, q)$ be the first such edge.
    Namely, there exist edges $(x_i, p_1), (p_1, p_2), \dots, (p_{k-1}, p_k), (p_k, p) \in P(x_i, x_{i + 1}) \cap B(A, \diam(A))$, but $(p, q) \notin P(x_i, x_{i + 1}) \cap B(A, \diam(A))$.
    Then 
    \begin{align*}
        &\qquad w\Big(P(x_i, x_{i + 1}) \cap B(A, \diam(A))\Big) \\
        &\geq \dist(x_i, p_1) + \sum_{j = 1}^{k - 1} \dist(p_j, p_{j + 1}) + \dist(p_k, p) \\
        &\geq \dist(x_i, p) &&\text{by triangle inequality}\\
        &\geq \diam(A) &&\text{since } x_i \in A \text{ and } p \notin B(A, \diam(A))\\
        &\geq \dist(x_i, x_{i + 1}) &&\text{since } x_i, x_{i + 1} \in A.
    \end{align*}
    
    In conclusion, for every $i \in [s]$ we have $ w\Big(P(x_i, x_{i + 1}) \cap B(A, \diam(A))\Big) \geq \dist(x_i, x_{i + 1})$.
    Hence,
    \begin{align*}
        2 w\Big(B(A, \diam(A)) \cap T^*\Big)
        = \sum_{i = 1}^s w\Big(P(x_i, x_{i + 1}) \cap B(A, \diam(A))\Big) 
        \geq \sum_{i = 1}^s \dist(x_i, x_{i + 1}) 
        \geq s \cdot \rho,
    \end{align*}
    where the last inequality holds since $N$ is $\rho$-packing for $A$.
    Rearranging concludes the proof.
\end{proof}

We are now ready to prove \Cref{lemma:improved_alg_correctness}.

\begin{proof}[Proof of \Cref{lemma:improved_alg_correctness}]
    Let $\Psi$ be the collection of pairs of net points constructed in Line \ref{line:WSPC}.
    Denote $\Phi := \{(\wspcball{u}, \wspcball{v}) \colon (u, v) \in \Psi\}$.
    Then $\Phi$ is a $32$-WSPC of $X$ by \Cref{lemma:WSPC_doubling}.
    For each pair $(u, v) \in \Psi$, 
    assume $(p, q)$ is the edge that \Cref{alg:MST_improved} adds to $G$ in Line \ref{line:U_V_ends}. 
    Ideally, $\dist(p, q)$ is a $(1 + \epsilon)$-approximation of the distance between $\wspcball{u}$ and $\wspcball{v}$;
    otherwise, we consider bounding the difference between $\dist(p, q)$ and $\dist(\wspcball{u}, \wspcball{v})$.
    This motivates us to define the error incurred by the pair $(u, v)$ as 
    \begin{align*}
        \err(u, v) := \max\left\{0, \dist(p, q) - (1 + \epsilon) \dist\left(\wspcball{u}, \wspcball{v}\right)\right\}.
    \end{align*}
Our goal is to show that 
    \begin{equation}\label{eqn:err_bound_goal}
        \sum_{(u, v) \in \Psi} \err(u, v) \leq 2^{O(\ddim)} \epsilon \cdot \MST(X).
    \end{equation}

    Let us first assume that \eqref{eqn:err_bound_goal} holds, and see how it implies \Cref{lemma:improved_alg_correctness}.
    For the $i$-th edge $e_i^*$ of $T^*$, let $(u_i, v_i) \in \Psi$ be an arbitrary pair such that $e_i^* \in \wspcball{u_i} \times \wspcball{v_i}$,
    and assume \Cref{alg:MST_improved} adds $(p_i, q_i)$ to $G$ in Line \ref{line:U_V_ends} for $(u_i, v_i)$.

    We first show that $\{(p_i, q_i)\}_{i = 1}^{n - 1}$ is a spanning tree of $X$ (and a subgraph of $G$).
    We cannot directly apply \Cref{lemma:WSPC_MST} \eqref{it:WSPC_spanning} to $\{(p_i, q_i)\}_{i = 1}^{n - 1}$ and $\Phi$, as $(p_i, q_i)$ is not necessarily contained in $\wspcball{u_i} \times \wspcball{v_i}$.
    The idea is to enlarge the radius of each pair $(\wspcball{u}, \wspcball{v}) \in \Phi$ by a constant factor, so that $(p_i, q_i)$ is contained in the expanded pair, and all expanded pairs still form an $\Omega(1)$-WSPC.
    Specifically, for each pair $(u, v) \in \Psi$, define the expanded pairs $\Phi'$ as follows:
    If $\lev(u) = \lev(v) = 0$, then add $(\{u\}, \{v\})$ to $\Phi'$;
    otherwise, add the pair $(B(u, 2^{\lev(u) + 2}), B(v, 2^{\lev(v) + 2}))$ to $\Phi'$.
    We claim that $\Phi'$ is a $7$-WSPC of $X$.
    Indeed, for each pair of points $(x, y) \in \binom{X}{2}$, assume $(x, y) \in \wspcball{u} \times \wspcball{v}$ for some $(u, v) \in \Psi$.
    Then $(x, y) \in B(u, 2^{\lev(u) + 2}) \times B(v, 2^{\lev(v) + 2})$.
    Hence, \Cref{def:WSPC} \eqref{it:WSPC_covering} is satisfied.
    Regarding \Cref{def:WSPC} \eqref{it:WSPC_separation}, for each $(u, v) \in \Psi$ where $u, v$ are not both at level-$0$,
    by \Cref{lemma:WSPC_doubling} \eqref{it:separation} we have $\dist(u, v) > 66 \max\{2^{\lev(u)}, 2^{\lev(v)}\}$.
    Thus, \begin{align*}
        \dist(B(u, 2^{\lev(u) + 2}), B(v, 2^{\lev(v) + 2})) \geq \dist(u, v) - 2^{\lev(u) + 2} - 2^{\lev(v) + 2}
        > 7 \max\{2^{\lev(u) + 3}, 2^{\lev(v) + 3}\},
    \end{align*}
    indicating each $(B(u, 2^{\lev(u) + 2}), B(v, 2^{\lev(v) + 2})) \in \Phi'$ is $7$-separated.
    Therefore, $\Phi'$ is a $7$-WSPC of $X$.
    By the construction of $(p_i, q_i)$ in Lines \ref{line:low_level}, \ref{line:edge_middle_lev} and \ref{line:edge_high_lev} and \Cref{lemma:descendants_packing_covering} \eqref{it:desc_packing}, we have $(p_i, q_i) \in B(u, 2^{\lev(u) + 2}) \times B(v, 2^{\lev(v) + 2})$.
    Applying \Cref{lemma:WSPC_MST} \eqref{it:WSPC_spanning} to $\{(p_i, q_i)\}_{i = 1}^{n - 1}$ and $\Phi'$, we conclude that $\{(p_i, q_i)\}_{i = 1}^{n - 1}$ is a spanning tree of $G$. 
    Therefore,
    \begin{align*}
        \MST(G) 
        &\leq \sum_{i = 1}^{n - 1} \dist(p_i, q_i) 
        \leq \sum_{i = 1}^{n - 1} \Big((1 + \epsilon) \dist(\wspcball{u_i}, \wspcball{v_i}) + \err(u_i, v_i) \Big) \\
        &\leq \sum_{i = 1}^{n - 1} (1 + \epsilon) w(e_i^*) + \sum_{(u, v) \in \Psi} \err(u, v) 
        \leq (1 + \epsilon) \MST(X) + 2^{O(\ddim)} \epsilon \cdot \MST(X) \\
        &= (1 + 2^{O(\ddim)} \epsilon) \MST(X),
    \end{align*}
    proving \Cref{lemma:improved_alg_correctness}.

In our following analysis we focus on the proof of \eqref{eqn:err_bound_goal}.
    In particular, we upper bound the error for each of the three different types of levels.
    Recall that 
    $(u, v) \in \Psi$ is called a low-level (resp. middle-level, high-level) pair if 
    $\max\{\lev(u), \lev(v)\} < \ell^* - 10 \log(1/\epsilon)$ (resp. $\in [\ell^* - 10 \log(1/\epsilon), \ell^* + 10 \log(1/\epsilon)]$, $> \ell^* + 10 \log(1/\epsilon)$).

    \paragraph{Low levels.}
    Assume $(u, v) \in \Psi$ is a low-level pair, and wlog that $\lev(u) = \max\{\lev(u), \lev(v)\} = \ell < \ell^* - 10 \log(1/\epsilon)$.
    We trivially bound the error of $(u, v)$ by $
        \err(u, v) \leq 2^{\lev(u)} + 2^{\lev(v)} \leq 2^{\ell + 1}.
    $
    Thus the total error of low level pairs can be bounded by 
\begin{align}
    \mathrm{Err}_{\mathrm{low}}
    &\leq \sum_{\ell < \ell^* - 10 \log(1/\epsilon)} \sum_{u \in N_\ell} 2^{\ell + 1} \cdot |\{v \colon (u, v) \in \Psi\}| \notag
    \ineqref{ineq:degree_of_U}{I1} 2^{O(\ddim)} \sum_{\ell < \ell^* - 10 \log(1/\epsilon)} 2^\ell |N_\ell| \notag\\
    &= 2^{O(\ddim)} \sum_{\ell < \ell^* - 10 \log(1/\epsilon)} 2^{\ell / 2} \cdot 2^{\ell/2} |N_\ell| \notag
    \ineqref{ineq:by_def_of_lstar}{I2} 2^{O(\ddim)} \sum_{\ell < \ell^* - 10 \log(1/\epsilon)} 2^{\ell / 2} \cdot 2^{\ell^*/2} |N_{\ell^*}| \notag\\
    &\ineqref{ineq:geometric_series}{I3} 2^{O(\ddim)} \epsilon^5 2^{\ell^*} |N_{\ell^*}| \notag
    \ineqref{ineq:net_size_times_density}{I4} 2^{O(\ddim)} \epsilon^5 \MST(X). \label{eqn:Err_low}
\end{align}
In (I1), we use \eqref{it:bounded_deg} of \Cref{lemma:WSPC_doubling} to get $|\{v \colon (u, v) \in \Psi\}| \leq 2^{O(\ddim)}$.
    In (I2) we use the definition of $\ell^*$ in \eqref{eqn:def_l*}.
In (I3), we use that
\[
\sum_{\ell < \ell^* - 10 \log(1/\epsilon)} 2^{\ell/2}
= O\!\left(2^{(\ell^* - 10 \log(1/\epsilon))/2}\right)
= O\!\left(\epsilon^5 2^{\ell^*/2}\right),
\]
because the sum is geometric.
If $|N_{\ell^*}|\geq2$, (I4) follows from \Cref{lemma:net_size_MST} with $A=X$ and $N=N_{\ell^*}$. If $|N_{\ell^*}|=1$, use $2^{\ell^*}|N_{\ell^*}|\leq2^L<64\diam(X)\leq64\MST(X)$ instead.

    \paragraph{Middle levels.}
    Assume $(u, v) \in \Psi$ is a middle-level pair.
    Assume wlog that $\lev(u) = \max\{\lev(u), \lev(v)\} = \ell$, where $\ell^* - 10 \log(1/\epsilon) \leq \ell \leq \ell^* + 10 \log(1/\epsilon)$.
    \sloppy
    Recall that $p \in \desc_{\lev(u) - k}(u)$ and $q \in \desc_{\lev(v) - k}(v)$, where the set of descendants $\desc_{\lev(u) - k}(u)$ and $\desc_{\lev(v) - k}(v)$ are either dense enough (i.e., $k = \log(1/\epsilon)$) or large enough (i.e., $k < \log(1/\epsilon)$ and $\max\{|\desc_{\lev(u) - k - 1}(u)|, |\desc_{\lev(v) - k - 1}(v)|\} > \gamma$).\footnote{Strictly speaking, there is a third case that $k = \min\{\lev(u), \lev(v)\}$.
    However, in this case we can argue by \Cref{lemma:descendants_packing_covering} \eqref{it:desc_covering} that $\desc_{\lev(u) - k}(u) \supseteq B(u, 2^{\lev(u)})$ and $\desc_{\lev(v) - k}(v) \supseteq B(v, 2^{\lev(v)})$.
    The subsequent analysis can be merged with the dense case $k = \log(1/\epsilon)$.}

    If $k= \log(1/\epsilon)$, then
    by \Cref{lemma:descendants_packing_covering} \eqref{it:desc_covering}, every point of $\wspcball{u}$ and $\wspcball{v}$ is within distance $(\epsilon / 4) 2^{\lev(u)}$ and $(\epsilon / 4) 2^{\lev(v)}$ from a net point in $\desc_{\lev(u) - k}(u)$ and $\desc_{\lev(v) - k}(v)$, respectively.
    Let $(x^*, y^*) := \argmin_{x \in \wspcball{u}, y \in \wspcball{v}} \dist(x, y),$ i.e., $\dist(x^*, y^*)$ realizes $\dist(\wspcball{u}, \wspcball{v})$.
    Let $p^* \in \desc_{\lev(u) - k}(u)$ such that $\dist(x^*, p^*) \leq (\epsilon/4) 2^{\lev(u)}$ and $q^* \in \desc_{\lev(v) - k}(v)$ such that $\dist(y^*, q^*) \leq (\epsilon/4) 2^{\lev(v)}$.
    Then 
    \begin{align*}
        \dist(p, q) 
        &\leq \dist(p^*, q^*) && \text{by definition of } p, q \text{ in Line \ref{line:edge_middle_lev}}\\
        &\leq \dist(x^*, y^*) + \dist(x^*, p^*) + \dist(y^*, q^*) && \text{by triangle inequality}\\
        &\leq \dist(\wspcball{u}, \wspcball{v}) + (\epsilon/4) 2^{\lev(u)} + (\epsilon/4) 2^{\lev(v)} \\
        &\leq (1 + \epsilon) \dist(\wspcball{u}, \wspcball{v}) && \text{by \Cref{lemma:WSPC_doubling} \eqref{it:separation}}.
    \end{align*}
    Therefore, by definition $\err(u, v) = 0.$

   We now focus on the case that $\max\{|\desc_{\lev(u) - k - 1}(u)|, |\desc_{\lev(v) - k - 1}(v)|\} > \gamma$:
   Let $x^*, y^*$ denote the closest pair in $B(u, 2^{\lev(u)}) \times B(v, 2^{\lev(v)})$, and 
   denote $p^*, q^*$ as the points in $\desc_{\lev(u) - k}(u)$ and $\desc_{\lev(v) - k}(v)$ closest to $x^*$ and $y^*$ respectively.
   By \Cref{lemma:descendants_packing_covering} \eqref{it:desc_covering}, $\dist(x^*, p^*) \leq 2^{\lev(u) - k}$ and $\dist(y^*, q^*) \leq 2^{\lev(v) - k}$.
   Similar calculation gives
   \begin{align*}
        \dist(p, q) 
        &\leq \dist(p^*, q^*) && \text{by definition of } p, q \text{ in Line \ref{line:edge_middle_lev}}\\
        &\leq \dist(x^*, y^*) + \dist(x^*, p^*) + \dist(y^*, q^*) && \text{by triangle inequality}\\
        &\leq \dist(\wspcball{u}, \wspcball{v}) + 2^{\lev(u) - k} + 2^{\lev(v) - k}
\end{align*}
    Therefore, \begin{equation}\label{eqn:err_UV_bound}
        \err(u, v) \leq 2^{\lev(u) - k} + 2^{\lev(v) - k}
        \leq 2^{\lev(u) - k + 1}
        \leq 2^{\lev(v) - k + 2},
    \end{equation}
    where the last inequality follows from \eqref{it:bounded_lev_diff} of \Cref{lemma:WSPC_doubling}.

If $|\desc_{\lev(u) - k - 1}(u)| > \gamma$,
then since $\desc_{\lev(u) - k - 1}(u)$ is $2^{\lev(u) - k - 4}$-packing for $B(u, 2^{\lev(u) + 2})$ (\Cref{lemma:descendants_packing_covering} \eqref{it:desc_packing}), by \Cref{lemma:net_size_MST} (applied to $A = B(u, 2^{\lev(u) + 2})$ and $N = \desc_{\lev(u) - k - 1}(u)$) we obtain
\[
2^{\lev(u) - k - 4}
\leq \frac{2\, w(B(u, 2^{\lev(u) + 4}) \cap T^*)}{|\desc_{\lev(u) - k - 1}(u)|}
\leq \frac{2}{\gamma} w(B(u, 2^{\lev(u) + 4}) \cap T^*).
\]
If $|\desc_{\lev(v) - k - 1}(v)| > \gamma$, analogously we have 
\[
2^{\lev(v) - k - 4}
\leq \frac{2}{\gamma} w(B(v, 2^{\lev(v) + 4}) \cap T^*).
\]
We therefore conclude that 
\[
\err(u, v) 
\leq 2^{\lev(v) - k + 2} 
\leq
\frac{128}{\gamma} \Big(
    w(B(u, 2^{\lev(u) + 4}) \cap T^*) + w(B(v, 2^{\lev(v) + 4}) \cap T^*)
\Big).
\]
   
Thus, the total error incurred by all middle level pairs is bounded by 
    \begin{align}
        \mathrm{Err}_{\mathrm{mid}}
        &\leq \sum_{\substack{(u, v) \in \Psi \text{ is a} \\ \text{middle-level pair}}}
        \frac{128}{\gamma} \Big(
        w(B(u, 2^{\lev(u) + 4}) \cap T^*) + w(B(v, 2^{\lev(v) + 4}) \cap T^*)
        \Big) \notag \\
        &\leq \sum_{\ell = \ell^* - 10 \log(1/\epsilon) - 1}^{\ell^* + 10 \log(1/\epsilon)} 
        \sum_{u \in N_\ell}
        \frac{128}{\gamma} w(B(u, 2^{\ell + 4}) \cap T^*)
        \cdot |\{v \colon (u, v) \in \Psi\}| \notag \\
        &\leq \frac{2^{O(\ddim)}}{\gamma} 
        \sum_{\ell = \ell^* - 10 \log(1/\epsilon) - 1}^{\ell^* + 10 \log(1/\epsilon)} 
        \sum_{u \in N_\ell}
        \sum_{e \in B(u, 2^{\ell + 4}) \cap T^*} w(e) \label{eqn:num_pairs_contain_cluster}\\
        &= \frac{2^{O(\ddim)}}{\gamma} 
        \sum_{e \in T^*} w(e) \sum_{\ell = \ell^* - 10 \log(1/\epsilon) - 1}^{\ell^* + 10 \log(1/\epsilon)} |\{u \in N_\ell \colon e \in B(u, 2^{\ell + 4})\}| \notag\\
        &\leq \frac{2^{O(\ddim)}}{\gamma} 20 \log(1/\epsilon) \sum_{e \in T^*} w(e) \label{eqn:num_clusters_contain_edge}\\
        &\leq 2^{O(\ddim)} \epsilon\cdot \MST(X). \notag
    \end{align}
    In \eqref{eqn:num_pairs_contain_cluster}, we use $|\{v \colon (u, v) \in \Psi\}| \leq 2^{O(\ddim)}$ by \eqref{it:bounded_deg} of \Cref{lemma:WSPC_doubling}.
    In \eqref{eqn:num_clusters_contain_edge}, we use
    that for every edge $e$ and every level $\ell$, $|\{u \in N_\ell \colon e \in B(u, 2^{\ell + 4})\}| \leq 2^{O(\ddim)}$.
    We explain the reason below:
    Denote $e = (x, y)$;
    if $u \in N_\ell$ satisfies $e \in B(u, 2^{\ell + 4})$, then either $\dist(u, x) \leq 2^{\ell + 4}$ or $\dist(u, y) \leq 2^{\ell + 4}$.
    Equivalently, $u \in N_\ell \cap (B(x, 2^{\ell + 4}) \cup B(y, 2^{\ell + 4}))$,
    so by 
    \Cref{lemma:packing} the number of such $u$ is bounded by $2^{O(\ddim)}$.
    The last inequality follows from our choice $\gamma = \epsilon^{-1} \log(1/\epsilon)$.

    \paragraph{High levels.}
    Assume $(u, v) \in \Psi$ is a high-level pair.
    Assume wlog that $\lev(u) = \max\{\lev(u), \lev(v)\} = \ell$, where $\ell > \ell^* + 10 \log(1/\epsilon)$.
    Analogous to the middle levels, we focus on the case that $\max\{|\desc_{\lev(u) - k - 1}(u)|, |\desc_{\lev(v) - k - 1}(v)|\}  > \tau_\ell$ (as $\err(u, v) = 0$ in the other case).
    By \eqref{eqn:err_UV_bound}, 
    $\err(u, v) \leq 2^{\lev(v) - k + 2}$.

    If $|\desc_{\lev(u) - k - 1}(u)|  > \tau_\ell$, then by \Cref{lemma:descendants_packing_covering,lemma:net_size_MST} we have 
    \[
    2^{\lev(u) - k - 4}
    \leq \frac{2}{\tau_\ell} w(B(u, 2^{\lev(u) + 4}) \cap T^*).
    \]
    Similarly, if $|\desc_{\lev(v) - k - 1}(v)|  > \tau_\ell$, then 
    \[
    2^{\lev(v) - k - 4}
    \leq \frac{2}{\tau_\ell} w(B(v, 2^{\lev(v) + 4}) \cap T^*)
    \leq \frac{2}{\tau_{\lev(v)}} w(B(v, 2^{\lev(v) + 4}) \cap T^*).
    \]
    In conclusion, 
    \begin{align*}
        \err(u, v) 
        \leq 
            \frac{128}{\tau_{\lev(u)}} w(B(u, 2^{\lev(u) + 4}) \cap T^*)
            + \frac{128}{\tau_{\lev(v)}} w(B(v, 2^{\lev(v) + 4}) \cap T^*)
    \end{align*}

Therefore, the total error incurred by high levels can be bounded by 
    \begin{align*}
        \mathrm{Err}_{\mathrm{high}}
        &\leq \sum_{\substack{(u, v) \in \Psi \text{ is a} \\ \text{high-level pair}}} 
        \left(
            \frac{128}{\tau_{\lev(u)}} w(B(u, 2^{\lev(u) + 4}) \cap T^*)
            + \frac{128}{\tau_{\lev(v)}} w(B(v, 2^{\lev(v) + 4}) \cap T^*)
        \right)\\
        &\leq \sum_{\ell \geq \ell^* + 10 \log(1/\epsilon)}
        \sum_{u \in N_\ell}
        \frac{128}{\tau_\ell} w(B(u, 2^{\ell + 4}) \cap T^*)
        \cdot |\{v \colon (u, v) \in \Psi\}| \\
        &\leq 2^{O(\ddim)} \sum_{\ell \geq \ell^* + 10 \log(1/\epsilon)} \frac{1}{\tau_\ell}
        \sum_{u \in N_\ell} \sum_{e \in B(u, 2^{\ell + 4}) \cap T^*} w(e) \\
        &= 2^{O(\ddim)}
        \sum_{e \in T^*} w(e) \sum_{\ell \geq \ell^* + 10 \log(1/\epsilon)} \frac{1}{\tau_\ell} |\{u \in N_\ell \colon e \in B(u, 2^{\ell + 4})\}| \\
        &\leq 2^{O(\ddim)} 
        \sum_{e \in T^*} w(e) \sum_{\ell \geq \ell^* + 10 \log(1/\epsilon)} \frac{1}{\tau_\ell} \\
        &\leq 2^{O(\ddim)} 
        \sum_{e \in T^*} w(e) \sum_{j = 10 \log(1/\epsilon)}^\infty 1.1^{-j} \\
        &\leq 2^{O(\ddim)} \epsilon \cdot \MST(X).
    \end{align*}
    For the last inequality, we use that
\[
\sum_{j = 10 \log(1/\epsilon)}^\infty 1.1^{-j}
= O\!\left(1.1^{-10 \log(1/\epsilon)}\right)
\leq O(\epsilon).
\]

\paragraph{Wrapping up.}
    The total error is bounded by 
    \begin{align*}
        \mathrm{Err}_{\mathrm{low}}
        + \mathrm{Err}_{\mathrm{mid}}
        + \mathrm{Err}_{\mathrm{high}}
        \leq 2^{O(\ddim)} \epsilon \MST(X),
    \end{align*}
    concluding the proof of \eqref{eqn:err_bound_goal}.
\end{proof}

\subsection{Time Complexity}
\label{sec:main:time}
\begin{lemma}\label{lemma:improved_alg_time}
    \Cref{alg:MST_improved} runs in time $2^{O(\ddim)} n (\log n + \epsilon^{-1} \log(1/\epsilon))$.
\end{lemma}

We first give a brief overview of our time complexity analysis:
It turns out that the main bottleneck is Line \ref{line:edge_middle_lev}, where we compute the bichromatic closest pair $(p, q) \in \desc_{\lev(u) - k}(u) \times \desc_{\lev(v) - k}(v)$ for every middle level pair $(u, v) \in \Psi$.
Since $|\desc_{\lev(u) - k}(u)|, \  |\desc_{\lev(v) - k}(v)| \leq \gamma$, the brute force algorithm gives $\gamma^2$ time complexity per pair $(u, v)$, and $\gamma^2 |\Psi| = 2^{O(\ddim)} n \epsilon^{-2} \log^2(1/\epsilon)$ in total.
This already matches the time complexity in~\cite{DBLP:conf/soda/AryaM16}.

To further improve the dependence on $\epsilon$, we do a more refined analysis:
Instead of using the na\"ive bound 
\begin{align*}
    |\desc_{\lev(u) - k}(u)| \cdot |\desc_{\lev(v) - k}(v)| \leq \gamma^2,
\end{align*}
we utilize the geometric properties that $\desc_{\lev(u) - k}(u) \subseteq B(u, 2^{\lev(u) + 2})$ and $\desc_{\lev(v) - k}(v) \subseteq B(v, 2^{\lev(v) + 2})$ (\Cref{lemma:descendants_packing_covering} \eqref{it:desc_packing}), and bound 
\begin{align*}
    |\desc_{\lev(u) - k}(u)| \cdot |\desc_{\lev(v) - k}(v)|
    \leq \min\{|B(u, 2^{\lev(u) + 2})|, \gamma\}
    \cdot \min\{|B(v, 2^{\lev(v) + 2})|, \gamma\}.
\end{align*}
Thus, the total time complexity of Line \ref{line:edge_middle_lev} can be bounded by 
\begin{align*}
    \sum_{\substack{(u, v) \in \Psi \text{ is a} \\ \text{middle-level pair}}}
    \min\{|B(u, 2^{\lev(u) + 2})|, \gamma\}
    \cdot \min\{|B(v, 2^{\lev(v) + 2})|, \gamma\}.
\end{align*}

Our key technical ingredient is the following lemma, which directly upper bounds the summation above by  $2^{O(\ddim)} n \gamma = 2^{O(\ddim)} n \epsilon^{-1} \log(1/\epsilon)$, thus reducing the time complexity of~\cite{DBLP:conf/soda/AryaM16} by a factor of $\epsilon^{-1} \log(1/\epsilon)$.

\begin{lemma}[Middle-level time complexity]
    \label{lemma:middle_level_time}
    Let $\Gamma > 1$ and let $\Psi$ be constructed from \Cref{lemma:WSPC_doubling} with parameter $t \geq 32$. Then 
    \begin{align*}
        \sum_{(u, v) \in \Psi} \min\{|B(u, 2^{\lev(u) + 2})|, \Gamma\}
        \cdot \min\{|B(v, 2^{\lev(v) + 2})|, \Gamma\}
        \leq t^{O(\ddim)} n \Gamma.
    \end{align*}
\end{lemma}

We are now ready to prove \Cref{lemma:improved_alg_time}.
The proof of \Cref{lemma:middle_level_time} will be deferred to the end of this section.

\begin{proof}[Proof of \Cref{lemma:improved_alg_time}]
    We handle different stages of \Cref{alg:MST_improved} separately.
    Throughout the proof, we assume wlog that the minimum interpoint distance of $X$ is greater than $1$ and the diameter of $X$ is $\Delta \leq O(n^{10}/ \epsilon)$.
    See \Cref{sec:spread_and_ddim}.

    \paragraph{Preprocessing.}
    By \Cref{lemma:net_hierarchy_construction}, $\nets$ can be constructed in $2^{O(\ddim)} n \log(n/\epsilon)$ time in Line \ref{line:net_hierarchy}.
    Finding the level $\ell^*$ takes $O(n \log(n/\epsilon))$ time in Line \ref{line:critical_level}.
    In Line \ref{line:WSPC}, $\Psi$ can be constructed in $2^{O(\ddim)} n \log(n/\epsilon)$ time by \Cref{lemma:WSPC_doubling}.
    Hence, 
    the preprocessing stage (Lines \ref{line:net_hierarchy}-\ref{line:G_initial}) takes $2^{O(\ddim)} n \log(n/\epsilon)$ time.
    Below we focus on the time complexity of constructing $G$ (Lines \ref{line:U_V_starts}-\ref{line:U_V_ends}).
    We bound the time complexity of the three types of levels respectively.

    \paragraph{Low levels.}
    For every low-level pair $(u, v)$, processing $(u, v)$ takes $O(1)$ time.
    Therefore, the total time complexity of low levels is $O(|\Psi|) \leq 2^{O(\ddim)} n$.

    \paragraph{Middle levels.}
Let $(u, v) \in \Psi$ be a middle-level pair.
    To find $k$ in Line \ref{line:middle_level}, the algorithm maintains $\desc_{\lev(u) - j}(u)$ for each $1 \leq j \leq \log(1/\epsilon)$ --- 
    given $\desc_{\lev(u) - j + 1}(u)$, construct 
    $\desc_{\lev(u) - j}(u) := \bigcup_{x \in \desc_{\lev(u) - j + 1}(u)} \child(x)$, until $|\desc_{\lev(u) - j}(u)| > \gamma$.
    By \Cref{lemma:num_children_parent} \eqref{it:child_num}, $|\child(x)| \leq 2^{O(\ddim)}$, and the net hierarchy data structure stores pointers from $x$ to all its children.
    Hence, each $\desc_{\lev(u) - j}(u)$ can be constructed in $\sum_{x \in \desc_{\lev(u) - j + 1}(u)} |\child(x)| \leq  2^{O(\ddim)} |\desc_{\lev(u) - j + 1}(u)|$ time, and the time complexity of finding $k$ (Line \ref{line:middle_level}) is at most 
    \[
    \sum_{j = 1}^{\log(1/\epsilon)} 2^{O(\ddim)} \Big(|\desc_{\lev(u) - j + 1}(u)| + |\desc_{\lev(v) - j + 1}(v)| \Big).
    \]
    Summing over all middle-level pairs, we get 
    \begin{align}
        &\qquad 2^{O(\ddim)} \sum_{\substack{(u, v) \in \Psi \text{ is a} \\ \text{middle-level pair}}}
        \sum_{j = 1}^{\log(1/\epsilon)}
        \Big(|\desc_{\lev(u) - j + 1}(u)| + |\desc_{\lev(v) - j + 1}(v)|\Big) \notag \\
        &\leq 2^{O(\ddim)} \sum_{\ell = \ell^* - 10 \log(1/\epsilon) - 1}^{\ell^* + 10 \log(1/\epsilon)} 
        \sum_{u \in N_\ell}
        \sum_{j = 1}^{\log(1/\epsilon)}
        |\desc_{\ell - j + 1}(u)| \notag \\
        &= 2^{O(\ddim)} \sum_{\ell = \ell^* - 10 \log(1/\epsilon) - 1}^{\ell^* + 10 \log(1/\epsilon)} 
        \sum_{j = 1}^{\log(1/\epsilon)}
        \sum_{v \in N_{\ell - j + 1}} 
        |\{u \in N_\ell \colon v \in \desc_{\ell - j + 1}(u)\}| \notag \\
        &\leq 2^{O(\ddim)} n \log^2(1/\epsilon),
        \label{eqn:middle_lev_time}
    \end{align}
    where the last inequality follows from \Cref{lemma:num_children_parent} \eqref{it:ancestor_num}.
    This upper bounds the total time complexity of finding $k$ (Line \ref{line:middle_level}) over all middle level pairs.

    For each middle-level pair $(u, v) \in \Psi$,
    the time complexity of finding the closest pair $(p, q) \in \desc_{\lev(u) - k}(u) \times \desc_{\lev(v) - k}(v)$ in Line \ref{line:edge_middle_lev} is $|\desc_{\lev(u) - k}(u)| \cdot |\desc_{\lev(v) - k}(v)|$ by brute force.
    Recall that $|\desc_{\lev(u) - k}(u)|, |\desc_{\lev(v) - k}(v)| \leq \gamma$ and that $\desc_{\lev(u) - k}(u) \subseteq B(u, 2^{\lev(u) + 2}), \desc_{\lev(v) - k}(v) \subseteq B(v, 2^{\lev(v) + 2})$.
    Therefore,
    \begin{align*}
        |\desc_{\lev(u) - k}(u)| \cdot |\desc_{\lev(v) - k}(v)|
        \leq \min\{|B(u, 2^{\lev(u) + 2})|, \gamma\}
        \cdot \min\{|B(v, 2^{\lev(v) + 2})|, \gamma\}.
    \end{align*}
    Summing over all middle-level pairs $(u, v)$, we upper bound the total time complexity of Line \ref{line:edge_middle_lev} by
    \begin{align*}
        &\qquad \sum_{\substack{(u, v) \in \Psi \text{ is a} \\ \text{middle-level pair}}}
        \min\{|B(u, 2^{\lev(u) + 2})|, \gamma\}
        \cdot \min\{|B(v, 2^{\lev(v) + 2})|, \gamma\}.
    \end{align*}
    Taking $t = 32$ and $\Gamma = \gamma = \epsilon^{-1} \log(1/\epsilon)$ in \Cref{lemma:middle_level_time}, we upper bound the summation above by
    \begin{align*}
        2^{O(\ddim)} n \gamma = 2^{O(\ddim)} n \epsilon^{-1} \log(1/\epsilon).
    \end{align*}
    Combined with \eqref{eqn:middle_lev_time}, this also upper bounds the total time complexity of middle levels.

    \paragraph{High levels.}
    Let $(u, v)$ be a high-level pair, and assume that
    $\lev(u) = \max\{\lev(u), \lev(v)\} = \ell$, where $\ell > \ell^* + 10 \log(1/\epsilon)$.
    Analogous to the middle levels,
    the time complexity of finding $k$ in Line \ref{line:high_level} is at most 
    \begin{align*}
        \sum_{j = 1}^{k + 1} 2^{O(\ddim)} |\desc_{\lev(u) - j + 1}(u)|
        + \sum_{j = 1}^{k + 1} 2^{O(\ddim)} |\desc_{\lev(v) - j + 1}(v)|
        \leq 2^{O(\ddim)} \tau_\ell \log(1/\epsilon),
    \end{align*}
    where we use that $|\desc_{\lev(u) - j + 1}(u)|, |\desc_{\lev(v) - j + 1}(v)| \leq \tau_\ell$ for $1 \leq j \leq k+1$.
    
    Regarding the time complexity of
    computing the closest pair $(p, q) \in \desc_{\lev(u) - k}(u) \times \desc_{\lev(v) - k}(v)$ in Line \ref{line:edge_high_lev}, we trivially bound it by 
    \[|\desc_{\lev(u) - k}(u)| \cdot |\desc_{\lev(v) - k}(v)| \leq \tau_\ell^2. \]
    Hence, the time complexity of high levels is 
    \begin{align}
        &\quad \sum_{\ell > \ell^* + 10\log(1/\epsilon)}
        \sum_{u \in N_\ell}
        \Big(
            2^{O(\ddim)} \tau_\ell \log(1/\epsilon) 
            + \tau_\ell^2
        \Big)
        \cdot |\{v \colon (u, v) \in \Psi\}| \notag\\
        &\leq 2^{O(\ddim)} \log(1/\epsilon) \sum_{\ell > \ell^* + 10\log(1/\epsilon)} \tau_\ell^2 \cdot |N_\ell| \notag \\
        &= 2^{O(\ddim)} \log(1/\epsilon) \sum_{\ell > \ell^* + 10\log(1/\epsilon)} \frac{\tau_\ell^2}{2^{\ell/2}} \cdot 2^{\ell/2} |N_\ell| \notag\\
        &\leq 2^{O(\ddim)} \log(1/\epsilon) \sum_{\ell > \ell^* + 10 \log(1/\epsilon)} \frac{\tau_\ell^2}{2^{\ell/2}} \cdot 2^{\ell^*/2} |N_{\ell^*}| &&\text{by the definition of }\ell^* \notag \\
        &= 2^{O(\ddim)} \log(1/\epsilon) |N_{\ell^*}| \sum_{\ell > \ell^* + 10 \log(1/\epsilon)} \left(\frac{1.21}{\sqrt{2}}\right)^{\ell - \ell^*} &&\text{by the choice of } \tau_\ell \notag \\
        &\leq 2^{O(\ddim)} n \log(1/\epsilon).
        \label{eqn:edge_high_level_time}
    \end{align}

    Combining the three types of levels above, the total complexity of constructing the graph $G$ (Lines \ref{line:U_V_starts}-\ref{line:U_V_ends}) is bounded by $2^{O(\ddim)} n \epsilon^{-1} \log(1/\epsilon)$.

    \paragraph{Constructing the MST of $G$.}
    Finally, note that $G$ has $n$ vertices and $m = |\Psi| = 2^{O(\ddim)} n$ edges.
    The MST of $G$ can be computed in time $O(m + n \log n) = 2^{O(\ddim)} n + O(n \log n)$, using Prim's algorithm implemented with a Fibonacci Heap~\cite{fredman1987fibonacci}.

    \paragraph{Wrapping up.}
    In conclusion, the total time complexity of \Cref{alg:MST_improved} is $2^{O(\ddim)} n (\log n + \epsilon^{-1} \log(1/\epsilon))$.
\end{proof}

We prove \Cref{lemma:middle_level_time} below.
Our proof will be based on property \eqref{it:few_NN} of \Cref{lemma:WSPC_doubling}, combined with the probabilistic method.

\begin{proof}[Proof of \Cref{lemma:middle_level_time}]
    To bound the summation above, consider the following random sampling process\footnote{This sampling argument is only used for the analysis and does not appear in our algorithm.}:
    Let $S$ be a random subset of $X$ constructed by including each $x \in X$ independently with probability $1/\Gamma$.
    Denote the random variable 
    \begin{align*}
        \xi := |\{(u, v) \in \Psi \colon \dist(u, S) \leq 2^{\lev(u) + 2} \text{ and } \dist(v, S) \leq 2^{\lev(v) + 2}\}|.
    \end{align*}
    By \Cref{lemma:WSPC_doubling} \eqref{it:few_NN}, $\xi \leq t^{O(\ddim)} |S|$, thus 
    \begin{align}\label{eqn:Exi_ub}
        \E [\xi] \leq t^{O(\ddim)}\E [|S|] = \frac{t^{O(\ddim)} n}{\Gamma}.
    \end{align}
    On the other hand, 
    for each point $u$, 
    \begin{align*}
        \Pr[\dist(u, S) \leq 2^{\lev(u) + 2}]
        &= \Pr[S \cap B(u, 2^{\lev(u) + 2}) \neq \emptyset]
        = 1 - \left(1 - \frac{1}{\Gamma}\right)^{|B(u, 2^{\lev(u) + 2})|} \\
        &\geq 1 - \exp\left(-\frac{|B(u, 2^{\lev(u) + 2})|}{\Gamma}\right) \\
        &\geq \frac{1}{2\Gamma} \min\left\{|B(u, 2^{\lev(u) + 2})|, \Gamma\right\},
    \end{align*}
    where we use $1 - \beta \leq e^{-\beta}, \forall \beta$ in the second last inequality and $e^{-\beta} \leq 1 - \frac{1}{2} \min\{\beta, 1\}, \forall \beta \geq 0$ in the last inequality.
    For each pair $(u, v) \in \Psi$, note that $B(u, 2^{\lev(u) + 2}) \cap B(v, 2^{\lev(v) + 2}) = \emptyset$ as long as $u, v$ are not both at level $0$ and the separation parameter in \Cref{lemma:WSPC_doubling} is $t \geq 32$.
    In this case, the events $\{S \cap B(u, 2^{\lev(u) + 2}) \neq \emptyset\}$ and $\{S \cap B(v, 2^{\lev(v) + 2}) \neq \emptyset\}$ are independent.
    Therefore, 
    \begin{align*}
        &\qquad \Pr\left[\dist(u, S) \leq 2^{\lev(u) + 2} \text{ and } \dist(v, S) \leq 2^{\lev(v) + 2}\right] \\
        &= \Pr[S \cap B(u, 2^{\lev(u) + 2}) \neq \emptyset] \cdot \Pr[S \cap B(v, 2^{\lev(v) + 2}) \neq \emptyset] \\
        &\geq \frac{1}{4 \Gamma^2} \min\{|B(u, 2^{\lev(u) + 2})|, \Gamma\}
        \cdot \min\{|B(v, 2^{\lev(v) + 2})|, \Gamma\}.
    \end{align*}
    Summing over $(u, v) \in \Psi$ with $\{\lev(u), \lev(v)\} \neq \{0\}$ yields 
    \begin{align}\label{eqn:Exi_lb}
        \E[\xi] \geq \frac{1}{4 \Gamma^2} \sum_{\substack{(u, v) \in \Psi \\ \{\lev(u), \lev(v)\} \neq \{0\}}} \min\{|B(u, 2^{\lev(u) + 2})|, \Gamma\}
        \cdot \min\{|B(v, 2^{\lev(v) + 2})|, \Gamma\}.
    \end{align}
    Combining \eqref{eqn:Exi_ub} with \eqref{eqn:Exi_lb}, we have
    \begin{align*}
        \sum_{\substack{(u, v) \in \Psi \\ \{\lev(u), \lev(v)\} \neq \{0\}}} \min\{|B(u, 2^{\lev(u) + 2})|, \Gamma\}
        \cdot \min\{|B(v, 2^{\lev(v) + 2})|, \Gamma\}
        \leq t^{O(\ddim)} n \Gamma.
    \end{align*}

    Finally, note that there are at most $|\Psi| \leq t^{O(\ddim)} n$ pairs $(u, v)$ with $\lev(u) = \lev(v) = 0$, and that $|B(u, 2^2)| \leq 2^{O(\ddim)}$ for each $u \in N_0$.
    We hence conclude that
    \begin{align*}
        \sum_{(u, v) \in \Psi} \hspace*{-0.1cm} \min\{|B(u, 2^{\lev(u) + 2})|, \Gamma\}
        \cdot \min\{|B(v, 2^{\lev(v) + 2})|, \Gamma\}
        \leq t^{O(\ddim)} n \Gamma + (2t)^{O(\ddim)} n 
        \leq t^{O(\ddim)} n \Gamma.
    \end{align*}
\end{proof}

\subsection{\texorpdfstring{Proof of \Cref{thm:MST_alg_improved}}{Proof of the main algorithmic theorem}}

\begin{proof}[Proof of \Cref{thm:MST_alg_improved}]
    If $n^2 \leq n \epsilon^{-1} \log(1/\epsilon)$,  the MST of the complete graph with vertex set $X$ can be computed in time $O(n^2)$.
    Assume $n^2 > n \epsilon^{-1} \log(1/\epsilon)$ below.
    By \Cref{sec:spread_and_ddim}, one can 
    compute an $O(1)$-approximation of the doubling dimension $\ddim = \ddim(X)$ in time $2^{O(\ddim)} n \log(n/\epsilon)$.
    Then run \Cref{alg:MST_improved} with parameter $\epsilon' = \epsilon/2^{O(\ddim)}$.
By \Cref{lemma:improved_alg_correctness}, the algorithm returns a $(1 + \epsilon)$-approximate MST.
    By \Cref{lemma:improved_alg_time}, the time complexity is 
    \begin{align*}
        2^{O(\ddim)} n (\log n + \epsilon^{-1} 2^{O(\ddim)} \log(2^{O(\ddim)}/\epsilon))
        = 2^{O(\ddim)} n (\log n + \epsilon^{-1} \log(1/\epsilon)).
    \end{align*}
\end{proof}

\section{Low-Dimensional Euclidean MST}
\label{sec:euclidean}

In this section, we prove \Cref{thm:euclidean_mst}.

\thmeuclideanmst*

Our improvement comes from a faster implementation of Euclidean bichromatic closest pair due to Agarwal et al., as stated below.

\begin{lemma}[{Euclidean bichromatic closest pair~\cite[Corollary 9]{AgarwalES91}}]
    \label{lemma:bichromatic_pair}
    Let $d \in \Z_+$ be a constant.
    There is an algorithm that, given two nonempty subsets $A, B \subset \R^d$ with $|A| = m, |B| = n$, computes the closest pair in $A \times B$, i.e., $(a, b) = \argmin_{a \in A, b \in B} \|a - b\|_2$ in expected time
    \begin{align*}
        O_{d, \kappa}\Big(
        (mn)^{1 - 1/(\lceil d/2 \rceil + 1) + \kappa} + m \log n + n \log m
        \Big),
    \end{align*}
    for any constant $\kappa > 0$.
\end{lemma}

\begin{proof}[Proof of \Cref{thm:euclidean_mst}]
Assume that $n^2 \geq n \epsilon^{-(1 - \beta_d)} \log^{1 - \beta_d}(1/\epsilon)$; otherwise, the MST of the complete graph with vertex set $X$ can be computed in time $O(n^2)$.

We use the same algorithm as \Cref{alg:MST_improved}, with the only difference that we apply \Cref{lemma:bichromatic_pair} to compute closest pairs in Lines \ref{line:edge_middle_lev} and \ref{line:edge_high_lev}.
The error bound thus follows directly from \Cref{lemma:improved_alg_correctness}.
Using the fact that $\ddim(X, \|\cdot\|_2) \leq O(d)$, we obtain an approximation ratio $1 + 2^{O(d)} \epsilon$.

We analyze the time complexity below:
The preprocessing remains unchanged, which runs in time $2^{O(d)} n \log(n/\epsilon)$.
Similarly, the time complexity of low levels remains $O(|\Psi|) \leq 2^{O(d)} n$.

Regarding middle levels, the total time complexity of finding $k$ (Line \ref{line:middle_level}) over all middle level pairs, by \eqref{eqn:middle_lev_time}, is $2^{O(d)} n \log^2(1/\epsilon)$.
The only difference occurs in Line \ref{line:edge_middle_lev}, where we now use \Cref{lemma:bichromatic_pair}.
Fix an arbitrary constant $0 < \kappa < 1/(\lceil d/2 \rceil + 1)$ in \Cref{lemma:bichromatic_pair}, and denote 
\begin{align*}
    \beta_d := \frac{1}{\lceil d/2 \rceil + 1} - \kappa.
\end{align*}
The time complexity of Line \ref{line:edge_middle_lev} for each pair $(u, v) \in \Psi$ is 
\begin{align*}
    &|\desc_{\lev(u) - k}(u)|^{1 - \beta_d} \cdot |\desc_{\lev(v) - k}(v)|^{1 - \beta_d} \\
    &\quad + |\desc_{\lev(u) - k}(u)| \log |\desc_{\lev(v) - k}(v)| \\
    &\quad + |\desc_{\lev(v) - k}(v)| \log |\desc_{\lev(u) - k}(u)| \\
    &\leq 
    \underbrace{|\desc_{\lev(u) - k}(u)|^{1 - \beta_d} \cdot |\desc_{\lev(v) - k}(v)|^{1 - \beta_d}}_{T_1(u, v)}
    + \underbrace{(|\desc_{\lev(u) - k}(u)| + |\desc_{\lev(v) - k}(v)|) \log \gamma}_{T_2(u, v)},
\end{align*}
where we use $|\desc_{\lev(u) - k}(u)|, |\desc_{\lev(v) - k}(v)| \leq \gamma$.
Above, we omit leading factors that depend solely on $d$ and $\kappa$.

Since $\desc_{\lev(u) - k}(u) \subseteq B(u, 2^{\lev(u) + 2})$ and $\desc_{\lev(v) - k}(v) \subseteq B(v, 2^{\lev(v) + 2})$,
the first term $T_1(u, v)$ can be further bounded by
\begin{align*}
    \Big(\min\{|B(u, 2^{\lev(u) + 2})|, \gamma\}
        \cdot \min\{|B(v, 2^{\lev(v) + 2})|, \gamma\}
    \Big)^{1 - \beta_d}.
\end{align*}
Summing over all middle level pairs $(u, v)$, we have
\begin{align*}
    &\qquad \sum_{\substack{(u, v) \in \Psi \text{ is a} \\ \text{middle-level pair}}}
    \Big(\min\{|B(u, 2^{\lev(u) + 2})|, \gamma\}
        \cdot \min\{|B(v, 2^{\lev(v) + 2})|, \gamma\}
    \Big)^{1 - \beta_d} \\
    &\leq |\Psi|^{\beta_d} 
    \Bigg(\sum_{\substack{(u, v) \in \Psi \text{ is a} \\ \text{middle-level pair}}}
    \min\{|B(u, 2^{\lev(u) + 2})|, \gamma\}
        \cdot \min\{|B(v, 2^{\lev(v) + 2})|, \gamma\}
    \Bigg)^{1 - \beta_d} \\
    &\leq (2^{O(d)} n)^{\beta_d} \cdot (2^{O(d)} n \epsilon^{-1} \log(1/\epsilon))^{1 - \beta_d} \\
    &\leq O_{d, \kappa}\Big(n \epsilon^{-(1 - \beta_d)} \log^{1 - \beta_d}(1/\epsilon)\Big),
\end{align*}
where the first inequality is H\"older,
and the second last inequality follows from \Cref{lemma:WSPC_doubling} \eqref{it:few_NN} and \Cref{lemma:middle_level_time}.

Regarding the second term $T_2(u, v)$, we have
\begin{align*}
    &\qquad \sum_{\substack{(u, v) \in \Psi \text{ is a} \\ \text{middle-level pair}}}
    (|\desc_{\lev(u) - k}(u)| + |\desc_{\lev(v) - k}(v)|) \log \gamma \\
    &\leq \log \gamma \sum_{\ell = \ell^* - 10 \log(1/\epsilon) - 1}^{\ell^* + 10 \log(1/\epsilon)}
    \sum_{u \in N_\ell} |B(u, 2^{\ell + 2})| \cdot |\{v \colon (u, v) \in \Psi\}| \\
    &\leq 2^{O(d)} \log \gamma \sum_{\ell = \ell^* - 10 \log(1/\epsilon) - 1}^{\ell^* + 10 \log(1/\epsilon)} \sum_{x \in X} |\{u \in N_\ell \colon \|u - x\|_2 \leq 2^{\ell + 2}\}| \\
    &\leq 2^{O(d)} n \log \gamma \log(1/\epsilon) \\
    &\leq O_d\Big(n \log^2(1/\epsilon)\Big).
\end{align*}

Combining the bounds on $T_1$ and $T_2$, the total time complexity of Line \ref{line:edge_middle_lev} over all middle level pairs is 
\begin{align*}
    O_{d, \kappa}\Big(n \epsilon^{-(1 - \beta_d)} \log^{1 - \beta_d}(1/\epsilon)
    + n \log^2(1/\epsilon)\Big)
    \leq O_{d, \kappa}\Big(n \epsilon^{-(1 - \beta_d)} \log^{1 - \beta_d}(1/\epsilon)\Big).
\end{align*}

For high level pairs, by \eqref{eqn:edge_high_level_time}, the time complexity is bounded by $O_d(n \log(1/\epsilon))$ if a brute force algorithm is used in Line \ref{line:edge_high_lev}.
The new bichromatic closest pair procedure only improves the running time. 
The time complexity of post-processing (Line \ref{line:return}) also remains unchanged, which is $2^{O(d)} n + O(n \log n)$.

We conclude that the total time complexity is 
\begin{align*}
    2^{O(d)} n \log n + O_{d, \kappa}\Big(n \epsilon^{-(1 - \beta_d)} \log^{1 - \beta_d}(1/\epsilon)\Big).
\end{align*}
Rescaling $\epsilon$ completes the proof.
\end{proof}

\section{Lower Bounds}
\label{sec:lower-bounds}

In this section, we prove the following lower bound for constructing an explicit approximate MST in a labeled metric space of constant doubling dimension.

\thmlowerbounds*

The theorem follows by combining \Cref{thm:nlogn,ll:thm:main}, proved in \Cref{sec:nlogn,sec:lastlog-lower}, respectively.
The first gives an $\Omega(n\log n)$ comparison lower bound, even for one-dimensional Euclidean metrics.
It applies to the real RAM with a unit-cost exact-distance oracle of \Cref{def:computational-model}: after fixing the random tape, the real oracle answers influence the discrete computation only through comparison outcomes.
Thus a comparison transcript determines the oracle-query labels and the final set of labeled output edges.
Suppressing the intervening oracle calls, arithmetic, and RAM operations gives a decision tree with at most three children per comparison node and one fixed labeled output at each leaf, exactly as required by the output-counting argument.

The second gives an
\[
 \Omega\!\left(\min\left\{n^2,\frac{n}{\epsilon}\log(1/\epsilon)\right\}\right)
\]
lower bound on exact-distance queries, allowing unrestricted, free computation between queries.
It therefore applies to our real-RAM model as well.
Since comparisons and distance queries each cost one unit, both bounds imply running-time lower bounds.
Taking the maximum of the two lower bounds yields their sum up to an absolute constant factor.
The two bounds need not be attained on the same input.

\subsection{Computational Models}
\label{sec:model}

Our main results for general doubling metrics use the real RAM with a unit-cost exact-distance oracle specified in \Cref{def:computational-model}.
The input consists of point labels without geometric information; exact distances are returned in real registers.
Throughout this section, the algorithm outputs an explicit spanning tree as a set of $n-1$ pairs of integer point labels.

For clarity, we also define the atomic distance-comparison model and the exact-distance query model below.
The comparison lower bound proved in \Cref{sec:nlogn} applies more generally to finite-outcome comparisons of computed real values, as explained above.

\begin{definition}[Atomic distance-comparison model]
At an internal decision node, an algorithm names two unordered pairs $\{x,y\}$ and $\{u,v\}$ and branches according to
\[
  \dist(x,y)<\dist(u,v),\qquad
  \dist(x,y)=\dist(u,v),\qquad
  \dist(x,y)>\dist(u,v).
\]
The choice of the next comparison may depend on previous outcomes and on the point labels.
A leaf is labeled by one fixed combinatorial output, here an explicit set of $n-1$ labeled edges.
\end{definition}

Such an algorithm can be represented by a ternary decision tree:
each internal node corresponds to a comparison with at most three
possible outcomes, and each leaf specifies a fixed combinatorial
output. The depth of a leaf equals the number of comparisons made
on the corresponding execution.

\begin{definition}[Atomic exact-distance oracle]
A distance query names an unordered pair $\{x,y\}$ and returns the exact
value $\dist(x,y)$. Computation on values already returned is unrestricted;
each query costs one unit.
\end{definition}

An atomic distance comparison can be simulated by at most two exact-distance queries, caching previously
queried values. Thus a lower bound on exact-distance queries also gives,
up to a factor of two, a lower bound on atomic distance comparisons.
The converse does not follow, since an exact value may reveal more
information than one comparison outcome.

\subsection{The \texorpdfstring{$\boldsymbol{\Omega(n\log n)}$}{Omega(n log n)} Comparison Lower Bound}
\label{sec:nlogn}

The classical exact Euclidean-MST lower bound uses collinear points and the fact that the exact MST reveals their sorted order~\cite{ShamosHoey75}.
That argument does not directly survive approximation, because an approximate tree need not be the consecutive-order path.
Das, Kapoor, and Smid established an $\Omega(n\log(n/r))$ lower bound for $r$-approximate Euclidean MST construction in the algebraic computation-tree model, for sufficiently large $r<n$~\cite{DKS97}.
Their result also implies an $\Omega(n\log n)$ lower bound for every fixed approximation factor: an algorithm achieving a smaller factor also achieves any larger constant factor.
A related counting argument for approximate one-dimensional TSP was given by Young, who considers permutations of equally spaced points and bounds how many inputs can share one sufficiently good output~\cite{Young13}.
Below, we give a direct counting proof that bounds how many labeled orders can share one valid approximate tree.
The proof applies to randomized algorithms with success probability at least $2/3$ and bounds the worst-case expected number of comparisons, allowing arbitrary computations between comparisons as long as each comparison has a constant number of possible outcomes and the comparison transcript determines the labeled output.
We prove the following bound for every fixed approximation factor $\alpha\geq1$.

\begin{theorem}[$n\log n$ lower bound]
\label{thm:nlogn}
Fix an integer $b\ge2$ and a factor $\alpha\ge1$.
Any $b$-ary comparison algorithm that outputs an explicit
$\alpha$-approximate MST on every labeled set of $n$ points on the real
line requires, in worst-case expected complexity,
\[
 n\log_b n-O\bigl(n\log_b(2\alpha)\bigr)
\]
comparisons if it is deterministic or Las Vegas, and
\[
 \frac23 n\log_b n-O\bigl(n\log_b(2\alpha)\bigr)
\]
if it succeeds with probability at least $2/3$ on every input.
In particular, for fixed $\alpha$ and $b$ the bound is
$\Omega_\alpha(n\log n)$. The input metrics have doubling dimension
at most $2$.
\end{theorem}

\subsubsection{Output-Multiplicity Lemma}

We first establish a general counting argument for comparison-based
algorithms. Consider a finite family of $m$ inputs and suppose that
every fixed output is valid for at most $k$ of them.
A comparison algorithm can be represented by a decision tree in which
each internal node corresponds to a comparison and each leaf specifies
a fixed output.
Since at most $k$ inputs can be correctly handled by the same leaf,
a correct algorithm needs at least $m/k$ leaves.
If each comparison has at most $b$ possible outcomes, this immediately
gives a lower bound of $\log_b(m/k)$ on the worst-case number of
comparisons.

The following lemma formalizes this observation and extends it to
expected complexity and randomized algorithms using an
information-theoretic argument.

\begin{lemma}[Output multiplicity]
\label{lem:multiplicity}
Let $\mathcal I$ be a finite family of $m$ inputs.
Suppose every fixed combinatorial output is valid for at most $k$ inputs in $\mathcal I$.
Then a correct $b$-ary comparison tree has worst-case depth at least
\[
  \log_b(m/k).
\]
Under the uniform distribution on $\mathcal I$, its expected depth is at least the same quantity.

More generally, let a possibly randomized comparison algorithm have average success probability at least $s$ under the uniform input distribution and its internal coins.
If $c$ is its number of comparisons, then
\[
  \E[c]\geq
  s\log_b^+\!\left(\frac{sm}{k}\right),
  \qquad
  \log_b^+x:=\max\{0,\log_bx\}.
\]
\end{lemma}

\begin{proof}
We first consider a deterministic comparison tree, which may be
incorrect on some inputs. For each leaf $z$, let $d_z$ denote
its depth, $m_z$ the number of inputs reaching it, and $a_z$
the number of those inputs for which its output is valid.
By assumption, $a_z\leq k$.

Let
\[
  q:=\frac{1}{m}\sum_z a_z
\]
be the fraction of inputs on which the algorithm succeeds.
If $q=0$, the claimed bound is trivial, so assume $q>0$.
Conditioned on a successful execution, the probability of
reaching leaf $z$ is
\[
  p_z:=\frac{a_z}{qm}.
\]
In particular, $\sum_z p_z=1$ and $p_z\leq k/(qm)$.

We use the following elementary property of $b$-ary trees:
\[
  \sum_z b^{-d_z}\leq 1.
\]
Indeed, if $d$ is the maximum depth, each leaf $z$
corresponds to $b^{d-d_z}$ distinct sequences of $d$
branching outcomes. These sets of sequences are disjoint,
and there are at most $b^d$ sequences in total.

Since $x\mapsto b^{-x}$ is convex, Jensen's inequality
implies
\[
  b^{-\sum_z p_z d_z}
  \leq \sum_z p_z b^{-d_z}
  \leq \frac{k}{qm}\sum_z b^{-d_z}
  \leq \frac{k}{qm}.
\]
Taking logarithms gives
\[
  \sum_z p_z d_z
  \geq \log_b(qm/k).
\]
Since depths are nonnegative, the same bound holds with
$\log_b^+$ in place of $\log_b$.

The expected number of comparisons under the uniform
input distribution therefore satisfies
\[
  \E[c]
  =\frac{1}{m}\sum_z m_zd_z
  \geq\frac{1}{m}\sum_z a_zd_z
  =q\sum_z p_zd_z
  \geq q\log_b^+(qm/k).
\]
For an everywhere-correct deterministic algorithm, $q=1$,
which proves the expected-depth bound. The worst-case depth
is at least the expected depth.

For a randomized algorithm, fix its random tape and apply
the deterministic bound above. Let $q$ now denote the
success fraction for the fixed tape. Averaging over the
random tapes and using the convexity of
\[
  f(x):=x\log_b^+(xm/k)
\]
gives
\[
  \E[c]\geq \E[f(q)]
  \geq f(\E[q])
  \geq f(s),
\]
where the last inequality follows since $f$ is
nondecreasing and the average success probability is
at least $s$.
This proves the randomized bound.
\end{proof}

\subsubsection{Relabeled Path Metrics}

We now construct a family of metric spaces to which
\Cref{lem:multiplicity} can be applied.
The idea is to place $n$ labeled points at equally spaced
positions on the real line, but to vary which label occupies
which position. Although the MST of each such metric simply
connects consecutive positions, the algorithm does not know
their order.

Let $V=[n]$. For a permutation $\pi:V\to[n]$, define
\[
  \dist_\pi(u,v)=|\pi(u)-\pi(v)|.
\]
Thus, $\pi(v)$ specifies the position of vertex $v$ on the
real line. Different permutations generally yield different
metrics, except that reversing the order leaves all distances
unchanged.

\begin{lemma}
\label{lem:path-family}
The family $\{\dist_\pi:\pi\in S_n\}$ contains $n!/2$
distinct labeled metrics.
Every metric in the family has doubling dimension at most $2$
and MST value $n-1$.
\end{lemma}

\begin{proof}
The only isometries of $\{1,\ldots,n\}$ are the identity
and reversal.
Thus $\dist_\pi=\dist_\sigma$ precisely when
$\sigma=\pi$ or $\sigma(v)=n+1-\pi(v)$ for every $v$,
proving the count.

A radius-$r$ ball is contained in a real interval of
length $2r$.
Partition that interval into four intervals of length
$r/2$ and choose one input point from each nonempty part.
The corresponding four radius-$r/2$ balls cover the
original ball, so the doubling dimension is at most $2$.

The path connecting consecutive positions has weight
$n-1$.
Conversely, the tree path between the vertices at positions
$1$ and $n$ has weight at least their metric distance
$n-1$, by the triangle inequality.
Hence the MST value is $n-1$.
\end{proof}

To apply \Cref{lem:multiplicity}, we must now bound the
number of metrics for which a fixed spanning tree is
approximately optimal.
The key observation is that, after rooting the tree,
a permutation can be reconstructed from the position of
the root and the signed position differences along its
edges.
If the tree has small weight, the sum of the absolute
values of these differences must be small, which limits
the number of possible permutations.

\begin{lemma}\label{lem:few-orders}
Fix a labeled spanning tree $T$ and an approximation
factor $\alpha\geq1$.
The number of permutations $\pi$ for which
$  w_{\dist_\pi}(T)\leq\alpha(n-1)
$
is at most
\[
  k_{n,\alpha} \leq n(2e\alpha)^{n-1}.
\]
\end{lemma}

\begin{proof}
Root $T$ at a fixed label $r$ and orient its edges
away from $r$.
For an oriented edge $e=(u,v)$, put
\[
  \delta_e=\pi(v)-\pi(u).
\]
Then
\[
  w_{\dist_\pi}(T)=\sum_{e\in E(T)}|\delta_e|.
\]
Moreover, $\pi$ is uniquely determined by the root
position $\pi(r)$ and the signed increments
$\{\delta_e\}_{e\in E(T)}$:
the position of any vertex can be recovered by summing
the increments along its path from the root.

There are $n$ choices for $\pi(r)$ and $2^{n-1}$
possible sign patterns.
Since the absolute increments are positive integers
whose sum is at most $\alpha (n - 1)$, the number of possible
vectors $(|\delta_e|)_{e\in E(T)}$ is
\begin{align*}
    &\qquad \left|\left\{(a_1,\ldots,a_{n - 1})\in\mathbb Z_{\geq1}^{n - 1}:
           \sum_{i = 1}^{n - 1} a_i\leq \alpha (n - 1) \right\} \right| \\
  &=\binom{\lfloor \alpha (n - 1) \rfloor }{n - 1}
  \leq\left(\frac{e \lfloor \alpha (n - 1) \rfloor }{n - 1}\right)^{n - 1}
  \leq(e\alpha)^{n - 1}.
\end{align*}
To see the first equality, observe that the partial sums
$s_j=\sum_{i=1}^j a_i$ satisfy
$1\leq s_1<\cdots<s_{n - 1}\leq \lfloor \alpha(n - 1) \rfloor$.
Thus, choosing such a vector is equivalent to
choosing $n - 1$ distinct integers from $\{1,\ldots,\lfloor \alpha(n - 1) \rfloor\}$.

Therefore, 
\begin{align*}
    k_{n, \alpha} \leq \underbrace{n}_{\text{choices for } \pi(r)} \cdot \underbrace{2^{n - 1}}_{\text{sign patterns}} 
    \cdot \underbrace{(e \alpha)^{n - 1}}_{(|\delta_e|)_{e \in E(T)}}
    = n (2 e \alpha)^{n - 1}.
\end{align*}
Restricting that the reconstructed positions must be distinct integers in $[n]$ can only decrease the count.
This concludes the proof.
\end{proof}

Combining the two lemmas with the output-multiplicity bound yields the desired comparison lower bound.

\begin{proof}[Proof of \Cref{thm:nlogn}]
Take one representative from each reversal pair in \Cref{lem:path-family}, so $m=n!/2$.
By \Cref{lem:few-orders}, each output tree is valid for at most $k_{n,\alpha} \leq n(2 e \alpha)^{n - 1}$ inputs.
\Cref{lem:multiplicity} gives
\[
  \log_b\!\left(
    \frac{n!/2}
    {n (2 e \alpha)^{n - 1}}
  \right)
    =n\log_b n-O\bigl(n\log_b(2\alpha)\bigr),
\]
where we use Stirling's formula.
This proves the deterministic and Las Vegas bounds.  If the success
probability is at least $2/3$ on every input, the bounded-error part of
\Cref{lem:multiplicity} instead gives
$\frac23 n\log_b n-O\bigl(n\log_b(2\alpha)\bigr)$.
\end{proof}

\subsection{The Exact-Distance Lower Bound}
\label{sec:lastlog-lower}

In the previous subsection, we established an $\Omega(n\log n)$ lower bound in the comparison model, where each comparison has only a constant number of possible outcomes.
We now turn to the dependence on the approximation parameter $\epsilon$ and prove a lower bound that holds even when the algorithm can query exact distances and perform arbitrary computations on the returned values.
The main idea is to construct metric spaces in which certain edges are slightly cheaper than all alternatives, but identifying these edges requires many distance queries.
We show that a sufficiently accurate spanning tree must contain many of these hidden edges.
Andoni et al. proved an $\Omega(n^2)$ distance-query lower bound for exact MST construction in constant doubling dimension by hiding one cheaper connection between two clusters~\cite{ANOY14}.
Our construction uses several satellites sharing a common core at geometrically separated scales to establish the dependence on $\epsilon$ for approximate MST construction.

\begin{theorem}\label{ll:thm:main}
For every sufficiently large integer $n$ and $0<\epsilon\le \frac{1}{2}$, any randomized algorithm that returns a $(1+\epsilon)$-approximate spanning tree with probability at least $2/3$ on every $n$-point metric of doubling dimension at most $7$ has worst-case expected distance-query complexity
\[
 \Omega\!\left(\min\left\{n^2,
       \frac n\epsilon\log(1/\epsilon)\right\}\right).
\]
\end{theorem}

\paragraph{Proof overview.}
We first construct a metric consisting of a large core $S$ and several smaller satellites $P_1, \ldots, P_{k}$.
For each satellite $P_i$, all distances to the core are equal to some value $\rho_i+1$, except for one randomly chosen pair whose distance is $\rho_i$.
We call this pair the \emph{marked pair} of $P_i$.
The marked pair is the unique cheapest connection between the satellite and the core.
We will show that omitting any marked pair increases the minimum possible spanning-tree weight by one.
Consequently, if the optimal MST weight is sufficiently small relative to $1/\epsilon$, a $(1+\epsilon)$-approximate tree must contain most of the marked pairs.

The main ingredient for obtaining the logarithmic factor is that all satellites share the same core.
More precisely, we construct a core with $m=bk$ points and $k$ satellites with $b$ points each, giving a total of $2bk$ points.
Each marked pair is chosen independently and uniformly among the $mb=b^2k$ possible core--satellite pairs.
Even if the algorithm knows the cluster structure and all other distances, finding a given marked pair with constant probability requires $\Omega(mb)$ queries in expectation.
Thus, finding a constant fraction of the $k$ marked pairs requires $\Omega(kmb)=\Omega(b^2k^2)$ queries.
Relative to the $2bk$ points in the construction, this gives a lower bound of $\Omega(bk)$ queries per point.
The reuse of the core is crucial here: it allows each satellite to have $bk$ possible attachment points without requiring a separate core for every satellite.
\Cref{fig:exact-lb-gadget} illustrates the shared-core construction and the corresponding blocks of candidate marked pairs.

It remains to ensure that the construction has bounded doubling dimension and that the marked pairs are sufficiently important for approximation.
We equip each cluster with a unit-grid metric and place the satellites at geometrically increasing distances $\rho_i = 2^{i - 1} \rho$ from the core.
The geometric separation ensures that the doubling dimension remains bounded independently of $k$.
At the same time, the parameters are chosen such that the sum of the attachment distances is at most $m$.
As a result, the optimal MST has weight $O(m)$, and choosing $b=\Theta(1/\epsilon)$ forces an approximate tree to contain a constant fraction of the marked pairs.
The geometric construction permits $k=\Theta(\log b)=\Theta(\log(1/\epsilon))$ satellites, which yields the additional logarithmic factor.

Finally, we combine independent copies of this construction to obtain instances of arbitrary size.
When $n$ is sufficiently large relative to $1/\epsilon$, this gives the desired $\Omega((n/\epsilon)\log(1/\epsilon))$ lower bound.
Otherwise, the number of available points limits the construction, and we obtain the quadratic bound $\Omega(n^2)$ instead.
The remainder of the proof formalizes the metric construction, establishes the cost of omitting marked pairs, and bounds the number of queries needed to find them.

\subsubsection{Gadget Metric}

We first construct a gadget consisting of a large core and several smaller satellites, each with one hidden cheap attachment to the core.
The construction must satisfy three properties: it must define a metric of bounded doubling dimension, its MST must have relatively small weight, and omitting a marked attachment must increase the weight of any spanning tree.

We use two-dimensional grids as the building blocks.
For an integer $q\ge1$, let $G(q)$ be the first $q$ lattice points in row-major order in a square of width $\lceil\sqrt q\rceil$, equipped with the $\ell_1$ distance.
The unit grid graph on these points is connected, and every distance between distinct points is at least one.
Consequently,
\begin{equation}\label{ll:eq:grid}
 |G(q)|=q,\qquad \MST(G(q))=q-1,\qquad
 \operatorname{diam}(G(q))\le2\lceil\sqrt q\rceil.
\end{equation}

The grids allow us to have many potential attachment points while keeping the diameter of each cluster small compared to its number of points.
This will be important for separating the satellites geometrically without making their total attachment cost too large.

Fix an integer $k\ge2$ and set
\[
 b:=2^{8k},\qquad m:=bk,\qquad
 \rho:=16\lceil\sqrt m\rceil,\qquad
 \rho_i:=2^{i - 1}\rho\quad(1 \le i \leq k).
\]
Let $S$ be a copy of $G(m)$, and let
$P_1,\ldots,P_{k}$ be pairwise disjoint copies of $G(b)$,
also disjoint from $S$. Define
\[
 A:=S\mathbin{\dot\cup}P_1\mathbin{\dot\cup}\cdots
       \mathbin{\dot\cup}P_{k}.
\]
Thus $|A|=2m$. We call $S$ the \emph{core}, each $P_i$ a \emph{satellite},
and these $k+1$ sets the \emph{clusters}.

Independently for each $i$, choose a pair
$h_i\in S\times P_i$ uniformly at random among the $mb$ possible pairs.
We call the corresponding unordered edge the \emph{marked pair} of $P_i$.

For each point $x$, let $\xi(x)$ denote its coordinates
in its cluster's grid. Define a symmetric distance function
on $A$ by
\begin{equation}\label{ll:eq:distances}
 \dist_{\gad}(x,y):=
 \begin{cases}
  \|\xi(x)-\xi(y)\|_1,
       &x,y\text{ in the same cluster},\\
  \rho_i,
       &x\in S,\ y\in P_i,\ (x,y)=h_i,\\
  \rho_i+1,
       &x\in S,\ y\in P_i,\ (x,y)\ne h_i,\\
  \rho_i+\rho_j,
       &x\in P_i,\ y\in P_j,\ i\ne j.
 \end{cases}
\end{equation}
Namely, if $x$ and $y$ are in the same cluster, then their distance $\dist_\gad(x, y)$ is defined as their $\ell_1$ distance in the cluster's grid.
For a pair $(x, y) \in S \times P_i$, their distance is $\rho_i$ if the edge $(x, y)$ is marked, and $\rho_i + 1$ otherwise.
Thus, each satellite $P_i$ has exactly one connection to the core that is cheaper by $1$ than all other connections from $P_i$ to $S$.
The marked pair is hidden among $mb$ possibilities uniformly at random.
Finally, if $x$ and $y$ are in the $i$ and $j$-th satellites, respectively, for $i \neq j$, their distance is set to $\rho_i + \rho_j$.
We call the resulting metric space $(A, \dist_\gad)$ a \emph{gadget metric}.

The parameters are chosen to serve two purposes.
First, the geometric growth of the radii $\rho_i$ will allow us to keep the doubling dimension bounded independently of the number of satellites.
Second, their total sum must remain small compared to $m$, so that missing marked pairs is costly relative to the optimal MST weight.

Indeed, every cluster has diameter at most $\rho/8$, and
\begin{equation}\label{ll:eq:radii}
 \sum_{i=1}^{k}\rho_i 
 = \sum_{i = 1}^k \rho 2^{i - 1}
 < \rho 2^k
 \le32\sqrt m\,2^k\le m.
\end{equation}
The last inequality follows from $m=bk=2^{8k}k$ and is equivalent to $32\le 2^{3k}\sqrt k$, which holds for $k\ge2$.

We first verify that the construction defines a metric of bounded doubling dimension.

\begin{lemma}\label{ll:lem:metric}
Every choice of the marked pairs in~\eqref{ll:eq:distances} defines a metric.
Each resulting gadget metric $(A, \dist_\gad)$ has doubling dimension at most $7$.
\end{lemma}

\begin{proof}
We verify the triangle inequality and doubling dimension bound respectively.

\paragraph{Triangle inequality.}

For a triangle with two vertices in the same cluster, the two cross distances differ by at most one, while the distance between the two distinct vertices in that cluster is at least one.
Conversely, the internal distance is at most $\rho/8$, whereas both cross distances are at least $\rho$.
Thus, all triangle inequalities involving two vertices from the same cluster hold.

For vertices in $S,P_i,P_j$, write the two core--satellite distances as $\rho_i+\delta_i$ and $\rho_j+\delta_j$, where $\delta_i,\delta_j\in\{0,1\}$.
The satellite--satellite distance is $\rho_i+\rho_j$.
Its triangle inequality follows from $\delta_i+\delta_j\ge0$.
The other two follow from $\delta_i\le2\rho_j+\delta_j$ and $\delta_j\le2\rho_i+\delta_i$.

Triangles meeting three distinct satellites satisfy the triangle inequalities immediately, and triangles contained in one cluster satisfy them by the grid metric.

\paragraph{Doubling dimension.}
Consider an arbitrary ball $B$ of radius $r$.
First suppose $r<\rho$.
Since the distance between any two different clusters is at least $\rho$, the ball $B$ lies entirely within one grid cluster.
Such a ball is contained in a square of side length $2r$ in the $\ell_1$ plane.
Partition this square into at most $9\times9$ cells of side length at most $r/4$, and choose one input point from each nonempty cell.
The $\ell_1$ diameter of every cell is at most $r/2$, so the ball can be covered by at most $81$ balls of radius $r/2$.

Now suppose $r\ge \rho$.
Since every cluster has diameter at most $\rho/8\le r/8$, each cluster can be covered by a single ball of radius $r/2$.
If $B$ is contained in one cluster, it is covered by $1$ ball of radius $r/2$.

Suppose $B$ intersects more than $1$ cluster.
We claim that if $B$ intersects a satellite $P_j$, then $\rho_j \leq r$.
Indeed, if $B$ is centered in $P_j$, then it must have radius $r \geq \rho_j$ to intersect another cluster.
If $B$ is not centered in $P_j$, then $r$ must be at least the distance from $P_j$ to the center of $B$, which is at least $\rho_j$.
Therefore, 
\begin{align*}
    B \subseteq S \cup \bigcup_{j \colon \rho_j \leq r} P_j.
\end{align*}
All satellites with $\rho_j \le r/4$, together with the core $S$, can be covered by one ball of radius $r/2$ centered at an arbitrary core point.
Indeed, their distances to that point are at most $\rho_j+1\le r/4+1\le r/2$, where the last inequality uses $r\ge \rho\ge4$. The core itself has diameter at most $r/8$.
Since the radii $\rho_i$ grow geometrically, at most $3$ satellites have radii in $(r/4,r]$.
Each of these can be covered by one additional ball of radius $r/2$.

Thus, at most $4$ covering balls are needed when $r\ge \rho$.
Combining the two cases, we conclude that every ball of radius $r$ can be covered by at most $81$ balls of radius $r/2$.
The doubling dimension is at most $7$.
\end{proof}

The next lemma establishes the connection between the marked pairs and approximation quality.
In an optimal MST, each satellite uses its marked pair to connect to the core.
More importantly, omitting any marked pair increases the minimum possible weight by one, even if the spanning tree uses a different arrangement of edges between the clusters.

\begin{lemma}[MST value of $(A, \dist_\gad)$]\label{ll:lem:gap}
The optimum MST value of the gadget metric $(A, \dist_\gad)$ is
\begin{equation*}\MST(A) = 2m-k-1+\sum_{i = 1}^{k}\rho_i\le3m.
\end{equation*}
If a spanning tree omits $z$ marked pairs, its weight is at least $\MST(A) + z$.
\end{lemma}

\begin{proof}
All edges within a cluster are shorter than all edges between different clusters.
Kruskal's algorithm therefore completes an MST inside each cluster before adding any intercluster edges.
By~\eqref{ll:eq:grid}, the total weight of these internal trees is
\[
 (m-1)+k(b-1)=2m-k-1.
\]

After contracting the clusters, consider the resulting weighted complete graph on the core and the $k$ satellites.
For each satellite $P_i$, its marked edge of weight $\rho_i$ is the unique cheapest edge crossing the cut that separates $P_i$ from all other clusters.
Indeed, every other edge to the core has weight $\rho_i+1$, while every edge to another satellite $P_j$ has weight $\rho_i+\rho_j>\rho_i+1$.
Thus, the unique optimal contracted tree is the star consisting of all marked edges.
This proves the expression for $\MST(A)$, and the inequality $\MST(A) \le 3m$ follows from~\eqref{ll:eq:radii}.

Now consider an arbitrary spanning tree $T$ that omits $z$ marked pairs.
Delete these $z$ marked edges from the complete weighted graph.
The internal cluster metrics remain unchanged, and all internal edges are still cheaper than any intercluster edge.
For a satellite whose marked edge was deleted, the cheapest remaining connection to the core has weight $\rho_i+1$ instead of $\rho_i$.
For each satellite, choose one cheapest surviving edge to the core.
Each chosen edge is cheaper than every edge from that satellite to another satellite.
These chosen edges form a star and, by the cut property, an MST of the contracted graph.
Its weight is exactly $\MST(A)+z$ after adding the unchanged internal MST weights.

Since $T$ does not use any of the deleted edges, it is a feasible spanning tree in this modified graph.
Therefore, $w(T)\ge \MST(A) +z$, as claimed.
\end{proof}

\subsubsection{Replicated Metric}
\label{sec:replicated_metric}
The preceding construction gives a hard instance with exactly $2m=2bk$ points.
To obtain the lower bound for arbitrary input sizes, we combine independent copies of this gadget.
The copies are placed along a line at sufficiently large distances so that their internal MST structures remain independent, while the doubling dimension stays bounded.

Let $s \ge 1$ and take $s$ independent copies of the gadget metric, denoted by $(A_1, \dist_\gad^1), (A_2, \dist_\gad^2), \dots, (A_s, \dist_\gad^s)$.
Optionally append a padding block $A_{s + 1} = G(q)$ with $1\le q<2m$ points, and write $q = 0$ if no padding is needed.
Denote
\begin{equation*}\eta := \rho2^{k + 1} \le2m.
\end{equation*}
We define our replicated metric space $(X, \dist)$ as follows:
The point set $X$ is defined as 
$X := \bigcup_{i = 1}^{s + 1} A_i$ of size $|X| = n$, where 
\[
 n=2s m+q<2(s+1)m\le4sm.
\]
For $x,y\in X$, define the distance between $x$ and $y$ as
\begin{gather}\label{eqn:replicated_metric}
    \begin{aligned}
        \dist(x, y)
    := \begin{cases}
        \dist_\gad^i(x, y), &\text{if } x, y \in A_i \text{ for some } 1 \leq i \leq s; \\
        \|\xi(x) - \xi(y)\|_1, &\text{if } x, y \in A_{s + 1}; \\
        \eta |i - j|, &\text{if } x \in A_i, y \in A_j \text{ for some } 1 \leq i \neq j \leq s + 1.
    \end{cases}
    \end{aligned}
\end{gather}
Namely, if $x$ and $y$ are in the same block $A_i$ for $i \leq s$, then $\dist(x, y)$ is defined the same as their distance in the gadget metric $(A_i, \dist_\gad^i)$.
If $x$ and $y$ are both in the padded grid $A_{s + 1}$, then $\dist(x, y)$ is defined as their $\ell_1$ distance in the grid.
Finally, if $x$ and $y$ are in different blocks
numbered $i$ and $j$, their distance is set to $\eta |i - j|$.

The following lemma shows that replication preserves all properties needed for the lower bound.

\begin{lemma}\label{ll:lem:replication}
The distance function in \eqref{eqn:replicated_metric} defines a metric.
The resulting replicated metric $(X, \dist)$ has doubling dimension at most $7$.
\end{lemma}

\begin{proof}
We prove the triangle inequality and the doubling dimension bound respectively.

\paragraph{Triangle inequality.}
Every full gadget $A_i$ ($i \leq s$) has diameter at most $2\rho_{k}= \eta /2$, and the padding grid $A_{s + 1}$ also has diameter at most $\eta / 2$.
Indeed, if $q>0$, then
\[
\diam(G(q))
 \le2\lceil\sqrt q\rceil
 \le4\lceil\sqrt m\rceil
 =\rho/4
 \le\eta/2.
\]
Thus, triangles with two points in the same block satisfy the triangle inequality, since their internal distance is at most $\eta /2$ and their distances to any point in another block are equal.
Triangles with points in three distinct blocks satisfy the triangle inequality because their distances form a line metric.
Triangles with three points in the same block satisfy the triangle inequality because each $(A_i, \dist_\gad^i)$ is a metric.

\paragraph{Doubling dimension.}
For a ball $B$ of radius $r< \eta$, all its points belong to one block $A_i$.
If $i \leq s$, then by \Cref{ll:lem:metric}, $B$ can be covered by $128$ balls of radius $r/2$.
The same bound holds if $i = s + 1$.

If $\eta \le r<2 \eta$, the ball $B$ meets at most three consecutive blocks.
Each such block has diameter at most $\eta /2\le r/2$, so $3$ covering balls suffice.

Finally, suppose $r\ge2 \eta$.
Consider the index set $I = \{i \in [s + 1] \colon B \cap A_i \neq \emptyset\}$, i.e., 
$I$ is the set of indices of blocks that intersect $B$.
Then $I$ is consecutive and $\max I-\min I \leq 2r/\eta$.
Choose a maximal subset $J \subseteq I$ such that the difference between any two indices in $J$ exceeds $r/(2\eta)$.
Then $|J| \leq 4$, since five selected indices would have span strictly greater than $4\cdot r/(2\eta)=2r/\eta$.
For each $j \in J$, place a ball $B_j$ of radius $r/2$ centered at an arbitrary point in $A_j$.
We claim that $\bigcup_{j \in J} B_j \supseteq B$.
Indeed, for every $x \in B$, suppose $x \in A_i$ for some $i \in I$.
By the maximality of $J$, there exists $j \in J$ such that $|i - j| \leq r/(2 \eta)$.
Let $y \in A_j$ be the center of $B_j$; if $i\ne j$, then $\dist(x,y)=\eta|i-j|\le r/2$ by \eqref{eqn:replicated_metric}; if $i=j$, then $\dist(x,y)\le\diam(A_i)\le\eta/2\le r/2$. Thus $x$ is covered by $B_j$.
Hence, $B$ can be covered by
at most $4$ balls of radius $r/2$.

In conclusion, for all $r$ and an arbitrary ball $B$ of radius $r$, $B$ can be covered by at most $128$ balls of radius $r/2$.
The doubling dimension of the replicated metric $(X, \dist)$ is therefore at most~$7$.
\end{proof}

\begin{lemma}[MST value of $(X, \dist)$]\label{lemma:MST_replicated_metric}
    The optimum MST value of the replicated metric $(X, \dist)$ satisfies $\MST(X) \le7s m$.
    If a spanning tree $T$ of $X$ omits $z$ of the $s k$ marked pairs, then
    \begin{equation*}
     w(T)\ge \MST(X) + z.
    \end{equation*}
\end{lemma}

\begin{proof}
All distances within a block $A_i$ are strictly smaller than $\eta$, while distances between different blocks are at least $\eta$.
Kruskal's algorithm therefore completes each block internally before connecting different blocks.
After contraction, the blocks form a line metric whose MST connects consecutive blocks with edges of weight $\eta$.
Consequently,
\[
 \MST(X) \le \sum_{i = 1}^s \MST(A_i) +q+s \eta \le7sm.
\]

If $z$ marked edges are deleted, Lemma~\ref{ll:lem:gap} shows that the total optimum weight within the full gadgets increases by exactly $z$.
Since every block can still be connected internally using edges shorter than $\eta$, the optimal connections between blocks remain unchanged.
Thus, deleting these marked edges increases the global MST weight by exactly $z$.
Any spanning tree omitting them is feasible in the modified graph and hence has weight at least $\MST(X) +z$.
\end{proof}

\subsubsection{Finding the Hidden Pairs}

We can now quantify how many marked pairs an approximate spanning tree must contain.
Suppose that $\epsilon\le1/(128b)$ and that $T$ is a $(1+\epsilon)$-approximate spanning tree.
If $T$ omits $z$ marked pairs, then \Cref{lemma:MST_replicated_metric} gives
\[
 \MST(X) + z\le w(T)\le(1+\epsilon) \MST(X),
\]
and therefore
\[
 z\le\epsilon \MST(X) \le7\epsilon s m\le\frac{7}{128}s k.
\]
Since there are $s k$ marked pairs in total, every successful output must contain at least
\begin{equation}\label{ll:eq:required-marks}
 s k-\epsilon \MST(X) \ge \frac{121}{128}s k
\end{equation}
of them.
This conclusion holds for arbitrary spanning trees, regardless of how they connect the different clusters.

It remains to prove that identifying these $\frac{121}{128} sk$ pairs requires many distance queries.
Recall that each marked pair is chosen independently and uniformly from one block $S\times P_i$ containing
\[
 \ell:=mb=b^2k
\]
possible pairs.
We may reveal the cluster structure, all grid coordinates, all radii, and every distance outside these blocks for free.
Each remaining distance query then simply tests whether a particular pair is marked: the answer is either $\rho_i$ or $\rho_i+1$.
Thus, although the algorithm may query exact distances, each query provides only the information of whether a particular pair is marked.

We first prove a general lower bound for finding marked entries in independent blocks.
The bound allows the algorithm to choose queries adaptively, to interleave queries between blocks, and to stop at a random time.

\begin{lemma}[Independent marked blocks]\label{ll:lem:search}
Consider any adaptive, possibly randomized procedure that queries entries in independent blocks, each with one uniform marked entry among $\ell$ positions, where $\ell$ is even.
If $t$ is its number of queries and $f$ is the number of marked entries it discovers by querying them, then
\[
 \E[f]\le \frac{4\E[t]}{\ell}.
\]
The procedure may interleave queries to different blocks and decide when to stop based on the answers received so far and its internal randomness.
\end{lemma}

In particular, discovering a constant fraction of the $s k$ marked pairs requires $\Omega(s k\ell)$ queries in expectation, even when the algorithm chooses its queries adaptively.

\begin{proof}
Repeated queries provide no new information, so we may assume that every entry is queried at most once.
Call a block \emph{heavy} once at least $\ell/2$ distinct entries have been queried in it.
Since each heavy block accounts for at least $\ell/2$ queries, the number of heavy blocks is at most $2t/\ell$.

Now consider a query to a block that has not yet become heavy and whose marked entry has not been found.
Conditioned on all previous answers, the marked entry is uniformly distributed among the remaining unqueried positions.
Since at least $\ell/2$ positions remain, the conditional probability that the next query finds the mark is at most $2/\ell$.
Summing over all such queries, the expected number of marks discovered before their blocks become heavy is at most $2\E[t]/\ell$.

Every remaining discovery occurs in a heavy block.
Since each block contains only one marked entry, the number of such discoveries is at most the number of heavy blocks, whose expectation is at most $2\E[t]/\ell$.
Combining the two contributions gives
\[
 \E[f]\le\frac{4\E[t]}{\ell}.
\]
The argument also applies to random stopping times by summing the conditional probabilities over all query indices.
If the expected number of queries is infinite, the statement is immediate.
\end{proof}

We now apply this bound to the replicated construction.
There is one subtlety: an algorithm might output a marked edge without having queried its distance, simply by guessing correctly.
To account for this possibility, we append queries to the edges of its output tree after the algorithm terminates.
This costs at most $n-1$ additional queries and ensures that every marked edge used by the output has been discovered.

\begin{lemma}\label{ll:lem:queries}
For the replicated construction, any algorithm that returns a $(1+\epsilon)$-approximate tree with probability at least $2/3$ on every input, where $\epsilon\le1/(128b)$, has expected query cost $\Omega(s b^2k^2)$ on some choice of the marks.
\end{lemma}

\begin{proof}
Supply all cluster memberships, grid coordinates, radii, and distances outside the marked blocks for free.
Thus, every remaining distance query tests one entry of a block containing one uniformly random marked entry among $\ell=mb$ possibilities.

Average over the independent uniform choices of the marked pairs and the algorithm's internal random bits.
Let $t$ denote the number of distance queries made by the original algorithm.

After the algorithm terminates, query every edge of its output whose distance has not already been queried, provided that the output is a spanning tree.
This increases the number of queries by at most $n-1$ and ensures that every marked edge in the output has been discovered.
By~\eqref{ll:eq:required-marks}, every successful output contains at least $(121/128)s k$ marked pairs.
Since the algorithm succeeds with probability at least $2/3$ on every input, the expected number of discovered marked pairs is at least
\[
 \frac23\cdot\frac{121}{128}s k
 =\frac{121}{192}s k.
\]

Applying Lemma~\ref{ll:lem:search} to the original and appended queries gives
\[
 \frac{121}{192}s k
 \le\frac{4(\E[t]+n)}{\ell}.
\]
Equivalently,
\[
 \E[t]+n
 \ge\frac{121}{768}s k\ell
 \ge\frac18 s k \ell
 = \frac{1}{8} s k \cdot mb 
 = \frac{1}{8} s b^2 k^2.
\]
Recall that $n\le4s m = 4 s b k \leq \frac{1}{16} s b^2 k^2$.
Thus, the additional queries are negligible compared to the lower bound, and subtracting them yields
\begin{equation*}\E[t]\ge \frac{1}{8} s b^2 k^2 - n \geq \frac1{16}s b^2k^2.
\end{equation*}

The expectation above is taken over both the random marked pairs and the algorithm's internal randomness.
Consequently, there exists a fixed choice of the marked pairs for which the algorithm's expected query cost is at least this large.
The argument applies directly to expected query complexity and does not require truncating the running time.
\end{proof}

\subsubsection{\texorpdfstring{Proof of \Cref{ll:thm:main}}{Proof of the exact-distance lower bound}}

\begin{proof}[Proof of \Cref{ll:thm:main}]

We first prove the lower bound for sufficiently small $\epsilon \leq 2^{-23}$ by applying \Cref{ll:lem:queries}.
For larger $\epsilon$, the desired lower bound reduces to $\Omega(n)$, which follows from the relabeled path family and the output-counting argument.
Throughout, we may assume $n\geq2^{18}$, since the theorem concerns sufficiently large $n$.

First suppose $0<\epsilon\leq2^{-23}$.
Let $k\geq2$ be the largest integer such that
\begin{align*}
    \text{for } b=2^{8k}, \qquad
    \begin{cases}
        b\leq\frac{1}{128\epsilon};\\
        2bk\leq n.
    \end{cases}
\end{align*}
Such an integer exists because $k=2$ satisfies both constraints.
Let $s=\lfloor n/(2bk)\rfloor\geq1$.
Construct the replicated metric $(X,\dist)$ as in \Cref{sec:replicated_metric}.
By \Cref{ll:lem:replication}, its doubling dimension is at most~$7$.

By \Cref{ll:lem:queries}, any randomized algorithm satisfying the theorem's success guarantee has worst-case expected query complexity $\Omega(sb^2k^2)$ on this family.
By our choice of parameters,
\begin{align*}
    sbk=\Omega(n), \qquad
    bk=\Omega\!\left(\min\{n,\epsilon^{-1}\log(1/\epsilon)\}\right).
\end{align*}
We thus obtain the desired query lower bound
\[
    \Omega\!\left(\min\{n^2,n\epsilon^{-1}\log(1/\epsilon)\}\right).
\]

Now suppose $2^{-23}<\epsilon\leq1/2$.
Use the relabeled path metrics from \Cref{lem:path-family}.
Every exact-distance query has at most $n$ possible answers on this family, so, after fixing the random tape, the query transcript defines an $n$-ary decision tree with a fixed output at each leaf, even with unrestricted computation between queries.
By \Cref{lem:few-orders}, every output is valid for at most $n(3e)^{n-1}$ members of this family because $1+\epsilon\leq3/2$.
Applying \Cref{lem:multiplicity} with branching factor $n$ to the number $Q$ of queries gives
\[
\E[Q]\geq\frac23\log_n^+\!\left(\frac{n!/3}{n(3e)^{n-1}}\right)=\Omega(n).
\]
This proves the required bound in this range, where $\epsilon^{-1}\log(1/\epsilon)$ is bounded by an absolute constant.
\end{proof}

\section{Conclusion}
\label{sec:conclusion}

We gave a deterministic algorithm for computing a $(1+\epsilon)$-approximate minimum spanning tree in an $n$-point metric space of doubling dimension $\ddim$ in time $2^{O(\ddim)}n(\log n+\epsilon^{-1}\log(1/\epsilon))$.
The algorithm adapts the accuracy of closest-pair computations to their scale, using a critical level of the net hierarchy to control the total approximation error.
A key ingredient in the running-time analysis is the additional WSPC property in \Cref{lemma:WSPC_doubling} \eqref{it:few_NN}: for every subset of the input, only linearly many pairs have both centers nearby at their respective scales.
Combined with the probabilistic method, this property bounds the aggregate cost of the closest-pair computations across the hierarchy.

Our lower bounds show that the dependence on $n$ and $\epsilon$ is essentially optimal for labeled doubling metrics in the real RAM with a unit-cost exact-distance oracle of \Cref{def:computational-model}.
The accuracy-dependent bound holds even when exact distances are available and all computation between queries is free.
In fixed-dimensional Euclidean spaces, access to coordinates and faster bichromatic closest-pair algorithms allow a better dependence on $\epsilon$.
We also showed that exact MSTs in bounded doubling metrics can have arbitrarily large maximum degree, while approximate MSTs of degree $2^{O(\ddim)}(1+\log(1/\epsilon))$ always exist.

The main remaining question is to determine the optimal dependence on $\epsilon$ for approximate MST construction in fixed-dimensional Euclidean spaces.
It would also be interesting to understand whether the WSPC property used in our analysis can improve other algorithms based on hierarchical nets and well-separated pairs.

\section*{AI Methodology}
\label{sec:ai-methodology}

We used ChatGPT (Astra-6 Ultra) to assist with proof development and language editing.
The tool materially affected the running-time analysis of \Cref{thm:MST_alg_improved} and the proofs of \Cref{thm:euclidean_mst,thm:lowerbounds}.
Specifically, for \Cref{thm:MST_alg_improved}, AI assistance was used to refine the running-time analysis by removing additional $\log(1/\epsilon)$ factors.
For \Cref{thm:euclidean_mst,thm:lowerbounds}, the model assisted in developing the proofs.

We also used ChatGPT to improve the language, clarity, and presentation throughout the manuscript.
In addition, we used Gemini (Argon), Google DeepReview, and Claude (Fable 5.1, Opus 5) to review the manuscript and suggest improvements to its presentation.
The authors verified the correctness and originality of all content, including references, and take full responsibility for the paper.
We devoted particular attention to making the arguments clear and readable.

\bibliography{biblio.bib}
\bibliographystyle{alphaurl}

\appendix
\label{sec:appendix}

\section{Lemmas about (Approximate) MSTs}

\begin{restatable}{proposition}{lemmaxdegree}
	\label{lemma:MST_degree}
	Let $(X, \dist)$ be a metric space of doubling dimension $\ddim$.
Let $T$ be an arbitrary minimum spanning tree of $X$.
Then for every vertex $p \in X$,
\[
\deg_T(p) \le 2^{O(\ddim)} \cdot \log(2 \Delta_p) \ .
\]
Here $\Delta_p = \max_{q\in X\setminus\{p\}}\dist(p,q)/\min_{q\in X\setminus\{p\}}\dist(p,q)$.
\end{restatable}

\begin{proof}
Fix a vertex $p$ and let its neighbors in $T$ be
$q_1,\dots,q_k$, ordered so that
\[
r_i := \dist(p,q_i),
\qquad
r_1 \le r_2 \le \dots \le r_k.
\]

We claim the following separation property: For every $i<j$,
\[
\dist(q_i,q_j) \ge r_j.
\]
Indeed, if $\dist(q_i,q_j) < r_j$, then replacing the edge
$p q_j$ (of length $r_j$) by the edge $q_i q_j$
would strictly decrease the total weight of the tree,
contradicting minimality of the MST.

Next, we partition the neighbors of $p$ into geometric scales.
Fix a constant $C \ge 2$.
For each integer $t$, define
\[
S_t := \{\, i : C^t \le r_i < C^{t+1} \,\}.
\]

We will bound the size of each $S_t$.
Let $i<j$ both lie in $S_t$.
Then $r_j \ge C^t$.
By the separation property,
\[
\dist(q_i,q_j) \ge r_j \ge C^t.
\]
Moreover, every $q_i$ with $i \in S_t$ lies in the ball
$B(p, C^{t+1})$.
Thus, the points $\{q_i : i \in S_t\} \subseteq B(p,C^{t+1})$ form a
set that is $C^t$-packing and has diameter at most $2\cdot C^{t+1}$. By Lemma \ref{lemma:packing} we therefore get that the size of $S_t$ is bounded by $|S_t| = |\{q_i : i \in S_t\}| = 2^{O(\ddim)}$.

Finally, note that we can bound the total number of scales as follows:
The smallest radius is $r_1$ and the largest is $r_k$.
The number of indices $t$ for which $S_t$ is nonempty
is at most
\[
O\!\left(
\log_C \frac{2 r_k}{r_1}
\right)
=
O(\log (2 \Delta_p)).
\]

Summing over all scales, we obtain
\[
\deg_T(p)
=
\sum_t |S_t|
\le
2^{O(\ddim)} \cdot \log(2 \Delta_p),
\]
as claimed.
\end{proof}

\begin{restatable}{proposition}{lemmaapxupperbounddegree}
    \label{lemma:apx_MST_degree}
    Let $(X, \dist)$ be a finite metric space with doubling dimension $\ddim$ and $\epsilon \in (0, 1/2)$.
    Then there exists a spanning tree $T$ of $X$, such that 
    \begin{enumerate}[(a)]
        \item The maximum degree of $T$ is $2^{O(\ddim)} \log(1/\epsilon)$, and 
        \item $w(T) \leq (1 + \epsilon) \MST(X)$.
    \end{enumerate}
\end{restatable}

\begin{proof}
    Assume wlog that the minimum interpoint distance of $X$ is $1$ and the diameter of $X$ is $\Delta$.
    Let $T^*$ be the MST of $X$.
    By \Cref{lemma:MST_degree}, the maximum degree of $T^*$ is $2^{O(\ddim)}\log(2\Delta)$.
    Fix a threshold $\tau = 2^{O(\ddim)} \log(1/\epsilon)$.
    Say a point $x \in X$ is \emph{heavy} w.r.t. $T^*$, if $\deg_{T^*}(x) \geq \tau$, and \emph{light} otherwise.
    Denote by $H_{T^*} \subseteq X$ the set of heavy points w.r.t. $T^*$.
    The overall plan is to start with $T^*$, and 
    modify the edge connection around each heavy point $x \in H_{T^*}$ to make it light.
    The final spanning tree will be the desired $T$.

    We treat $T^*$ as a rooted tree.
    Initially, all points in $H_{T^*}$ are marked as active.
    At each step, we find the lowest level (i.e., farthest from the root) active (heavy) point, denoted by $x \in H_{T^*}$.
    Denote $\wdeg_{T^*}(x) := \sum_{(x, y) \in E(T^*)} \dist(x, y)$ as the weighted degree of $x$.
For integer $1 \leq \ell \leq \log( \wdeg_{T^*}(x))$, denote the level $\ell$ neighbor of $x$ by $N_{T^*}^\ell(x) := \{y \in B(x, 2^\ell) \setminus B(x, 2^{\ell - 1}) \colon (x, y) \in E(T^*)\}$.
    By the proof of \Cref{lemma:MST_degree}, we have $|N_{T^*}^\ell(x)| \leq 2^{O(\ddim)}$.
    Let $p(x)$ be the parent of $x$ on $T^*$.
    We do the following modification to the edges incident to $x$:
    \begin{enumerate}[(1)]
        \item For every $1\leq \ell \leq \log(\epsilon \wdeg_{T^*}(x))$, remove all edges $\{(x, y) \colon y \in N_{T^*}^\ell(x) \setminus \{p(x)\}\}$ from the current tree, 
        pick an arbitrary $c_\ell \in N_{T^*}^\ell(x) \setminus \{p(x)\}$, 
        and add edges $\{(c_\ell, y) \colon y \in N_{T^*}^\ell(x) \setminus \{p(x)\}\}$.
        Add a path $(x, c_1, c_2, \dots, c_{\log(\epsilon \wdeg_{T^*}(x))})$.\footnote{If $p(x)$ is the only point in $N_{T^*}^\ell(x)$, then $c_\ell$ does not exist.
        In this case we simply skip level $\ell$.
        For the ease of presentation, we assume every $c_\ell$ exists.}

        \item For $\ell > \log(\epsilon \wdeg_{T^*}(x))$, keep the edges $\{(x, y) \colon y \in N_{T^*}^\ell(x)\}$ untouched.
        \item Keep the edge $(x, p(x))$ untouched.
        \item Mark $x$ as inactive.
    \end{enumerate}
    Roughly speaking, we remove all ``low level'' edges incident to $x$, and add an arbitrary star at each level (whose center is guaranteed not to be $p(x)$).
    To make everything connected, we then add a path through the star centers.
    Crucially, we never touch $p(x)$ and thus all modifications are made between $x$ and its children on $T^*$.
    This ensures that the degree of $p(x)$ stays unchanged during the modification around $x$.

    After all heavy points in $H_{T^*}$ are marked as inactive, we terminate the procedure and denote the resulting tree by $T$.
    Before we show the degree and weight bounds for $T$, let us first check that $T$ is indeed a spanning tree of $X$.
    This follows immediately as our algorithm maintains connectivity, and does not change the total number of edges.

    We now show that the maximum degree of $T$ is at most $2^{O(\ddim)} \log(1/\epsilon)$.
    We first state the following observation, which follows directly from our procedure:
    \begin{observation}\label{obs:degree}
        For every $x \in H_{T^*}$, $\deg_T(x) = \deg_{T^*}(x)$ holds before the step when $x$ becomes inactive.
        For every $x \notin H_{T^*}$, $\deg_T(x) = \deg_{T^*}(x)$ holds before the step when $p(x)$ becomes inactive.
    \end{observation}
    We prove $\deg_T(x) \leq 2^{O(\ddim)} \log(1/\epsilon)$ for every $x \in X$.
    If $x \notin H_{T^*}$, then by \Cref{obs:degree}, $\deg_T(x)$ changes only at the step when $p(x)$ becomes inactive.
    By our procedure, $\deg_T(x)$ increases only if $x$ is the center of the star at a certain level $\ell$.
    In this case, $x$ is incident to at most $|N_{T^*}^\ell(p(x))| - 1$ edges of the star and at most $2$ edges of the path.
    Thus \[
    \deg_T(x) \leq \deg_{T^*}(x) + (|N_{T^*}^\ell(p(x))| + 1)
    \leq \tau + 2^{O(\ddim)} + 1 = 2^{O(\ddim)} \log(1/\epsilon).
    \]
    If $x \in H_{T^*}$, then by \Cref{obs:degree}, $\deg_T(x)$ changes at the steps when $x$ or $p(x)$ becomes inactive.
    At the step when $x$ becomes inactive, edges are modified in such a way that $x$ is only adjacent to 
    (i) all points in $N_{T^*}^\ell(x)$ for $\ell > \log(\epsilon \wdeg_{T^*}(x))$,
    (ii) its parent $p(x)$ and 
    (iii) the center $c_1$ of the first level star.
    Hence for this step, 
    \begin{align*}
        \deg_T(x) \leq \sum_{\ell = \log(\epsilon \wdeg_{T^*}(x))}^{\log(\wdeg_{T^*}(x))} |N_{T^*}^\ell(x)|
        + 2
        \leq 2^{O(\ddim)} \log(1/\epsilon) + 2.
    \end{align*}
    At the stage when $p(x)$ becomes inactive, $\deg_T(x)$ increases by at most $2^{O(\ddim)} + 1$ (if it is selected as a star center $c_\ell$).
    Therefore, the final degree of $x$ in $T$ is bounded by 
    \begin{align*}
        \deg_T(x) \leq 2^{O(\ddim)} \log(1/\epsilon) + 2 + 2^{O(\ddim)} + 1
        \leq 2^{O(\ddim)} \log(1/\epsilon).
    \end{align*}
    We thus conclude that the maximum degree of $T$ is 
    $2^{O(\ddim)} \log(1/\epsilon)$.

    We next prove that $T$ is a $(1 + 2^{O(\ddim)} \epsilon)$-approximate MST of $X$.
    For this purpose, let us bound the increase of edge weight at each step.
    Assume the current step marks a point $x \in H_{T^*}$ as inactive.
    Recall that the newly added edges are 
    (i) edges in the star $\{(c_\ell, y) \colon y \in N_{T^*}^\ell(x) \setminus \{p(x)\}\}$ for $\ell \leq \log(\epsilon \wdeg_{T^*}(x))$ and 
    (ii) edges in the path $(x, c_1, c_2, \dots, c_{\log(\epsilon \wdeg_{T^*}(x))})$.
    The total length of the newly added edges is
    \begin{align*}
        &\qquad \sum_{\ell = 1}^{\log(\epsilon \wdeg_{T^*}(x))}
        \sum_{y \in N_{T^*}^\ell(x) \setminus \{p(x)\}}
        \dist(y, c_\ell)
        + \sum_{\ell = 0}^{\log(\epsilon \wdeg_{T^*}(x))}
        \dist(c_\ell, c_{\ell + 1}) \\
        & \leq \sum_{\ell = 1}^{\log(\epsilon \wdeg_{T^*}(x))}
        2^{O(\ddim)} \cdot 2 \cdot 2^\ell
        + \sum_{\ell = 0}^{\log(\epsilon \wdeg_{T^*}(x))}
        2 \cdot 2^{\ell + 1} \\
        & \leq 2^{O(\ddim)} \epsilon \wdeg_{T^*}(x).
    \end{align*}
    Summing the increased length up over all steps, we have
    \begin{align*}
        w(T) - w(T^*)
        \leq \sum_{x \in H_{T^*}} 2^{O(\ddim)} \epsilon \wdeg_{T^*}(x)
        \leq 2^{O(\ddim)} \epsilon \sum_{x \in X} \wdeg_{T^*}(x)
        \leq 2^{O(\ddim)} \epsilon \cdot w(T^*).
    \end{align*}
    Therefore, $w(T) \leq (1 + 2^{O(\ddim)} \epsilon) w(T^*) = (1 + 2^{O(\ddim)} \epsilon) \MST(X)$.
    Rescaling $\epsilon$ concludes the proof.
\end{proof}

\begin{restatable}{proposition}{lemmaapxupperbounddegtight}
\label{lemma:apx_MST_degree_lower_bound}
For every $\epsilon\in(0,1/2)$, there exists a finite metric space $(X,\dist)$ with $\ddim(X)=O(1)$ such that every $(1+\epsilon)$-approximate MST of $X$ has maximum degree $\Omega(\log(1/\epsilon))$.
\end{restatable}

\begin{proof}
Let $k=\lfloor\log_2(1/\epsilon)\rfloor$ and let $X=\{o,x_1,\ldots,x_k\}$ carry the shortest-path metric of the star with edge lengths $\dist(o,x_i)=2^{i-1}$.
Thus $\dist(x_i,x_j)=2^{i-1}+2^{j-1}$ for distinct leaves.

Every nonsingleton ball of radius $R$ is contained in $B(o,R)$.
Since the distances $\dist(o,x_i)$ are successive powers of two, at most one point of $B(o,R)$ lies outside $B(o,R/2)$.
Consequently, every ball is covered by at most two balls of half its radius, so $\ddim(X)\leq1$.

For any spanning tree $T$,
\[
w(T)=\sum_{i=1}^k 2^{i-1}\deg_T(x_i).
\]
The star therefore has optimum weight $2^k-1$, while every other spanning tree has weight at least~$2^k$.
Since $\epsilon(2^k-1)<1$, every $(1+\epsilon)$-approximate MST must be the star.
Its maximum degree is $k=\Omega(\log(1/\epsilon))$.
\end{proof}

\section{\texorpdfstring{Omitted Proofs from \Cref{sec:preliminaries}}{Omitted proofs from the preliminaries}}

\subsection{Proofs about the Net Hierarchy}
\label{sec:appendix:net-hierarchy}

\lemmadescpackingcovering*

\begin{proof}
    First assume $0\leq k<\ell$.
    We first show that $\desc_k(u) \subseteq B(u, 2^{\ell + 2})$.
    To this end, we inductively prove that for every $k < \ell$, 
    \begin{align}
        \desc_k(u) \subseteq B\left(u, \sum_{j = k + 2}^{\ell + 1} 2^j\right)
        \subseteq B(u, 2^{\ell + 2}).
        \label{eqn:descendant_is_subset}
    \end{align}
    For $k = \ell - 1$, we have $\desc_{\ell - 1}(u) = \child(u) \subseteq B(u, 2^{\ell + 1})$ by \Cref{def:children,def:descendants}.
    Now assume \eqref{eqn:descendant_is_subset} holds for $1\leq k < \ell$.
    For every $v \in \desc_{k - 1}(u)$, by \Cref{def:descendants}, there exists $w \in \desc_k(u)$, such that $v$ is a child of $w$.
    Hence, \begin{align*}
        \dist(v, u) \leq \dist(v, w) + \dist(w, u)
        \leq 2^{k + 1} + \sum_{j = k + 2}^{\ell + 1} 2^j 
        = \sum_{j = k + 1}^{\ell + 1} 2^j.
    \end{align*}
    Therefore, $\desc_{k - 1}(u) \subseteq B(u, \sum_{j = k + 1}^{\ell + 1} 2^j)$, concluding the induction.
    Finally, since $N_k$ is $2^{k - 3}$-packing (\eqref{eqn:N_l_properties} and \Cref{lemma:net_hierarchy_construction}), $\desc_k(u)$ is also $2^{k - 3}$-packing for $B(u, 2^{\ell + 2})$. 

    For \eqref{it:desc_covering}, since $N_k$ is $2^{k - 2}$-covering for $X$, there exists $v \in N_k$, such that $\dist(x, v) \leq 2^{k - 2}$.
    It remains to prove that $v \in \desc_k(u)$:
    Denote $v_k := v$;
    for $k + 1 \leq j \leq \ell - 1$, denote $v_j \in N_j$ as the level-$j$ net point closest to $v_{j - 1}$.
    Then $\dist(v_j, v_{j - 1}) \leq 2^{j - 2}$, and thus $v_{j - 1} \in \child(v_j)$.
    Note that 
    \begin{align*}
        \dist(v_{\ell - 1}, u)
        \leq \dist(u, x) + \dist(x, v) + \sum_{j = k + 1}^{\ell - 1} \dist(v_{j - 1}, v_j)
        \leq 2^{\ell} + 2^{k - 2} + \sum_{j = k + 1}^{\ell - 1} 2^{j -2}
        \leq 2^\ell + 2^{\ell - 2}
        < 2^{\ell + 1}.
    \end{align*}
    Therefore, $v_{\ell - 1} \in \child(u)$, and consequently, $v \in \desc_k(u)$.

    For $k=\ell$, the asserted bounds follow directly from $\desc_\ell(u)=\{u\}$.
\end{proof}

\lemmanumchildrenparent*

\begin{proof}
    By \eqref{it:desc_packing} of \Cref{lemma:descendants_packing_covering}, $\child(u)$ is $2^{\ell - 4}$-packing for $B(u, 2^{\ell + 2})$.
    Hence by \Cref{lemma:packing}, $|\child(u)| \leq 2^{O(\ddim)}$.

    Assume $u$ is a descendant of $u'$, where $u' \in N_k$ for $k > \ell$.
    Then by \Cref{lemma:descendants_packing_covering} \eqref{it:desc_packing}, $\dist(u, u') \leq 2^{\lev(u') + 2} = 2^{k + 2}$.
    Since $N_{k}$ is $2^{k - 3}$-packing, the number of such $u' \in N_k$, by \Cref{lemma:packing}, is at most $2^{O(\ddim)}$.
\end{proof}

\lemmanethierarchyconstruction*

\begin{proof}
    For the purpose of construction, each net point $u \in N_\ell$ maintains, apart from its children $\child(u)$, two auxiliary sets:
    (i) the \emph{group} $G(u) := \{x \in X \colon u = \argmin_{u' \in N_\ell} \dist(x, u')\}$ (i.e., $G(u)$ contains all points $x \in X$ such that $u$ is the point in $N_\ell$ closest to $x$, breaking ties arbitrarily), and 
    (ii) the \emph{neighbors} $H(u) := B(u, 2^{\ell + 2}) \cap N_\ell$, i.e., all net points in $N_\ell$ within distance $2^{\ell + 2}$ from $u$.

    The data structure is constructed inductively in a top-down manner. 
    At level $L$, $N_L$ contains one arbitrary point $u \in X$.
    Its group and neighbors, by definition, are $G(u) = X$ and $H(u) = \{u\}$, respectively.
Assume $N_{\ell + 1}$ together with $\{G(u)\}_{u \in N_{\ell + 1}}$ and $\{H(u)\}_{u \in N_{\ell + 1}}$ is already constructed, we construct 
    $N_\ell$, $\{\child(u)\}_{u \in N_{\ell + 1}}$, 
    $\{G(v)\}_{v \in N_{\ell}}$ and $\{H(v)\}_{v \in N_{\ell}}$ by the following steps.

    \paragraph{Step 1: Construct nets within each group.}
    For each $u \in N_{\ell + 1}$, construct a $2^{\ell - 3}$-net $N_\ell'(u)$ of $G(u)$, such that $u \in N_\ell'(u)$.
    Denote $N_\ell'$ as the union of all these nets, i.e., 
    \[N_\ell' := \bigcup_{u \in N_{\ell + 1}} N_\ell'(u).
    \]

    In particular, each $N_\ell'(u) \subseteq G(u)$ can be constructed greedily in time $|N_\ell'(u)| \cdot |G(u)|$.
    Hence, the total complexity of constructing $N_\ell'$ is 
    \begin{align*}
        \sum_{u \in N_{\ell + 1}} |N_\ell'(u)| \cdot |G(u)|
        \leq \sum_{u \in N_{\ell + 1}} 2^{O(\ddim)} \cdot |G(u)|
        = 2^{O(\ddim)} n.
    \end{align*}
    In the first inequality, we use that $|N_\ell'(u)| \leq 2^{O(\ddim)}$, this is because $N_\ell'(u)$ is $2^{\ell - 3}$-packing for $G(u)$ and $\diam(G(u)) \leq 2^{\ell - 1} + 2^{\ell - 1} = 2^\ell$.

    \paragraph{Step 2: Construct $N_\ell$ by eliminating non-packing points in $N_\ell'$.}
    $N_\ell'$ is already $2^{\ell - 3}$-covering for $X$ by construction, but may not be $2^{\ell - 3}$-packing, as two net points from different groups can be very close to each other.
    More concretely, if $v, v' \in N_\ell'$ satisfies $\dist(v, v') < 2^{\ell - 3}$, and assume $v \in N_\ell'(u)$ and $v' \in N_\ell'(u')$, we then have 
    \begin{align*}
        \dist(u, u') \leq \dist(v, v') + \dist(v, u) + \dist(v', u')
        \leq 2^{\ell - 3} + 2^{\ell - 1} + 2^{\ell - 1}
        < 2^{\ell + 3}.
    \end{align*}
    Consequently, $u' \in H(u)$ and $v' \in \bigcup_{u' \in H(u)} N_\ell'(u')$.
    We thus have the following procedure to eliminate the close points in $N_\ell'$:

    Let $\sigma$ be an arbitrary ordering over $N_\ell'$, such that for $u \in N_{\ell + 1} \subseteq N_\ell'$, $\sigma(u) \leq |N_{\ell + 1}|$.
    Initially, let $N_\ell = N_\ell'$.
    Each time we pick one unexplored point $v \in N_\ell$ with the smallest order $\sigma(v)$, and find $u \in N_{\ell + 1}$ such that $v \in N_\ell'(u)$.
    Then, for each point $v' \in \bigcup_{u' \in H(u)} N_\ell'(u')$, 
    if $\dist(v, v') < 2^{\ell - 3}$ and $\sigma(v) < \sigma(v')$, delete $v'$ from $N_\ell$.
    The remaining set after all points are explored is defined as the level-$\ell$ net $N_\ell$.

    Note that $|H(u)| \leq 2^{O(\ddim)}$ and $|N_\ell'(u')| \leq 2^{O(\ddim)}$ for each $u' \in H(u)$, hence the time complexity of exploring each $v \in N_\ell'$ is $2^{O(\ddim)}$.
    Therefore, the time complexity of the above elimination procedure is $2^{O(\ddim)} |N_\ell'| \leq 2^{O(\ddim)} n$.

    We claim that $N_\ell$ is $2^{\ell - 3}$-packing and $2^{\ell - 2}$-covering for $X$.
    Indeed, the packing property follows directly from the way $N_\ell$ is constructed.
    Regarding the covering property, let us fix a point $x \in X$ and prove that $\dist(x, N_\ell) \leq 2^{\ell - 2}$. 
    Let $v \in N_\ell'$ be the closest point to $x$ in $N_\ell'$.
    Recall that $N_\ell'$ is $2^{\ell - 3}$-covering, so $\dist(x, v) \leq 2^{\ell - 3}$.
    If $v \in N_{\ell}$, then we already have $\dist(x, N_\ell) \leq \dist(x, v) \leq 2^{\ell - 3}$.
    If $v \notin N_\ell$, then it must be deleted during the aforementioned procedure.
    Assume $v$ is deleted while the procedure is exploring $v'$.
    We then have $\dist(v, v') < 2^{\ell- 3}$, and $v'$ must be contained in $N_\ell$.
    Therefore, $\dist(x, N_\ell) \leq \dist(x, v') \leq \dist(x, v) + \dist(v, v') \leq 2^{\ell - 2}$.
    In conclusion, $N_\ell$ is $2^{\ell - 3}$-packing and $2^{\ell - 2}$-covering for $X$.

    \paragraph{Step 3: Maintain the children $\child(u)$ for each $u \in N_{\ell + 1}$.}
    Recall that for each $u \in N_{\ell + 1}$, its children are $\child(u) = \{v \in N_\ell \colon \dist(u, v) \leq 2^{\ell + 2}\}$.
    Hence, it suffices to find for each $v \in N_\ell$ all $u' \in N_{\ell + 1}$ with $\dist(v, u') \leq 2^{\ell + 2}$, and add $v$ to $\child(u')$.
    Let $u \in N_{\ell + 1}$ be such that $v \in N_\ell'(u) \subseteq G(u)$, we then have 
    \begin{align*}
        \dist(u, u') \leq \dist(u, v) + \dist(u', v)
        \leq 2^{\ell - 1} + 2^{\ell + 2}
        < 2^{\ell + 3},
    \end{align*}
    implying $u' \in H(u)$.
    We thus have the following procedure:

    For each $v \in N_\ell$, find $u \in N_{\ell + 1}$ such that $v \in G(u)$.
    Enumerate all $u' \in H(u)$; if $\dist(u', v) \leq 2^{\ell + 2}$, then add $v$ to $\child(u')$.

    Since $|H(u)| \leq 2^{O(\ddim)}$, the total time complexity is $2^{O(\ddim)} |N_\ell| \leq 2^{O(\ddim)} n$.

    \paragraph{Step 4: Maintain the grouping $\{G(v)\}_{v \in N_\ell}$.}
    It suffices to find for every $x \in X$ its closest point $v$ in $N_\ell$, and add $x$ to $G(v)$.
    To this end, let $u \in N_{\ell + 1}$ be such that $x \in G(u)$.
    Then \begin{align*}
        \dist(u, v) \leq \dist(x, u) + \dist(x, v)
        \leq 2^{\ell - 1} + 2^{\ell - 2} 
        < 2^{\ell + 2}.
    \end{align*}
    This implies that $v \in \child(u)$.
    We thus have the following procedure:

    For each $x \in X$, find $u \in N_{\ell + 1}$ such that $x \in G(u)$.
    Find the point in $\child(u)$ closest to $x$, denoted by $v$.
    Add $x$ to $G(v)$.

    By \Cref{lemma:num_children_parent} \eqref{it:child_num}, $|\child(u)| \leq 2^{O(\ddim)}$, so the time complexity is $2^{O(\ddim)} n$.

    \paragraph{Step 5: Maintain the neighbors $\{H(v)\}_{v \in N_\ell}$.}
    Recall that $H(v) = \{v' \in N_\ell \colon \dist(v, v') \leq 2^{\ell + 2}\}$ for each $v \in N_\ell$.
    Consider $v \in N_\ell$ and $v' \in H(v)$, and denote $u, u' \in N_{\ell + 1}$ to be such that $v \in G(u)$ and $v' \in G(u')$, respectively.
    Then \begin{align*}
        \dist(u, u') \leq \dist(v, v') + \dist(u, v) + \dist(u', v') 
        \leq 2^{\ell + 2} + 2^{\ell - 1} + 2^{\ell - 1}
        < 2^{\ell + 3}.
    \end{align*}
    Hence, $u' \in H(u)$ and $v' \in \bigcup_{u' \in H(u)} N_\ell'(u')$.
    We thus have the following procedure:

    For each $v \in N_\ell$, find $u \in N_{\ell + 1}$ such that $v \in G(u)$.
    Enumerate $v' \in \bigcup_{u' \in H(u)} N_\ell'(u')$; if $v'\in N_\ell$ and $\dist(v, v') \leq 2^{\ell + 2}$, then add $v'$ to $H(v)$.
    A retained-point flag, set during Step 2, tests membership in $N_\ell$ in constant time.

    Since $|H(u)| \leq 2^{O(\ddim)}$ and $|N_\ell'(u')| \leq 2^{O(\ddim)}$ for each $u' \in H(u)$, the time complexity is $2^{O(\ddim)} |N_\ell|
    \leq 2^{O(\ddim)} n$.

    \paragraph{Wrapping up.}
    The time complexity of constructing each level (including $N_\ell$, $\{\child(u)\}_{u \in N_{\ell + 1}}$, 
    $\{G(v)\}_{v \in N_{\ell}}$ and $\{H(v)\}_{v \in N_{\ell}}$) is $2^{O(\ddim)} n$.
    Since we have $O(\log \Delta)$ levels, the total time complexity of constructing the data structure is 
    $2^{O(\ddim)} n \log \Delta$.
\end{proof}

\subsection{Bounding the Aspect Ratio}
\label{sec:bounded_spread}

\begin{restatable}{lemma}{lemmaboundedspread}\label{lemma:bounded_spread}
    Given an $n$-point metric space $(X, \dist)$ with doubling dimension $\ddim$ and $\epsilon \in (0, 1)$, one can compute in $2^{O(\ddim)} n \log(n/\epsilon)$ time a subset $X' \subseteq X$, such that
    \begin{enumerate}[(1)]
        \item The minimum interpoint distance of $X'$ is greater than $1$ and the diameter of $X'$ is $O(n^{10}/\epsilon)$.
        \item Given any $(1 + \epsilon)$-approximate MST of $X'$, one can reconstruct in $O(n)$ time a $(1 + 2 \epsilon)$-approximate MST of $X$.
    \end{enumerate}
\end{restatable}

\begin{proof}

    Denote the diameter of $X$ by $\Delta$.\footnote{For implementation, we only need a constant-approximation of the diameter, which can be computed in $O(n)$ time.
    We assume the diameter is known in advance for ease of presentation.}
    The subset $X'$ is constructed as follows:

    Let $\ell_0$ be the smallest integer $\ell$ with $2^{\ell} \geq n^{10}/\epsilon$.
    We run the algorithm in \Cref{sec:appendix:net-hierarchy} to construct the highest $\ell_0$ levels of the net hierarchy, i.e., $N_L, N_{L - 1}, \dots, N_{L - \ell_0 + 1}$, and maintain the grouping $\{G(u)\}_{u \in N_{L - \ell_0 + 1}}$ at level $L - \ell_0 + 1$.
    Define $X' := N_{L - \ell_0 + 1}$.
    We note that this construction does not require rescaling $X$ in advance.

    By the construction in the proof of \Cref{lemma:net_hierarchy_construction} computing each level takes $2^{O(\ddim)} n$ time.
    Since we have $\ell_0 = O(\log(n/\epsilon))$ levels, the time complexity of constructing $X'$ is $2^{O(\ddim)} n \log(n/\epsilon)$.

    Property (1) follows directly from the construction, as $\diam(X') = \Delta$ and the minimum interpoint distance of $X'$, by \Cref{lemma:net_hierarchy_construction}, is at least $\Delta / 2^{\ell_0} \geq \epsilon \Delta / n^{10}$.
    Rescaling the distances yields the desired bounds.

    For property (2), 
we first show that 
    \begin{align}
        \MST(X') \leq \MST(X) + \frac{4 \epsilon \Delta}{n^9}.
        \label{eqn:MST_change}
    \end{align}
    Indeed, let $T^*$ be the MST of $X$;
    we construct a spanning tree $T'$ of $X'$ as follows: 
    For every edge $e = (x, y) \in T^*$, let $u$ and $v$ be points in $X' = N_{L - \ell_0 + 1}$ such that $x \in G(u)$ and $y \in G(v)$; add the edge $e' = (u, v)$ to $T'$.
    Clearly $T'$ is connected in $X'$, as each path in $T^*$ has a corresponding path in $T'$.
    Thus, after removing cycles and duplicated edges, $T'$ becomes a spanning tree of $X'$.
On the other hand, recall the definition of grouping that $u$ (resp. $v$) is the closest point to $x$ (resp. $y$) in $N_{L - \ell_0 + 1}$, so $\dist(x, u) \leq 2 \Delta / 2^{\ell_0} \leq 2 \epsilon \Delta / n^{10}$ (resp. $\dist(y, v) \leq 2 \epsilon \Delta / n^{10}$).
    Therefore, for each edge $e = (x, y) \in T^*$, 
    \begin{align*}
        w(e') - w(e) = \dist(u, v) - \dist(x, y) \leq \dist(x, u) + \dist(y, v)
        \leq \frac{4 \epsilon \Delta}{n^{10}}.
    \end{align*}
We hence have 
    \begin{align*}
        w(T') - w(T^*)
        \leq \sum_{e \in T^*} \frac{4 \epsilon \Delta}{n^{10}}
        \leq \frac{4 \epsilon \Delta}{n^9},
    \end{align*}
    completing the proof of \eqref{eqn:MST_change}.

    Given an arbitrary $(1 + \epsilon)$-approximate MST $T'$ of $X'$, we can construct a spanning tree $T$ of $X$ as follows:
    For each $x \in X \setminus X'$, find $u \in X' = N_{L - \ell_0 + 1}$ such that $x \in G(u)$ and connect the edge $(x, u)$.
    Since each $x$ already stores its closest point $u$ in $N_{L - \ell_0 + 1}$, 
    $T$ can be constructed in $O(n)$ time.
    Moreover,
    \begin{align*}
        w(T) &= w(T') + \sum_{x \in X \setminus X'} \dist(x, X') \\
        &\leq (1 + \epsilon) \MST(X') + n \cdot \frac{2 \epsilon \Delta}{n^{10}}  \\
        &\leq (1 + \epsilon) (\MST(X) + \frac{4 \epsilon \Delta}{n^9}) + \frac{2\epsilon \Delta}{n^9} &&\text{by \eqref{eqn:MST_change}} \\
        &\leq (1 + \epsilon) \MST(X) + \frac{10 \epsilon \Delta}{n^9} \\
        &\leq (1 + 2 \epsilon) \MST(X) &&\text{since } \MST(X) \geq \Delta \text{ and } n \geq 2.
    \end{align*}
    This concludes the proof of property (2).
\end{proof}

\subsection{Approximating the Doubling Dimension}

\begin{restatable}{lemma}{lemmaddimapx}\label{lemma:ddim_apx}
    Let $n \geq 2$ be an integer.
    There is an algorithm that, given as input an $n$-point metric space $(X, \dist)$ with minimum interpoint distance greater than $1$ and diameter $\Delta$, computes in $2^{O(\ddim(X))} n \log \Delta$ time a number $t > 0$, such that $c_1 \ddim(X) \leq t \leq c_2 \ddim(X)$, where $c_1$ and $c_2$ are absolute constants.
\end{restatable}

\begin{proof}
    Our approximation algorithm is rather simple:
    We first construct the net hierarchy $\nets = \{N_0, N_1, \dots, N_L\}$ over $X$ by \Cref{lemma:net_hierarchy_construction}. 
    (Note that the construction of $\nets$ in the proof of \Cref{lemma:net_hierarchy_construction} does not require knowing $\ddim(X)$ in advance.)
    We then output 
    \begin{align*}
        t = \max_{u \in \bigcup_{\ell = 1}^L N_\ell} \log |\child(u)|
    \end{align*}
    as an approximation of $\ddim$.\footnote{For implementation, we do not need to compute the $\log(\cdot)$ precisely.
    It suffices to compute the smallest integer $K$, such that $2^K \geq |C(u)|$.
    This only introduces an extra $O(\ddim)$ overhead.}

    The running time follows directly from \Cref{lemma:net_hierarchy_construction}.
    It then remains to prove that $t = \Theta(\ddim)$.
    For the upper bound, by \Cref{lemma:num_children_parent} \eqref{it:child_num}, we have that for each net point $u$, $|\child(u)| \leq 2^{O(\ddim)}$, so $t = O(\ddim)$.

    For the lower bound, fix an arbitrary point $x \in X$ and a radius $r \in [1, \Delta]$, we show that the ball $B(x, r)$ can be covered by at most $2^{O(t)}$ balls of half the radius, implying the desired bound $\ddim = O(t)$.
    In fact, let $\ell$ be such that $2^{\ell - 1} < r \leq 2^\ell$;
    let $u \in N_{\ell + 2}$ be the closest net point to $x$ in $N_{\ell + 2}$.
    By \eqref{eqn:N_l_properties}, $\dist(x, u) \leq 2^{\ell}$, so 
    \[B(x, r) \subseteq B(u, 2^{\ell + 1}).\]
    Moreover, by \Cref{lemma:descendants_packing_covering} \eqref{it:desc_covering}, 
    \[
    B(u, 2^{\ell + 1}) \subseteq \bigcup_{v \in \desc_{\ell}(u)} B(v, 2^{\ell - 2})
    \subseteq \bigcup_{v \in \desc_{\ell}(u)} B(v, r/2).
    \]
    Therefore, $B(x, r)$ can be covered by $|\desc_\ell(u)| \leq 2^{2t}$ balls of half the radius.
    Hence, $\ddim \leq 2 t$ as $x$ and $r$ are arbitrarily chosen.
\end{proof}

\section{Omitted Proofs from \Cref{sec:WSPC_main}}
\label{sec:WSPC}

\lemmaWSPCMST*

\begin{proof}
    For (1), assume by contradiction that there exists a pair $(U, V) \in \Phi$ and two different edges $(u, v), (u', v') \in U \times V$ in the MST $T^*$.
    One can then replace $(u', v')$ by two edges $(u, u')$ and $(v, v')$.
    Such replacement does not change the connectivity, and 
    \begin{align*}
        \dist(u, u') + \dist(v, v') < \frac{1}{t} \dist(U, V) + \frac{1}{t} \dist(U, V) 
        \leq \frac{2}{t} \dist(u', v')
        \leq \dist(u', v'),
    \end{align*}
    contradicting the minimality of $T^*$.

    For (2), it suffices to show that the resulting graph is connected.
    Denote the MST edges by $\{e_1^*, e_2^*, \dots, e_{n - 1}^*\}$, where $w(e_1^*) \leq w(e_2^*) \leq \dots \leq w(e_{n - 1}^*)$,
    and suppose $(U_i, V_i) \in \Phi$ is the selected pair with $e_i^* \in U_i \times V_i$.
    Let $e_i$ be an arbitrary edge in $U_i \times V_i$.
    For $1 \leq i \leq n - 1$, let $G_i^*$ (resp. $G_i$) be the graph formed by $\{e_1^*, e_2^*, \dots, e_{i}^*\}$ (resp. $\{e_1, e_2, \dots, e_{i}\}$).
    Let 
    $\mathcal{C}_i^* \subseteq 2^X$ (resp. $\mathcal{C}_i \subseteq 2^X$) be the collection of connected components in $G_i^*$ (resp. $G_i$).
    We prove by induction that $\mathcal{C}_i^* = \mathcal{C}_i$.
    (2) then follows by setting $i = n - 1$.

    The equation holds for $i = 0$, as initially both $\mathcal{C}_0^*$ and $\mathcal{C}_0$ are collections of singletons.
    Assume $\mathcal{C}_{i - 1}^* = \mathcal{C}_{i - 1}$.
    We first claim that $U_i$ is connected in $G_{i - 1}^*$ (and thus in $G_{i - 1}$).
    In fact, if there are two points $u_1, u_2 \in U_i$ which are not connected in $G_{i - 1}^*$, then there exists an edge $f$ on the MST path between $u_1$ and $u_2$ with $w(f) \geq w(e_i^*)$.
    On the other hand, by the separation property, $\dist(u_1, u_2) \leq \diam(U_i) < \dist(U_i, V_i) \leq w(e_i^*)$.
    Replacing $f$ by $(u_1, u_2)$ yields a cheaper solution, which contradicts the minimality of $T^*$.
    Similarly, $V_i$ is also connected in $G_{i - 1}^*$ (and thus in $G_{i - 1}$).

    Therefore, $e_i^*$ merges the connected components in $G_{i - 1}^*$ which contain $U_i$ and $V_i$, and $e_i$ merges the connected components in $G_{i - 1}$ which contain $U_i$ and $V_i$.
    Since $\mathcal{C}_{i - 1}^* = \mathcal{C}_{i - 1}$, we have $\mathcal{C}_i^* = \mathcal{C}_i$, concluding the proof.
\end{proof}

\lemmaWSPCdoubling*

\begin{proof}

To construct $\Psi$, we call \Cref{alg:WSPC} with $\WSPC(r, r)$ and parameter $t$, where $r$ is the ``root'' of the net hierarchy (i.e., the only net point in $N_L$).
Properties \eqref{it:bounded_lev_diff} and \eqref{it:separation} follow directly from \Cref{alg:WSPC}.

    \paragraph{Property \eqref{it:Psi_covering}.}
    Let us fix a pair of points $(x, y) \in \binom{X}{2}$, and show that there exists $(u, v) \in \Psi$ such that $(x, y) \in \wspcball{u} \times \wspcball{v}$.
    To this end, assume \Cref{alg:WSPC} processes an unmarked state $\WSPC(u, v)$, where $x \in B(u, 2^{\lev(u)})$ and $y \in B(v, 2^{\lev(v)})$;
    we show that the call $\WSPC(u, v)$ either returns $(u, v)$, or makes a call to some $\WSPC(u', v)$, such that $x \in B(u', 2^{\lev(u')})$ and $y \in B(v, 2^{\lev(v)})$.
    Indeed, it suffices to prove 
    \begin{align*}
        B(u, 2^{\lev(u)}) \subseteq \bigcup_{u' \in \child(u)} B(u', 2^{\lev(u')}).
    \end{align*}
    This follows directly from \eqref{it:desc_covering} of \Cref{lemma:descendants_packing_covering}.

    \paragraph{Property \eqref{it:bounded_deg}.}
    We will show a slightly stronger result: 
    For each net point $u$, the number of net points $v$ such that 
    \Cref{alg:WSPC} makes a call to $\WSPC(u, v)$ is at most $t^{O(\ddim)}$. 
    Note that this bound concerns distinct possible partners, rather than the number of attempted calls.

    If $u = r$ then the claim holds, as \Cref{alg:WSPC} calls $(r, v)$ only if $v = r$ or $v$ is a child of $r$.
    By \Cref{lemma:num_children_parent} \eqref{it:child_num}, $|\child(r)| \leq 2^{O(\ddim)}$.
    Assume $u \neq r$, and the procedure calls $(u, v)$ for some net point $v$.
    We consider the last call that invokes $(u, v)$.
    We have two cases:

    If the last call is $(u', v)$, where $u$ is the child of $u'$,
    then by property \eqref{it:bounded_lev_diff} we have $\lev(v) \leq \lev(u) + 1$ and $\lev(v) \geq \lev(u') - 1 = \lev(u)$.
    Moreover, since the call to $(u', v)$ does not return $\{(u', v)\}$, we have
    $\dist(u', v) \leq (2t + 2) \max\{2^{\lev(u')}, 2^{\lev(v)}\} = (2 t + 2) 2^{\lev(u')}$.
    By triangle inequality, 
    \begin{align*}
        \dist(u, v) \leq \dist(u, u') + \dist(u', v)
        \leq 2^{\lev(u') + 1} + (2 t + 2) 2^{\lev(u')}
        = (2t + 4) 2^{\lev(u')}
        = (4t + 8) 2^{\lev(u)}.
    \end{align*}
    As $\lev(u) \leq \lev(v) \leq \lev(u) + 1$, all possible such $v$ are in 
    \[
    B(u, (4t + 8) 2^{\lev(u)}) \cap (N_{\lev(u)} \cup N_{\lev(u) + 1}).
    \]
    The number of such $v$ is at most $(64 t + 128)^{O(\ddim)}$.

    If the last call is $(u, v')$, where $v$ is a child of $v'$, then 
    by property \eqref{it:bounded_lev_diff} we have $\lev(v') \leq \lev(u) + 1$ and $\lev(v') = \lev(v) + 1 \geq \lev(u)$.
    Since the call to $(u, v')$ does not return $\{(u, v')\}$, we have
    $\dist(u, v') \leq (2t + 2) \max\{2^{\lev(u)}, 2^{\lev(v')}\} \leq (4t + 4) 2^{\lev(u)}$.
    So all such $v'$ are in 
    \[
    B(u, (4t + 4) 2^{\lev(u)}) \cap (N_{\lev(u)} \cup N_{\lev(u) + 1}).
    \]
    The number of such $v'$ is at most $(64 t + 64)^{O(\ddim)}$.
    As $v$ is one child of $v'$, by \Cref{lemma:num_children_parent} \eqref{it:child_num}, the number of $v$ is at most $(64 t + 64)^{O(\ddim)} \cdot 2^{O(\ddim)}$.

    In conclusion, the number of net points $v$ such that 
    the procedure makes a call to $(u, v)$ is at most $t^{O(\ddim)}$.
    Property \eqref{it:bounded_deg} thus follows.

    Moreover, the time complexity of constructing $\Psi$ also follows by noting that 
    the algorithm only makes a call on $(u,v)$ if it is not yet marked. 
    Therefore the Algorithm makes a call for every pair at most once.

    \paragraph{$\Phi$ and its construction.}
    To prove the ``In particular'' part, let us check that $\Phi$ satisfies \Cref{def:WSPC}.
    The covering property (\Cref{def:WSPC} \eqref{it:WSPC_covering}) follows directly from property \eqref{it:Psi_covering}.
    Regarding the separation (\Cref{def:WSPC} \eqref{it:WSPC_separation}),
    note that for a pair $(u, v) \in \Psi$, if $\lev(u) = \lev(v) = 0$, then both $\wspcball{u}$ and $\wspcball{v}$ are singletons, thus have diameter $0$.
    Otherwise, by property \eqref{it:separation},
    $\dist(u, v) > (2t + 2) \max\{2^{\lev(u)}, 2^{\lev(v)}\}$.
    Hence,
    \begin{align*}
        \dist(\wspcball{u}, \wspcball{v}) \geq \dist(u, v) - 2^{\lev(u)} - 2^{\lev(v)}
        > 2t \cdot \max\{2^{\lev(u)}, 2^{\lev(v)}\}.
    \end{align*}
    So $\Phi$ is a $t$-WSPC of $X$.

    To construct $\Phi$ from $\Psi$, it suffices to construct $\wspcball{u}$ for every net point $u$ on the net hierarchy.
    For each $x \in X$ and $0 \leq \ell \leq L$, denote $P_\ell(x) := \{u \in N_\ell \colon \dist(u, x) \leq 2^\ell\}$.
    The problem then reduces to computing $P_\ell(x)$ for all $x \in X$ and $0 \leq \ell \leq L$.
    We can construct $P_\ell(x)$ inductively: The base case is $P_0(x) = \{x\}$.
    For $\ell \geq 1$, assume we already have $P_{\ell - 1}(x)$.
    Note that for every $u \in P_\ell(x)$ and $v \in P_{\ell - 1}(x)$, $\dist(u, v) \leq \dist(u, x) + \dist(v, x) \leq 2^\ell + 2^{\ell - 1} < 2^{\ell + 1}$, so $v \in \child(u)$.
    Hence to construct $P_\ell(x)$, we pick an arbitrary point $v \in P_{\ell - 1}(x)$, construct the set $A = \{u \in N_\ell \colon v \in \child(u)\}$ using pointers to $v$, and compute $P_\ell(x) := A \cap B(x, 2^\ell)$.
    Since $|A| \leq 2^{O(\ddim)}$ by \Cref{lemma:num_children_parent} \eqref{it:ancestor_num}, the time complexity of constructing $P_\ell(x)$ is $2^{O(\ddim)}$, thus $P_\ell(x)$ for all $0 \leq \ell \leq L$ and $x \in X$ can be constructed in $2^{O(\ddim)} n \log\Delta$ time.
    Therefore, $\Phi$ can be constructed from $\Psi$ in time $2^{O(\ddim)} n \log\Delta + t^{O(\ddim)} n \log\Delta = t^{O(\ddim)} n \log\Delta$.
\end{proof}

\end{document}